 \documentclass[journal]{IEEEtran}

\IEEEoverridecommandlockouts                              
\usepackage{amsmath}
\usepackage{graphics}
\usepackage{graphicx}
\usepackage{setspace}
\usepackage{amssymb}
\usepackage{amsthm}
\usepackage[usenames,dvipsnames,svgnames,table]{xcolor}
\usepackage{tikz}
\usepackage{multirow}
\usepackage{cite}
\usepackage{subcaption}
\newcommand \red[1] {{\color{red}#1}}
\newcommand \edit[1] {{{\red{#1}}}}

\def \bi {\begin{itemize}\item}
\def \ei {\end{itemize}}
\def \be {\begin{equation}}
\def \ee {\end{equation}}
\def \ba {\begin{aligned}}
\def \ea {\end{aligned}}

\def \l {\bigg{(}}
\def \r {\bigg{)}}

\def\baselinestretch{0.98}

\newcommand{\LyX}{L\kern-.1667em\lower.25em\hbox{Y}\kern-.125emX\spacefactor1000}

\title{Is Non-Neutrality Profitable for the Stakeholders of the Internet Market?}

\author{Mohammad Hassan Lotfi, Saswati Sarkar, and George Kesidis\thanks{Parts of this work were presented in CISS'16~\cite{CISS}. Some of the  preliminary results and ideas of this work were presented as a poster in NetEcon'14 \cite{NetEcon}.}\thanks{This research was supported by NSF CNS (NeTS) collaborative grants 1526133 and 1525457.}
\thanks{ M. H. Lotfi is at Capital One.  Saswati Sarkar is with the Department of ESE at University of Pennsylvania, Philadelphia,
		PA, U.S.A.  George Kesidis is with the school of EECS of Pennsylvania State University, University Park, PA, U.S.A. Their email addresses are mohammadhassanlotfi@gmail.com, 
		swati@seas.upenn.edu, gik2@psu.edu, respectively.}}

\begin{document}
\maketitle\newtheorem{lemma}{Lemma}
\newtheorem{note}{Note}
\newtheorem{property}{Property}
\newtheorem{theorem}{Theorem}
\newtheorem{definition}{Definition}
\newtheorem{corollary}{Corollary}
\newtheorem{remark}{Remark}
\newtheorem{assumption}{Assumption}

\begin{abstract}
We consider an Internet market with one “big” monopolistic Content Provider (CP), one neutral and another non-neutral Internet Service Provider (ISP) and a
continuum of End-Users (EUs). The CP can differentiate between ISPs by controlling the quality of the content she offers on
each. EUs have different levels of innate preferences for ISPs. We formulate a
sequential game, and prove that if a Sub-game Perfect Nash Equilibrium (SPNE)
 exists, it would be one of the five possible strategies each
of which we explicitly characterize. We prove that 1) when EUs have sufficiently low innate preferences
for ISPs, a unique SPNE exists in which the neutral ISP is driven out of the market,  and 2) when these preferences are sufficiently high, there exists a unique SPNE with a non-neutral
outcome in which both ISPs are active. Numerical results reveal that the neutral ISP receives a lower
payoff and the non-neutral ISP receives a higher payoff (most of the time) in a non-neutral scenario.
However, we identify scenarios in which the non-neutral ISP loses payoff by adopting non-neutrality.
We also show that a non-neutral regime may yield a higher or a lower welfare for EUs in comparison to a neutral one -  the former happens
if the market power of the non-neutral ISP is small and the sensitivity of EUs (respectively, the CP) to the
quality is low (respectively, high).
\end{abstract}

\section{Introduction}

\subsection{Motivation} Net-neutrality on the Internet is the set of policies that prevents  discrimination by Internet Service Providers (ISPs) among different types of transmitted data \cite{progressive}. The net-neutrality debate has received impetus since
 parts of the Federal Communication Commission's (FCC) rules for net-neutrality has been struck down in court in January 2014 \cite{NetNeutrality}.  In February 2015, the FCC reclassified the Internet as a utility \cite{FCCutility}, providing grounds for securing even  stricter net-neutrality rules.
However, the situation is still fluid because both  ISPs and Content Providers (CPs) have incentives to adopt a non-neutral regime: 1) some CPs are willing to pay for a premium quality by which they can increase the usage, the satisfaction, or the number of their subscribers \cite{sponsoring_journal}, 2) ISPs can increase their profit by charging CPs for a premium quality.  In October 2015, the European parliament has rejected legal amendments for strict net-neutrality rules, and has allowed  sponsored data plans and Internet fast lanes for ``specialized services" \cite{EU_net_neutrality}.

Net-neutrality rules often have loopholes. For example, as part of a peering agreement towards resolving the traffic imbalance, Netflix has agreed to  pay Comcast for a faster access to Comcast's subscribers since February 2014 \cite{NYtimes2}.   However, after deploying the agreement, the average Netflix download speed improved significantly \cite{netflix_deal}. Note that
a contract for resolving {\em aggregate}  traffic imbalance at tier-1 ties Service Level Agreements (SLAs)
(particularly between an ``eyeball" ISP  and one serving a CP) in which the party receiving the net traffic imbalance gets paid is considered ``neutral" \cite{Kesidis13,Kesidis14}. Thus, although the Netflix-Comcast deal does not violate the net-neutrality rules, it has a non-neutral outcome of a side-payment between a residential ISP and a CP.


We  consider a market in which some of the ISPs are neutral and some are non-neutral, and model their interaction with each other and with a  CP  in the presence of \emph{asymmetric} competition between ISPs. This is because at initial stages of migration to a non-neutral regime, some ISPs would adopt a non-neutral regime before others. We consider CPs that can differentiate between ISPs by controlling the quality of the content they  offer on each one.  
 We consider the incentives of individual ISPs to adopt a non-neutral regime, while the focus of  most of previous works  is on the social welfare analysis of the market when all ISPs are neutral and/or all are non-neutral. We seek  intuitions with respect to a diverse set of  parameters (e.g. market powers of ISPs\footnote{Market power is the ability of a decision maker to raise the market price for a good or service.},  sensitivity of EUs and the CP to the quality of the content)
for the new equilibrium of the market, when the current equilibrium (neutral regime) is disrupted due to adoption of a  non-neutral regime by some ISPs.  Intuitions from our model and analysis can be used by the regulator in designing efficient rules for the Internet market.

\subsection{Model and Formulation}
We consider a market with two ISPs, one neutral and one non-neutral. This can represent two competing \emph{groups} of  ISPs. We also consider  a ``big" monopolistic CP  (eg, Google in the search space) with high market power that chooses her strategies to influence the equilibrium outcome of the market. 
We consider a continuum of End-Users (EUs) that decide on the ISP they want to subscribe to.  EUs have different  levels of \emph{innate preferences} for each ISP which arise because of their pre-existing relations,  initial set-up costs they incur  upon switch and goodwill of the ISPs. These innate preferences capture the degree by which EUs are locked in with a particular ISP, and determine the market powers of ISPs.

Both ISPs offer a free service for CPs up to a threshold on quality. The non-neutral ISP offers a premium quality, in which CPs can send their content with a higher throughput,  in lieu of a side-payment from the CP. This side-payment can be negative or positive, where a negative side-payment means a net payment from the non-neutral ISP to the CP, and may help the non-neutral ISP ensure that the monopolistic CP offers with a premium quality and exclusively for her EUs. The CP earns through advertisements, with the advertisement profit  increasing with the quality she offers to EUs.

We formulate a sequential game and seek its Sub-game Perfect Nash Equilibrium (SPNE). 

Note that the SPNE has a complex dependency on a wide range of parameters.
Thus, the structure, the existence and the uniqueness of the SPNE is not apriori clear. One can expect different SPNEs in which either  the CP offers her content (i)  only with a free (best effort) quality, or (ii)  with free quality on the neutral and with premium quality on the non-neutral ISP, or (iii)  with a premium quality only on the non-neutral ISP.
 Moreover, the ISPs can select different SPNE Internet access fees and side-payments   whose value directly affect the welfare of EUs. For example, the non-neutral ISP can select a low Internet access fee to increase the number of her EUs and generate most of her revenue through the side-payment she charges the CP. Competition would now force the neutral ISP to decrease her Internet access fee. Thus, the welfare of EUs would be high. Or,  the non-neutral ISP may select a low side-payment (possibly negative) to ensure that the CP offers with a premium quality, and generate her revenue by increasing Internet access fees, which enables the neutral ISP to increase her access fees. This reduces the welfare for the EUs. Note that certain SPNEs may be more desirable for the CP, as they correspond to more desirable splits of the EUs between ISPs. The CP may steer the system to these SPNEs   by controlling the quality of her content on each ISP appropriately.

\subsection{Analytical Results}

We show that if an SPNE exists, it would be one of the five possible strategies each of which we explicitly characterize. 
 We also show that an SPNE does not always exist.

We prove that  when  EUs have sufficiently low \emph{inertia} for ISPs, i.e. when the preferences are ``relatively" small and do not overrule major discrepancies on price and quality, the SPNE is unique. In this SPNE, the CP offers her content with premium quality on the non-neutral ISP while she does not offer her content on the neutral ISP, to push all EUs to the non-neutral ISP which provides a better quality. Thus, the neutral ISP would be driven out of the market. 

We also consider the case that EUs have sufficiently high inertia for at least one of the ISPs, and EUs cannot easily switch between ISPs.   We prove that there exists a unique SPNE with a non-neutral outcome; in it  both  ISPs receive a positive share of EUs, i.e., both are active, and the CP offers her content with free quality on the neutral ISP and with premium quality on the non-neutral ISP.

 We also consider a benchmark case in which both ISPs are neutral, and prove that there exists a unique SPNE, in which  the CP offers her content over both ISPs with free quality, and both ISPs are  active. The results in the benchmark case helps us  assess the extent of benefit of switching to non-neutrality for  different entities of the market.

\subsection{Numerical Results}
Numerical results help pinpoint which of the five possible SPNE strategies occurs when the inertias are between the two extreme cases, high and low inertias, in which the analytical results guarantee existence and uniqueness of the SPNE. We learn that if the inertia is on the lower end of the intermediate range, the game has an SPNE outcome in which both  ISPs are active, but the CP offers her content with premium quality and only on the non-neutral ISP. Also, if the inertia is  on the upper end of the intermediate region, then the game has no SPNE. Simulation over large sets of parameters also suggest that in all scenarios, the SPNE is unique if it were to exist.

Numerical results reveal that  the neutral ISP loses payoff in all SPNE outcomes in comparison to the benchamrk case.  
In addition,  for a wide range of parameters, the non-neutral ISP receives a better payoff under a non-neutral scenario. 
However, switching to a non-neutral regime is \emph{not} always profitable for ISPs. If EUs or the CP are not sensitive to the quality of the content delivered and the market power of the non-neutral ISP is small, then ISPs are better off staying neutral.

Results also reveal that the  welfare of EUs (EUW) in  non-neutral scenarios may either be higher or lower than in neutral scenarios. The former happens if (i) the market power of the non-neutral ISP is low, (ii) the sensitivity of the CP to the quality of content is high, or (iii) EUs are not very sensitive to the quality.
 In this case, the non-neutral ISP offer a cheaper Internet access fee leading to a higher EUW.

\subsection{Related Works}
This work falls in the category of economic models for a
non-neutral Internet \cite{survey}.
This line of work can be divided into two broad categories:  those that consider   (a) a non-neutral regime in which a non-neutral ISP blocks the content of the CPs that do not pay the side-payment \cite{economides,asu},   (b) a non-neutral ISP that provides quality differentiations for CPs and do not necessarily block a content \cite{ma2013public,maille2016content,Kramer,Kim_hotelling,cheng_queue,walrand2009,altman2011_monopoly,altman2013_competition,choi2015net,bourreau2015hotel}. We consider the second scenario in this work, since it is likely to emerge owing to FCC restrictions on content blocking.

These works can also be further  divided into two other categories: (i) those that  consider monopolistic ISPs \cite{maille2016content,Kramer,Kim_hotelling,cheng_queue,Katz,walrand2009,altman2011_monopoly}, and (ii) those that consider competition between ISPs \cite{ma2013public,economides,asu,altman2013_competition,choi2015net,bourreau2015hotel}. Our work belongs to the latter case.  

 While we consider the incentives of individual ISPs to adopt a non-neutral regime, the  previous works (except \cite{ma2013public}) focus on the social welfare analysis of the market when all ISPs are neutral and/or all are non-neutral. \cite{ma2013public}    considers  competition between a neutral (public option) ISP with non-neutral ISPs, and  argues that the existence of a public option ISP in a non-neutral scenario (when other ISPs are non-neutral) increases the customer surplus in comparison to when all ISPs are neutral. In contrast, we show that  the competition between the neutral and non-neutral ISPs would not always increase the customers welfare.  The reason for the differences between the results of our paper with those of \cite{ma2013public} lies in the differences in the models of the two papers. We show that different market powers of ISPs, and the sensitivity of EUs and CPs to the quality  of the content   are important factors in determining the welfare of EUs. These factors are absent in the model of \cite{ma2013public}.

In addition, in contrast to the previous works, we consider  competition between ISPs that have different market powers, i.e. an asymmetric competition. Also, in most of the previous works, CPs are passive in that they are only price-takers. However, we consider that the CP can influence the market equilibrium by appropriately choosing the quality of the content that she offers for EUs of each ISP, e.g.,  she can select a particular ISP and offer with a high quality on  this ISP, and stop offering her content on other ISPs. Thereby, the CP might be able to migrate EUs of other ISPs to the selected ISP.



\subsection{Organization of the Paper}
In Section~\ref{section:model}, we present the model. In Sections \ref{section:theory},  we find the SPNE strategies, and provide insights from the key analytical results. In Section    \ref{section:proofsBenchamrk}, we present the results for a benchmark case, in which both ISPs are neutral. 
In Section~\ref{section:numericalresults}, we present numerical results. In Section~\ref{section:implicationAssum}, we comment on the assumptions  and their generalizations and present directions for future research. Theorems have been proven in the appendices. 

\section{Model and Formulation}\label{section:model}
We consider two Internet Service Providers (ISPs), a Content provider (CP), and  a continuum of End Users (EUs).
\subsection{Internet Service Providers (ISPs)}
\label{section:ISPmodel}
One ISP is neutral (ISP N), and the other is non-neutral (ISP NoN); NoN can offer a premium quality for CPs in lieu of a side-payment.  The strategies of the  ISPs N and NoN are to determine the marginal Internet access fees for the EUs, i.e. $p_N$ and $p_{NoN}$, respectively. We show that most of  the results will depend on the difference, $\Delta p:=p_{NoN}-p_N$.

The CP will pay a premium quality fee, i.e. the side-payment, to  ISP NoN if she chooses to offer a quality ($q$) higher than the free quality threshold ($\tilde{q}_f$); she can offer with up to the quality $\tilde{q}_f$ for free on both ISPs.
  ISP NoN also determines $\tilde{p}$, i.e. the marginal side-payment. Note that $\tilde{p}$ can be positive or negative, in which a negative value implies a reverse flow of money from  NoN to the CP.    The side-payment\footnote{The side-payment is discontinuous in that it is $0$ for $q \leq \tilde{q}_f$, and exceeds $\tilde{p}\tilde{q}_f $ for $q > \tilde{q}_f$. As long as $\tilde{p}$ is chosen non-zero, $|\tilde{p}\tilde{q}_f| $ is bounded away from $0.$ This discontinuity captures the additional network-management costs NoN incurs for service differentiation, which is what providing premium service to the CP entails. This discontinuity, however, precludes the application of the standard game theoretical results for existence and uniqueness of the SPNE.} is:  
$$
\text{Side-payment} =
\left\{
\begin{array}{ll}
\tilde{p}q,  &  \mbox{if }   q> \tilde{q}_f,  \\
0, & \mbox{otherwise. }
\end{array}
\right.
$$

 Both ISPs generate their profits from EUs; additionally, ISP NoN  earns from the CP (if $\tilde{p}>0$ and the CP is willing to pay for a premium quality).
Their payoffs are:
\begin{equation}\label{equ:payoffISPsGeneral_new}
	\begin{aligned}
		\pi_N(p_N)&=(p_N-c)n_N,\\
		\pi_{NoN}(p_{NoN},\tilde{p})&=(p_{NoN}-c)n_{NoN}+z\tilde{p}q_{NoN},
	\end{aligned}
\end{equation}
where $n_N$ and $n_{NoN}$ are the fractions of EUs that have access to Internet via  ISPs, N and NoN, respectively, and $c$ is the service cost that an ISP incurs for each subscribing EU.\footnote{Different EUs may consume different amounts of data, and therefore pay different amounts to the ISPs and incur different costs for them. The consumptions would be random, apriori unknown. If these are uniformly distributed,  then $p_N, p_{NoN}$ (respectively, $c$)  may be considered, appropriately normalized, expected payments (respectively, costs) for EUs. \eqref{equ:payoffISPsGeneral_new} still constitute the expected overall payoffs. Also, the ISPs' marginal profits per EU, $p_N-c, p_{NoN}-c$, as in \eqref{equ:payoffISPsGeneral_new}, do not depend on quality, as the quality here is associated with only one CP. The EUs may use the Internet for various other purposes, including accessing other contents, which the same ISPs may deliver at other qualities. But, if the CP in question is popular enough, then the quality of her content would influence the EUs' decision to choose the ISP, which we consider in Section~\ref{EUjoin}. That said, the ISPs choose the marginal access fees strategically, so the equilibrium choices are coupled, that is, only certain combination of marginal access fees, qualities, etc. constitute the equilibria (Section~\ref{section:discussion}). Thus, the choices are only implicitly correlated. } In addition, $q_{NoN}$ is the quality of the content on  ISP NoN. The parameter $z=1$ when CP offers her content with premium quality, and $z=0$ otherwise. 
In Section~\ref{future} we outline the changes when ISP NoN incurs an additional cost for serving CP's content at a premium quality. 

\subsection{The Content Provider (CP)}
The CP can potentially offer different quality levels on different ISPs. Her strategy constitutes  choosing a quality of $q_N\in\{0,\tilde{q}_f\}$ on  ISP N, and a quality of $q_{NoN}\in\{0,\tilde{q}_f,\tilde{q}_p\}$ on  ISP NoN, with $\Delta q:=q_{NoN}-q_N$\footnote{All results  hold even when the CP selects qualities from continuous sets, i.e. $q_N\in [0,\tilde{q}_f]$ and $q_{NoN}\in[0,\tilde{q}_p]$. 
  Refer to the Appendix~\ref{section:general} for proofs.}. Here, $\tilde{q}_p$ is the premium quality, which serves the content with a higher throughput, providing for additional features or higher resolution, $\tilde{q}_f$ has been defined in Section~\ref{section:ISPmodel}.   The CP's advertising profit is proportional to the number of EUs and the content quality she delivers to these, $q_N, q_{NoN}$ \footnote{  Increasing the quality of the content improves the quality of delivery (throughput, resolution, enjoyment), and therefore the effectiveness, of the advertisements, eg, when they are video or sound, as provided by YouTube and Spotify. This increases advertising revenue for the CP.} (refer to the first two terms of \eqref{equ:payoffCP_new}).
She also pays (or receives if $\tilde{p}<0$) a side-payment to ISP NoN based on the marginal side-payment, $\tilde{p}$,  determined by NoN. Thus, the profit of the CP is,
\be\label{equ:payoffCP_new}
\small
\pi_{CP}(q_N,q_{NoN},z)=n_{N} \kappa_{ad} q_N+n_{NoN}\kappa_{ad} q_{NoN}-z \tilde{p} q_{NoN},
\ee
\normalsize
where $\kappa_{ad}$ represents the CP's sensitivity to, or marginal benefit with, the quality of the advertisement, $\kappa_{ad}=\kappa_{ad,rev}-\kappa_{ad,cost}$, 
and $z$ is as defined in Section~\ref{section:ISPmodel}.  

It may appear from \eqref{equ:payoffCP_new} that the CP would lose nothing by choosing at least a free quality on both ISPs. However, this is not the case.  As we explain later, $n_N$ and $n_{NoN}$ are dependent on $q_N$ and $q_{NoN}$, and increasing one of them (e.g. $n_N$), decreases the other one (e.g. $n_{NoN}$).  Therefore, the CP may stop offering her content on ISP N to increase the number of EUs on  ISP NoN on which they can receive a better quality. This may lead to higher advertisement revenues. Thus, the CP can control the number of subscribers with each ISP by controlling the quality of the content on each ISP appropriately.

\subsection{End-Users (EUs)}
\label{EUjoin}
The strategy of an EU is to choose one of the ISPs to buy Internet access from. We assume that  ISP N and NoN  are respectively located at 0 and 1, and EUs are distributed uniformly along the unit interval $[0,1]$. Here distance need not be physical distance, but represents  the preference of EUs: the closer an EU is to an ISP, the more this EU prefers this ISP.

The EU located at $x\in[0,1]$ incurs a \emph{transport cost} of $t_N x$  (respectively, $t_{NoN}(1-x)$) when joining  ISP N (respectively,  ISP NoN), where $t_N$ (respectively, $t_{NoN}$) is the marginal transport cost for ISP N (respectively, NoN). The transport cost represents the reluctancy of an EU to opt for an ISP owing possibly to 1) pre-existing contracts with the competitor, e.g. when ISPs bundle Internet access with other services like cable, phone and provide an overall lucrative package-deal   2) initial set-up costs\footnote{If  EUs require different devices to access the Internet through different ISPs, e.g.,  depending on whether they provide internet through cable or DSL, they incur high initial set-up costs upon switch.}  3) overall impressions about intangibles such as customer-service, environment-friendliness, etc. 

We consider a common valuation $v^*$ for connecting to the Internet for EUs regardless of the content of the CP. This common valuation also models the valuation of EUs for CPs other than the CP considered in this paper, i.e. the valuation for connecting to the Internet regardless of the status of the  CP considered. The overall valuation of an EU located at $x\in[0,1]$ for connecting to the Internet via  ISP N (respectively,  ISP NoN) is considered to be $v^*+\kappa_u q_N-t_Nx$ (respectively, $v^*+\kappa_u q_{NoN}-t_{NoN}(1-x)$). Thus, the utility of an  EU who connects to the ISP $j\in\{N,NoN\}$ located at distance $x_j$ of the ISP, and is receiving the content with quality $q_j$, is:
\begin{equation}\label{equation:CP_2}
	u_{EU,j}(x_j)=v^*+\kappa_u q_j-t_jx_j-p_j,\quad j\in\{N,NoN\}.
\end{equation}

This model is generally known as the Hotelling model, which has been used in various fields including Internet market (eg, \cite{bourreau2015hotel} with $t_N=t_{NoN}$). As is common in Hotelling models, we  assume that the market is fully covered, i.e.,  each EU chooses exactly one ISP to buy Internet access from.  Equivalently,   $v^*$ is assumed to be sufficiently large so that the utility of EUs for connecting to the Internet is positive regardless of the choice of the ISP.

Preference for an ISP can be equivalently regarded as the reluctance for the other ISP. The lower $t_N$ and $t_{NoN}$ are, the easier EUs can switch between ISPs, and thus the lower is the ``inertia'' of EUs. We use the terms inertia and marginal transport costs interchangeably.  Thus, if transport cost for an ISP is high, more EUs   are  ``locked in'' with the other  ISP.

We consider the ratio of $t_N$ and $t_{NoN}$ as the relative bias of EUs for ISPs. The higher $\frac{t_N}{t_N+t_{NoN}}$ (respectively, $\frac{t_{NoN}}{t_N+t_{NoN}}$), the higher the bias of EUs for connecting to the Internet via ISP NoN (respectively, ISP N). We define the
 market power of  N and NoN as $\frac{t_{NoN}}{t_N+t_{NoN}}$ and  $\frac{t_N}{t_N+t_{NoN}}$, respectively.

A schematic of the market is presented in Figure~\ref{figure:market}.

\begin{figure}[t]
	\centering
	\includegraphics[width=0.5\textwidth]{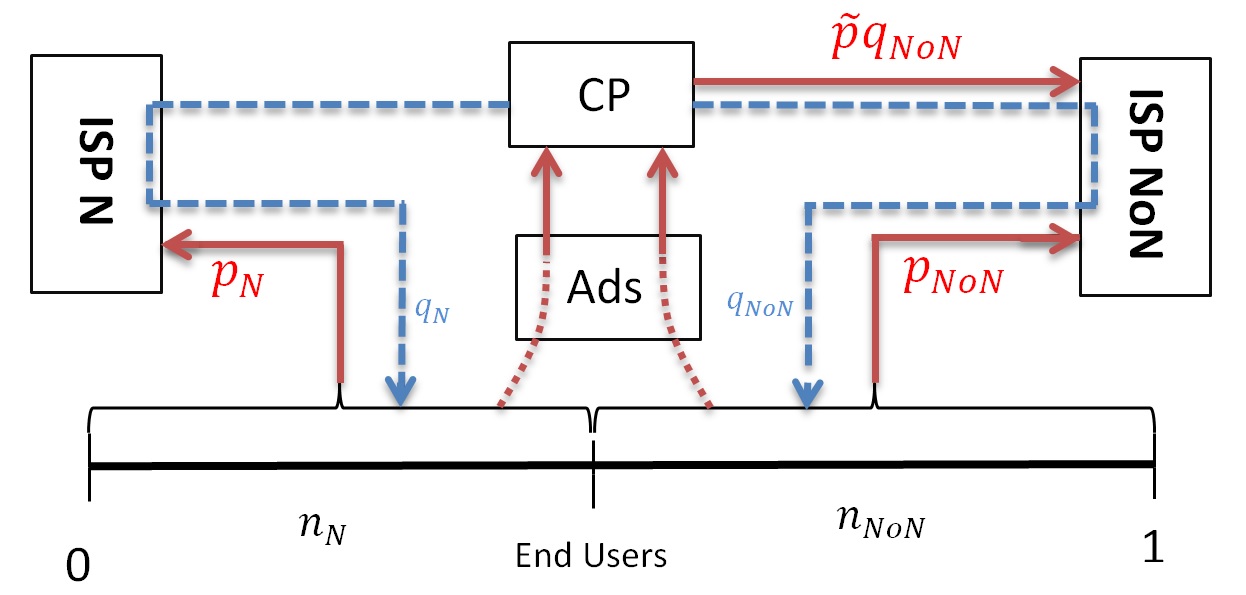}
	\caption{\small{A schematic of the market - red (solid) lines are the flow of the money and  blue (dashed) lines are the flow of the content.}}  \label{figure:market}\vspace{-7mm}
\end{figure}
\vspace{-3mm}

\subsection{Formulation}
We consider a complete-information game in which all players (the CP and ISPs) know all the assumptions, payoff expressions and the values of all the parameters, e.g., $\kappa_u, \kappa_{ad}, c, t_N, t_{NoN}, \tilde{q}_f, \tilde{q}_p$.
 The ISPs are the leaders of the game, and the CP and EUs are followers. The game proceeds in the following sequence:
\begin{enumerate}
	\item ISPs N and NoN determine Internet access fees for EUs ($p_N$ and $p_{NoN}$).
	\item  ISP NoN announces the marginal side-payment  ($\tilde{p}$) with knowledge of $p_N, p_{NoN}$.
	\item The CP decides on the quality of the content ($q_N$ and $q_{NoN}$) for EUs of each ISP, with knowledge of $p_N, p_{NoN}, \tilde{p}.$.
	\item EUs decide which ISP to join with the knowledge of $q_N, q_{NoN}, p_N, p_{NoN}$ (the EUs do not need to know $\tilde{p}$ for this decision).
\end{enumerate}
We assumed that the access fees are selected before the marginal side-payment as  the former usually remain  constant for a longer time horizon while the latter change more frequently depending on the demand and the network conditions.

In the sequential game framework,  we seek a \emph{Subgame Perfect Nash Equilibrium} (SPNE) using backward induction Section 3.5, p. 92, \cite{fudenberg1991tirole}.

\begin{definition}
	A strategy $(p^{eq}_N, p^{eq}_{NoN}, \tilde{p}^{eq}, q^{eq}_N, q^{eq}_{NoN})$ is a \emph{Subgame Perfect Nash Equilibrium (SPNE)} if and only if it constitutes a Nash Equilibrium (NE) of every subgame of the game, i.e., the CP or an ISP can not increase her payoff by  unilaterally varying one or more  of her choices. 
\end{definition}

\begin{remark}
\label{rfirst}
 $z^{eq}$ is a function of $q^{eq}_{NoN}$: it is $1$  if $q^{eq}_{NoN} = \tilde{q}_p$,  and $0$, otherwise. 
\end{remark}


\section{The Sub-Game Perfect Nash Equilibrium}\label{section:theory}

In Sections~\ref{section:stage4} to \ref{section:stage1}, we characterize the equilibrium strategies of each stage in a reverse order starting from Stage 4 and proceeding backward (backward induction). For each stage, we assume that each decision maker is aware of the strategies chosen by other decision makers in previous stages. 
 We enumerate all the possible SPNE strategies  and present the insights from the expressions in  Section~\ref{section:summaryof resutls}. In Section~\ref{determining}, we provide a computation strategy for determining which, if any,  of these possible SPNE strategies is (are) SPNE(s) at any given values of the parameters.  This strategy yields that there need not  exist an SPNE for certain values of parameters (Section~\ref{existence}).

 We start with a generic observation. From \eqref{equ:payoffISPsGeneral_new},  if $n_N>0$ then $p_N<c$ yields a negative payoff for  ISP N; while $p_N = c$ increases the payoff to $0.$   If $n_N=0$, the value of $p_N$ is of no importance. Therefore, without loss of generality we can consider $p_N\geq c$ under SPNE. Similarly, under SPNE,  $p_{NoN}\geq c$, if $z=0$. However, if $z=1$,   even with $p_{NoN}<c$, the payoff of ISP NoN may be positive.

\subsection{Stage 4: EUs decide which ISP to join}\label{section:stage4}

We characterize the division of EUs between ISPs i.e. $n_N$ and $n_{NoN}$, using the knowledge of the choices of the  ISPs and the CP in Stages 1, 2 and 3. From \eqref{equation:CP_2}, the location of the EU that is indifferent between joining either of the ISPs, $x_n$, is:


\begin{equation}\label{equ:xn}
	\ba
u_{EU,N}(x_n) & =  u_{EU,NoN}(x_{NoN}) \\
 \Rightarrow x_n &=\frac{t_{NoN}+\kappa_u(q_N-q_{NoN})+p_{NoN}-p_{N}}{t_{NoN}+t_N}. \ea
\end{equation}
\normalsize
EUs can determine $x_n$ in Stage $4$ since they know the  qualities $q_N, q_{NoN}$ and the access fees $p_N, p_{NoN}$ by then. And, EUs located at $[0,x_n)$ and $(x_n,1]$  join ISP, N and NoN, respectively (Recall that we assumed full market coverage):  

\be \label{equ:EUs_linear}
\ba
n_N &=
\left\{
\begin{array}{ll}
	0,  & \mbox{if } x_n < 0, \\
	\frac{t_{NoN}+\kappa_u(q_N-q_{NoN})+p_{NoN}-p_{N}}{t_{NoN}+t_N}, & \mbox{if } 0\leq x_n \leq 1,\\
	1, & \mbox{if } x_n>1,
\end{array}
\right. \\
n_{{NoN}}&=1-n_{N}.
\ea
\ee


\begin{table*}[t]
	\footnotesize
	\centering
	\begin{tabular}{c|c|c|c|c|}
		\cline{2-4}
		&  \multicolumn{3}{ |c| }{Conditions} \\ \cline{2-5}
		& $x_N\leq 0$ & $0< x_N< 1$ & $x_N\geq 1$ & \\
		& $\Big{(}q_{NoN}-q_N\geq  \frac{\Delta p + t_{NoN}}{\kappa_u}\Big{)}$ & $\Big{(}\frac{\Delta p-t_N}{\kappa_u}< q_{NoN}-q_N< \frac{\Delta p + t_{NoN}}{\kappa_u}\Big{)}$ & $\Big{(}q_{NoN}-q_N\leq \frac{\Delta p-t_N}{\kappa_u}\Big{)}$ & \\
		& All EUs join Non-neutral & EUs divide between both ISPs &  All EUs join Neutral & Union ($\bigcup$)\\
		\cline{1-5}
		\multicolumn{1}{ |c|  }{$z=0$}
		& $F^L_0$ & $F^I_0$ & $F^U_0$  & $F_0$   \\ \cline{1-5}
		\multicolumn{1}{ |c|  }{{$z=1$} }
		&  $F^L_1$ & $F^I_1$ & $F^U_1$  & $F_1$  \\ \cline{1-5}
		\multicolumn{1}{ |c|  }{{Union ($\bigcup$)} }
		&  $F^L$ & $F^I$ & $F^U$  & $\mathcal{F}$  \\ \cline{1-5}
	\end{tabular}
	\caption{Notations for different subsets of the feasible set. Expressions in parentheses are equivalent form of the conditions, e.g. $x_N\leq 0 \iff q_{NoN}-q_N\geq \frac{\Delta p+t_{NoN}}{\kappa_u}$.}\label{table:subsets}
\end{table*}

\subsection{Stage 3: The CP decides the qualities to offer over each ISP ($q_N$ and $q_{NoN}$)}\label{section:CPdecides}

 We characterize the equilibrium strategies, $q_N^{eq}$, $q_{NoN}^{eq}$,  using the knowledge of access fees $\vec{p}=(p_N,p_{NoN})$ and $\tilde{p}$ from stages 1 and 2. 


Let $F_1$ and $F_0$ be the set of strategies by which $z=1$ (i.e., $q_{NoN}> \tilde{q}_f$)  and $z=0$, respectively.

\be \label{equ:summarize_CP_candidate_f_new}
\ba
F_0=\{(0,0),(0,\tilde{q}_f),(\tilde{q}_f,0), (\tilde{q}_f,\tilde{q}_f)\},  
F_1=\{(0,\tilde{q}_p),(\tilde{q}_f,\tilde{q}_p)\}.
\ea
\ee
 $F_0$ and $F_1$ is further divided into three subsets, $F^L_i$, $F^I_i$, and $F^U_i$, for $i\in\{0,1\}$, depending on whether $x_N\leq 0$, $0<x_N<1$, or $x_N\geq 1$ (using \eqref{equ:xn}). Since $x_N$ is a function of $q_N$ and $q_{NoN}$, these conditions on $x_N$ lead to constraints on $q_N$ and $q_{NoN}$.
In Table~\ref{table:subsets}, we present the division of the feasible set into the above-mentioned subsets and the constraints on $q_N$ and $q_{NoN}$ for each subset. Note that $F^L_0\cup F^L_1=F^L$, $F^I_0\cup F^I_1=F^I$, and $F^U_0\cup F^U_1=F^U$.


\begin{remark}
\label{r3}
By offering $(\tilde{q}_f,\tilde{q}_f)$ ($(0, 0)$, respectively)   the CP earns $\kappa_{ad}\tilde{q}_f$ ($0$, respectively),  from  \eqref{equ:payoffCP_new}. Thus the CP does not choose  $(0, 0)$ in the SPNE.
\end{remark}

  In Section~\ref{section:Stage3_tie_break_assu}
 we present tie-breaking assumptions, motivated by natural practices; using these  we characterize the equilibrium strategies in Theorems~\ref{lemma:CP_z=0_new} and  \ref{theorem:p_tilde_new} in Section~\ref{section:Stage3_results}. 
\subsubsection{Tie- Breaking Assumptions}\label{section:Stage3_tie_break_assu}


First,  $(q^{eq}_N,q^{eq}_{NoN})\in F^L$ (respectively, $(q^{eq}_N,q^{eq}_{NoN})\in F^U$) yields that $n^{eq}_N=0$ (respectively, $n^{eq}_{NoN}=0$). Hence, the quality that the CP offers on ISP N (respectively,  ISP NoN) is of no importance. Thus:

\begin{assumption}\label{assumption:tie_n=0}
	If $(q^{eq}_N,q^{eq}_{NoN})\in F^L$ (respectively, $(q^{eq}_N,q^{eq}_{NoN})\in F^U$), then without loss of generality, $q^{eq}_{N}=0$ (respectively, $q^{eq}_{NoN}=0$).
\end{assumption}

\begin{remark}
\label{rback}
 Assumption~\ref{assumption:tie_n=0} rules out $(0,\tilde{q}_f)\in F^U_0$, $(\tilde{q}_f,0)\in F^L_0$, $(\tilde{q}_f,\tilde{q}_f)\in F^U_0\cup F^L_0$, $(0,\tilde{q}_p)\in F^U_0$, $(\tilde{q}_f,\tilde{q}_f)\in F^U_1\cup F^L_1$.
 \end{remark}

Next, in practice, it is natural to expect that if $z=1$ and $z=0$ both yield maximum payoffs  for the CP, then the CP chooses the higher (premium) quality,  $z=1$:

\begin{assumption}\label{assumption:tie4}
If the optimum solutions exist in $F_0$ and $F_1$, then the CP chooses the ones in  $F_1$.
\end{assumption}

The following  tie-breaking assumptions  are based on the natural assumption that the CP would prefer to diversify her content over different ISPs as long as this does not reduce her payoff. Thus, if the  global maximum payoff is attained by both the outcomes in which (1)   only one ISP is operating and (2) both ISPs are operating, then the CP chooses the strategies leading to the latter:

\begin{assumption}\label{assumption:tie}
	If there exists global optimum solutions in $F^I$, then they are preferred by the CP over  global optimum solutions in $F^L$ and $F^U$.
\end{assumption}

\begin{assumption}\label{assumption:tie_diversify}
	Let both the following strategies yield the same payoff for the CP: (i) $(q'_N,q'_{NoN})$ with  $\min(q'_N, q'_{NoN}) = 0$, and (ii)  $(q''_N,q''_{NoN})$ such that $\min(q''_N, q''_{NoN}) >0$. Then she chooses (ii), i.e., the one with positive quality on both ISPs.
\end{assumption}


 \begin{remark}
 \label{r5}
 Following Remarks~\ref{r3}, \ref{rback},  the following candidates for  $(q^{eq}_N,q^{eq}_{NoN})$
  remain: 
\be  \label{equ:summarize_CP_candidate_new}
\ba
&(0,\tilde{q}_f)\in F^I_0\cup F^L_0\ ,\ (\tilde{q}_f,0) \in F^I_0\cup F^U_0\ ,\ (\tilde{q}_f,\tilde{q}_f) \in F^I_0\ ,\\
& (0,\tilde{q}_{p}) \in F^I_1\cup F^L_1\ ,\  (\tilde{q}_f,\tilde{q}_{p}) \in F^I_1
\ea
\ee
\end{remark}

In the following tie-breaking assumption, we assume that the CP chooses the strategy that yields a higher social welfare for EUs, provided this does not reduce her payoff:

\begin{assumption}\label{assumption:tie_lowerp}
	If the payoff of the CP when only  ISP N is operating equals that when only ISP NoN is operating, then she prefers the strategy by which the ISP that offers the lower access fee, i.e., $p_i, \quad i\in\{N,NoN\}$, is operating.
\end{assumption}

Note that the above-mentioned assumptions override each other in the order specified. For example, if two strategies one in $F^L_1$ and the other in $F^I_0$ are both global maximum, then Assumption~\ref{assumption:tie4} suggests that the CP chooses the strategy in $F^L_1$, and Assumption~\ref{assumption:tie} suggests that the CP chooses the strategy in $F^I_0$. Since Assumption~\ref{assumption:tie4} comes before Assumption~\ref{assumption:tie}, the CP chooses the strategy in $F^L_0$.

Next, using these tie-breaking assumptions, we characterize the equilibrium strategies chosen by the CP.

\subsubsection{Main Results}\label{section:Stage3_results}
First, recall that $\Delta p=p_{NoN}-p_N$. We define  thresholds (1) $\tilde{p}_{t,1}, \tilde{p}_{t,2}, \tilde{p}_{t,3}$ and (2) $\Delta p_t$ that respectively appear in the results for side-payments and  access fees:

\begin{definition}\label{def:pt1,pt2}
	\begin{itemize}
		\item $\tilde{p}_{t,1}=\kappa_{ad}(1-\frac{\tilde{q}_{f}}{\tilde{q}_p}),$
		\item $\tilde{p}_{t,2}=\kappa_{ad} (n_{NoN}-\frac{\tilde{q}_f}{\tilde{q}_p})$, where $n_{NoN}=\frac{t_N+\kappa_u \tilde{q}_p-\Delta p}{t_N+t_{NoN}}$,
		\item  $ \tilde{p}_{t,3}= \kappa_{ad}n_{NoN}(1-\frac{\tilde{q}_f}{\tilde{q}_{p}})$, where $n_{NoN}=\frac{t_N+\kappa_u (\tilde{q}_{p}-\tilde{q}_f)-\Delta p}{t_N+t_{NoN}}$.
		\item  $\Delta p_{t}=\kappa_{u}(2\tilde{q}_{p}-\tilde{q}_f)-t_{NoN}$.
	\end{itemize}
\end{definition}



In  Theorem \ref{lemma:CP_z=0_new}, we characterize the equilibrium strategies of the CP by which $z^{eq}=0$, for different regions of $\Delta p$.



\begin{theorem}\label{lemma:CP_z=0_new}
	If $(q^{eq}_N,q^{eq}_{NoN})\in F_0$, then: \\
	1. if $-t_{NoN}< \Delta p< t_N$, then $(q^{eq}_N,q^{eq}_{NoN})=(\tilde{q}_f,\tilde{q}_{f})\in F^I_0$.\\
	2. if $\Delta p\geq t_N$, $(q^{eq}_N,q^{eq}_{NoN})=(\tilde{q}_f,0)\in F^U_0$.\\
	3. if $\Delta p\leq -t_{NoN}$,  $(q^{eq}_N,q^{eq}_{NoN})=(0,\tilde{q}_f)\in F^L_0$.\\
	Also,  each choice fetches the CP a payoff of $\kappa_{ad} \tilde{q}_f$.
\end{theorem}



Using Theorem \ref{lemma:CP_z=0_new} and Remark~\ref{r5},  in Theorem~\ref{theorem:p_tilde_new}, we characterize the equilibrium strategies of the CP.
We prove that if SPNE exists, the results are threshold-type: 1) when the marginal side-payment, i.e. $\tilde{p}$, is less than  a threshold, the CP chooses the premium quality on ISP NoN, i.e. $z^{eq}=1$, and 2) when $\tilde{p}$ exceeds the threshold, the CP chooses free quality on NoN, i.e.,  $z^{eq}=0$. 
 We also characterize the value of this threshold for different regions of $\Delta p$. Intuitively,     as $\Delta p$ increases, $n_{NoN}$  decreases. This affects the payoff of the CP, and subsequently the value of the side-payment that  NoN charges the CP. Thus, the value of the threshold on the marginal side-payment depends on $\Delta p$.

\begin{theorem}\label{theorem:p_tilde_new}
For each case in which SPNE exists: 
\begin{enumerate}
		\item If $\Delta p\leq  \kappa_u \tilde{q}_{p}-t_{NoN}$:
		\begin{itemize}
			\item if  $\tilde{p}\leq \tilde{p}_{t,1}$, then $z^{eq}=1$, and $(q^{eq}_N,q^{eq}_{NoN})=(0,\tilde{q}_{p})\in F^L_1$.
			\item if  $\tilde{p}> \tilde{p}_{t,1}$, then  $z^{eq}=0$, and  $q^{eq}_N$ and $q^{eq}_{NoN}$ are determined by Theorem~\ref{lemma:CP_z=0_new}.
		\end{itemize}
		\item If  $\kappa_u \tilde{q}_{p}-t_{NoN} <  \Delta p < t_N+\kappa_u \tilde{q}_{p}$, and $\tilde{q}_f\leq\frac{t_N+t_{NoN}}{\kappa_u}$:
		\begin{enumerate}
			\item if $\kappa_u \tilde{q}_{p}-t_{NoN}< \Delta p < t_N+\kappa_u(\tilde{q}_{p}-\tilde{q}_f)$, and:
			\begin{enumerate}
				\item if $\Delta p\geq \Delta p _{t}$:
				\begin{itemize}
					\item if  $\tilde{p}\leq \tilde{p}_{t,3}$, then $z^{eq}=1$ and  $(q^{eq}_N,q^{eq}_{NoN})=(\tilde{q}_f,\tilde{q}_{p})\in F^I_1$.
					\item if  $\tilde{p}> \tilde{p}_{t,3}$, then  $z^{eq}=0$, and  $q^{eq}_N$ and $q^{eq}_{NoN}$ are determined by Theorem~\ref{lemma:CP_z=0_new}.
				\end{itemize}
				\item if $\Delta p< \Delta p _{t}$:
				\begin{itemize}
					\item if  $\tilde{p}\leq \tilde{p}_{t,2}$, then $z^{eq}=1$ and $(q^{eq}_N,q^{eq}_{NoN})=(0,\tilde{q}_{p})\in F^I_1$.
					\item if  $\tilde{p}> \tilde{p}_{t,2}$, then  $z^{eq}=0$, and  $q^{eq}_N$ and $q^{eq}_{NoN}$ are determined by Theorem~\ref{lemma:CP_z=0_new}.
				\end{itemize}
				
			\end{enumerate}
			\item if $t_N+\kappa_u(\tilde{q}_{p}-\tilde{q}_f) \leq \Delta p < t_N+\kappa_u \tilde{q}_{p}$:
			\begin{enumerate}
				\item if  $\tilde{p}\leq \tilde{p}_{t,2}$, then $z^{eq}=1$, and $(q^{eq}_N,q^{eq}_{NoN})=(0,\tilde{q}_{p})\in F^I_1$.
				\item if  $\tilde{p}> \tilde{p}_{t,2}$, then  $z^{eq}=0$, and  $q^{eq}_N$ and $q^{eq}_{NoN}$ are determined by Theorem~\ref{lemma:CP_z=0_new}.
			\end{enumerate}
		\end{enumerate}
		\item If $ \kappa_u \tilde{q}_{p}-t_{NoN} < \Delta p < t_N+\kappa_u \tilde{q}_{p}$, and $\tilde{q}_f>\frac{t_N+t_{NoN}}{\kappa_u}$:
		\begin{enumerate}
			\item if  $\tilde{p}\leq \tilde{p}_{t,2}$, then $z^{eq}=1$, and $(q^{eq}_N,q^{eq}_{NoN})=(0,\tilde{q}_{p})\in F^I_1$.
			\item if  $\tilde{p}> \tilde{p}_{t,2}$, then  $z^{eq}=0$, and  $q^{eq}_N$ and $q^{eq}_{NoN}$ are determined by Theorem~\ref{lemma:CP_z=0_new}.
		\end{enumerate}
		\item If $\Delta p\geq  t_N+\kappa_u \tilde{q}_{p}$, then $z^{eq}=0$, and  $q^{eq}_N$ and $q^{eq}_{NoN}$ are determined by Theorem~\ref{lemma:CP_z=0_new}.
	\end{enumerate}
\end{theorem}


Note that the thresholds $\tilde{p}_{t,1}$, $\tilde{p}_{t,2}$, and $\tilde{p}_{t,3}$ are decreasing with respect to $\frac{\tilde{q}_f}{\tilde{q}_p}$. Thus, the Theorem confirms the intuition that as $\frac{\tilde{q}_{p}}{\tilde{q}_f}$ increases, the  threshold  on $\tilde{p}$ after which the CP chooses the free quality over the premium one, is higher.  Also, with high $\tilde{q}_{p}$ and low $t_{NoN}$, the CP prefers the strategy by which the neutral ISP is driven out of the market.


\subsection{Stage 2: ISP NoN determines the marginal side-payment, $\tilde{p}$:} \label{section:stage2}

The equilibrium marginal side payment,  $\tilde{p}^{eq}$, must maximize $\pi_{NoN}(p_{NoN},\tilde{p})$  in \eqref{equ:payoffISPsGeneral_new}, given $p_{NoN}$ and $p_N$. We determine it accordingly. Note that if $z=0$,  then the payoff of ISP NoN, $\pi_{NoN,z=0}(p_{NoN},\tilde{p})$,  is independent of $\tilde{p}$ (from \eqref{equ:payoffISPsGeneral_new}).  Thus, we only need to characterize $\tilde{p}^{eq}$ by which $z^{eq}=1$, 
   i.e., 
    $(q^{eq}_N,q^{eq}_{NoN})\in F_1$, which we do in  Theorem~\ref{theorem:NE_stage2_new_suff}. But, first, we introduce a tie-breaking assumption (Assumption~\ref{assumption:ISP_p_tilde}) for NoN,  motivated by natural practices.
	



 Due to additional scrutiny associated with a non-neutral regime (eg, additional monitoring by regulators), whenever  $z^{eq}=0$ and $z^{eq}=1$ fetch equal payoff for NoN, we assume that she opts for neutrality,  i.e., $\tilde{p}$ such that $z^{eq}=0$.

\begin{assumption} \label{assumption:ISP_p_tilde}
 Consider $\tilde{p}_1$  such that $(q^{eq}_N,q^{eq}_{NoN})\in F_1$, i.e., $z^{eq}=1$, and $\tilde{p}_2$  such that $(q^{eq}_N,q^{eq}_{NoN})\in F_0$, i.e., $z^{eq}=0$. If they yield the same payoff for ISP NoN, she chooses $\tilde{p}_2$. 
\end{assumption}




 We define a threshold,  $\tilde{p}_t$ which,  encapsulates $\tilde{p}_{t,1}$, $\tilde{p}_{t,2}, \tilde{p}_{t,3}$ (defined in Definition~\ref{def:pt1,pt2}), and by Theorem~\ref{theorem:p_tilde_new} becomes the maximum $\tilde{p}$  for the CP to choose $z^{eq}=1$:  

\begin{definition}\label{def:pt}
	We define  $\tilde{p}_t=\tilde{p}_{t,1}$ if conditions of item 1 of Theorem~\ref{theorem:p_tilde_new} are met,  $\tilde{p}_t=\tilde{p}_{t,2}$ if the conditions of items 2-a-ii, 2-b, and 3 of Theorem~\ref{theorem:p_tilde_new} are met, and $\tilde{p}_t=\tilde{p}_{t,3}$ if the conditions of item 2-a-i of  Theorem~\ref{theorem:p_tilde_new} are met. 
\end{definition}


\begin{theorem}\label{theorem:NE_stage2_new_suff}
	\begin{enumerate}
\item If $z^{eq}=1$, then $\tilde{p}^{eq}=\tilde{p}_{t}$.
\item $z^{eq}=1$ if and only if
	$\pi_{NoN}(p_{NoN},\tilde{p}_{t})>\pi_{NoN,z=0}(p_{NoN},\tilde{p})$ and $\Delta p<t_N+\kappa_u \tilde{q}_{p}$. 
\end{enumerate}
\end{theorem}
 Thus, $\Delta p$ being less than a threshold and the existence of  $\tilde{p}$ by which ISP NoN earns more than  when $z=0$,  are necessary and sufficient conditions for $z^{eq}=1$. Without the former, $n_{NoN} = 0$, and trivially the CP does not offer her content on NoN. The latter follows because NoN prefers neutrality, unless non-neutrality fetches her a higher payoff (Assumption~\ref{assumption:ISP_p_tilde}), and she can always enforce neutrality  by choosing an extremely large  $\tilde{p}$ by which $z=0$. Also, NoN chooses for $\tilde{p}$, the maximum value,   $\tilde{p}_t$,  by which   $z^{eq}=1$. 





\begin{remark}
\label{r0}
Once  the SPNE access fees $p^{eq}_N, p^{eq}_{NoN}$ are known, whenever SPNE exists,  1) the SPNE qualities of content,  $q^{eq}_N, q^{eq}_{NoN}$, SPNE marginal side-payment $\tilde{p}^{eq}$,  may now be uniquely determined from Theorems~\ref{lemma:CP_z=0_new}, \ref{theorem:p_tilde_new} and \ref{theorem:NE_stage2_new_suff},   2) subsequently \eqref{equ:EUs_linear} give the SPNE EU subscriptions $n^{eq}_N, n^{eq}_{NoN}$. Remark~\ref{rfirst} gives $z^{eq}$ from $q^{eq}_{NoN}$.   
\end{remark}

\subsection{Stage 1: ISPs determine the access fees  $p_N$ and $p_{NoN}$ for the EUs:} \label{section:stage1}

We characterize the SPNE access fees, $p^{eq}_N, p^{eq}_{NoN}$, considering the following cases separately: (i) low inertia, involving various upper bounds on  $t_N, t_{NoN}$
(ii) high inertia involving the lower bound  $t_N+t_{NoN} > \kappa_u \tilde{q}_p$.

Recall that $p^{eq}_N\geq c$, and if $z^{eq}=0$, $p^{eq}_{NoN}\geq c$ (Section~\ref{section:theory}, second paragraph). If $0 < x_n < 1$, i.e. $(q^{eq}_N,q^{eq}_{NoN})\in F^I$, from \eqref{equ:EUs_linear},   the payoffs of  ISPs are:

\be \label{equ:UN_new}
\footnotesize
\pi_N(p_N)=(p_N-c)\frac{t_{NoN}+\kappa_u (q_N-q_{NoN})+p_{NoN}-p_N}{t_N+t_{NoN}},
\ee

\be\label{equ:UNoN_new}
\footnotesize
\ba
\pi_{NoN}(p_{NoN},\tilde{p})&=(p_{NoN}-c)\frac{t_N+\kappa_u (q_{NoN}-q_N)+p_N-p_{NoN}}{t_N+t_{NoN}}\\
&\qquad \qquad +zq_{NoN}\tilde{p}.
\ea
\ee
\normalsize

We start with by characterizing the SPNE access fees in the neutral region, i.e., when $z^{eq} = 0.$
\begin{theorem}\label{theorem:neutralnotexists_q>}
\label{theorem:neutralz=0_q<}
\begin{enumerate}
\item The only possible SPNE access fee by which  $(q^{eq}_N,q^{eq}_{NoN})\in F_0$, i.e. $z^{eq}=0$ is 	$p^{eq}_N=c+\frac{1}{3}(2t_{NoN}+t_N)$ and $p^{eq}_{NoN}=c+\frac{1}{3}(2t_N+t_{NoN})$.  A necessary condition for this access fee to be a SPNE strategy is $\pi_{NoN, z=0}(p^{eq}_{NoN},\tilde{p}^{eq})\geq \pi_{NoN}(p^{eq}_{NoN},\tilde{p}_t) $.
\item 	If $t_N+t_{NoN}\leq \kappa_u \tilde{q}_p$, there is no SPNE by which $(q^{eq}_N,q^{eq}_{NoN})\in F_0$, i.e. $z^{eq}=0$.
\end{enumerate}
\end{theorem}

\subsubsection{Low inertia}
 Considering various upper bounds on the inertias $t_N, t_{NoN}$, we 1) prove that  there is no SPNE by which  $z^{eq}=0$ (Theorem~\ref{theorem:neutralnotexists_q>}) 2)  characterize the possible SPNE access fees by which $z^{eq}=1$ (Theorem~\ref{theorem:NE_stage1_new_q>}). In the latter case, we show that  (a) if the weighted sum of  inertias is very  small, then a unique SPNE exists (b) otherwise,  a unique SPNE exists only under certain additional conditions. Numerical analysis under a wide range of parameters  reveal that   the latter conditions are always satisfied.

\begin{theorem}\label{theorem:NE_stage1_new_q>}
 The SPNE access fees, $p^{eq}_N$ and $p^{eq}_{NoN}$ by which $(q^{eq}_N,q^{eq}_{NoN})\in F_1$, i.e.,  $z^{eq}=1$, are:
	\begin{enumerate}
		\item $p^{eq}_{NoN}=c+\kappa_u \tilde{q}_{p}-t_{NoN}$ and $p^{eq}_N=c$ if and only if $ t_N+2t_{NoN} \leq \tilde{q}_{p}(\kappa_u+\kappa_{ad})$ and $t_{NoN} < \kappa_u \tilde{q}_p+\kappa_{ad}(\tilde{q}_p-\tilde{q}_f).$
		\item  $p^{eq}_{NoN}=c+\frac{t_{NoN}+2t_N+\tilde{q}_{p}(\kappa_u -2\kappa_{ad})}{3}$ and $p^{eq}_{N}=c+\frac{2t_{NoN}+t_N-\tilde{q}_{p}(\kappa_u +\kappa_{ad})}{3}$ when $t_N+t_{NoN}\leq \kappa_u \tilde{q}_p$, $t_N+2t_{NoN} > \tilde{q}_{p}(\kappa_u+\kappa_{ad})$, and $\pi_N(p^{eq}_N)\geq p^d_t-c$, where
		$p^d_{t}=\frac{\kappa_{ad}\tilde{q}_f (t_N+t_{NoN})}{p^{eq}_{NoN}-c+\kappa_{ad} \tilde{q}_p}+p^{eq}_{NoN}-t_{NoN}-\kappa_u \tilde{q}_p$.
	\end{enumerate}
\end{theorem}

\begin{remark}
\label{rlast}
The inertias are  1) very small 2) small to moderate in the cases represented in
1) Theorem~\ref{theorem:NE_stage1_new_q>}-1 and 2) Theorem~\ref{theorem:NE_stage1_new_q>}-2, respectively.
If $t_N+t_{NoN}\leq \kappa_u \tilde{q}_p$, then $t_{NoN} < \kappa_u \tilde{q}_p+\kappa_{ad}(\tilde{q}_p-\tilde{q}_f)$, since $\tilde{q}_p > \tilde{q}_f.$ Thus, Theorem~\ref{theorem:NE_stage1_new_q>} covers the case that $t_N+t_{NoN}\leq \kappa_u \tilde{q}_p$. 
\end{remark}

\subsubsection{High inertia}

We focus on the case that inertias are lower bounded: $t_N+t_{NoN} >  \kappa_u \tilde{q}_p.$
In Theorem~\ref{theorem:NE_stage1_new_q<}, we characterize all possible SPNE access fees by which $z^{eq}=1$. In Theorem~\ref{theorem:neutralz=0_q<}, we characterize the only possible SPNE access fee by which $z^{eq}=0$. We list as ``possible'' SPNE those strategies that we could not rule out as SPNE by applying the lack of existence of profitable unilateral deviation criteria on the specific interaction-dynamic in question; thus, if SPNE exists, it would be one or more of the ``possible'' ones, the existence can not however be concluded at this point. But, subsequently, in Theorem~\ref{theorem:bigt}, we sharpen the result in Theorem~\ref{theorem:NE_stage1_new_q<}, by providing the unique SPNE when at least one of  $t_N, t_{NoN}$ is very large. 


\begin{theorem}\label{theorem:NE_stage1_new_q<}
 Let $t_N+t_{NoN} >  \kappa_u \tilde{q}_p$. The only possible  SPNE access fees by which $(q^{eq}_N,q^{eq}_{NoN})\in F_1$, i.e., $z^{eq}=1$, are:
\begin{enumerate}
\item \label{i1} If $\Delta p\leq \kappa_u \tilde{q}_p-t_{NoN}$, then $p^{eq}_{NoN}=c+\kappa_u \tilde{q}_{p}-t_{NoN}, p^{eq}_N=c$.
\item \label{i2} If  (i)  $\kappa_u \tilde{q}_{p}-t_{NoN}<\Delta p^{eq}<\kappa_u (2\tilde{q}_{p}-\tilde{q}_f)-t_{NoN}$ or (ii) $t_N+\kappa_u(\tilde{q}_{p}-\tilde{q}_f)<\Delta p^{eq}<t_N+\kappa_u \tilde{q}_{p}$, then $p^{eq}_{NoN}=c+\frac{t_{NoN}+2t_N+\tilde{q}_{p}(\kappa_u -2\kappa_{ad})}{3}, p^{eq}_{N}=c+\frac{2t_{NoN}+t_N-\tilde{q}_{p}(\kappa_u +\kappa_{ad})}{3}.$ 
\item \label{i3} If   $\kappa_u (2\tilde{q}_{p}-\tilde{q}_f)-t_{NoN}< \Delta p^{eq}<t_N+\kappa_u (\tilde{q}_{p}-\tilde{q}_f)$, then $p^{eq}_{NoN}=c+\frac{t_{NoN}+2t_N+  (\tilde{q}_{p}-\tilde{q}_f)(\kappa_u -2\kappa_{ad})}{3}, p^{eq}_{N}=c+\frac{2t_{NoN}+t_N-(\tilde{q}_{p}-\tilde{q}_f)(\kappa_u +\kappa_{ad})}{3}$. 
\item \label{i4} $p^{eq}_{NoN}=c, p^{eq}_N=c-\kappa_u(2\tilde{q}_{p}-\tilde{q}_f)+t_{NoN}.$ 
\end{enumerate}
A necessary condition for 2, 3, 4 to be SPNE is $\pi_{NoN}(p^{eq}_{NoN},\tilde{p}_{t})>\pi_{NoN,z=0}(\tilde{p}^{eq}_{NoN},\tilde{p}).$ Additional necessary conditions  for these to be SPNE are:  $\tilde{q}_{p}\leq \frac{2t_{NoN}+t_N}{\kappa_u+\kappa_{ad}}$ (for part-\ref{i2}),   $\tilde{q}_{p}-\tilde{q}_f\leq \frac{2t_{NoN}+t_N}{\kappa_u+\kappa_{ad}}$ (for part-\ref{i3}), $2\tilde{q}_{p}-\tilde{q}_f\leq \frac{t_{NoN}}{\kappa_{u}}$ (for part-\ref{i4}).
\end{theorem}

(\ref{i1}) follows from Theorem~\ref{theorem:NE_stage1_new_q>}-1 because the region   $ t_N+2t_{NoN} \leq \tilde{q}_{p}(\kappa_u+\kappa_{ad})$ and $t_{NoN} < \kappa_u \tilde{q}_p+\kappa_{ad}(\tilde{q}_p-\tilde{q}_f)$ in general has a non-empty intersection with that represented by $t_N+t_{NoN} >  \kappa_u \tilde{q}_p$.






We now consider the case that at least one of $t_N, t_{NoN}$  is large enough. It may be shown in this case that a) access fees 1), 2), and 4) listed in Theorem~\ref{theorem:NE_stage1_new_q<} are not SPNE,  b) 3) is an SPNE, which therefore becomes the unique SPNE.

\begin{theorem}\label{theorem:bigt}
	When either $t_N$ or $t_{NoN}$ is large enough,  the only SPNE access fee by which $z^{eq}=1$ is $p^{eq}_{NoN}=c+\frac{t_{NoN}+2t_N+  (\tilde{q}_{p}-\tilde{q}_f)(\kappa_u -2\kappa_{ad})}{3}$ and $p^{eq}_{N}=c+\frac{2t_{NoN}+t_N-(\tilde{q}_{p}-\tilde{q}_f)(\kappa_u +\kappa_{ad})}{3}$.
\end{theorem}


\begin{remark}
 If either $t_N$ or $t_{NoN}$ is large, then $q_N, q_{NoN}$ have negligible  effect  in EU subscription, by  \eqref{equation:CP_2}. Thus, this scenario resembles that when both ISPs are neutral, i.e. the benchmark case, in which   the unique SPNE  access fees are as in Theorem~\ref{theorem:bigt} with $\tilde{q}_p=\tilde{q}_f$ (from Theorem~\ref{lemma:NEz=0}).
\end{remark}


\begin{remark}
Following Remark~\ref{r0}, $\tilde{p}^{eq}, q^{eq}_N, q^{eq}_{NoN}, z^{eq}, n^{eq}_N$, and  $n^{eq}_{NoN}$
are unique  functions of  $p^{eq}_N, p^{eq}_{NoN}$, which leads to mathematical conditions,  through \eqref{equ:xn}, that ensure that
$q^{eq}_N, q^{eq}_{NoN} $ belong in the appropriate  $F^L_i, F^I_i, F^U_i$, where $i \in \{0, 1\}$ of Table~\ref{table:subsets}, as per Theorem~\ref{theorem:p_tilde_new}. These mathematical conditions provide the conditions in Theorems in Section~\ref{section:stage1} for the corresponding $p^{eq}_N, p^{eq}_{NoN}$ to constitute the SPNE access fees.
\end{remark}

Finally note that the candidate SPNE access fees  in different theorems (defined for different regions of $t_N$ and $t_{NoN}$) can be identical, e.g. {Theorem~\ref{theorem:NE_stage1_new_q>}-2} and {Theorem~\ref{theorem:NE_stage1_new_q<}-2} leading to identical values of other decisions, $q^{eq}_N, q^{eq}_{NoN}, \tilde{p}^{eq}, z^{eq}, n^{eq}_N, n^{eq}_{NoN}$ per Remark~\ref{r0}. Thus, different ranges of the inertias may correspond to the same SPNE.



 \subsection{The overall SPNE outcomes of the game}\label{section:discussion}\label{section:summaryof resutls}
We list the possible SPNEs, referred to as the SPNE candidates, comprising of the decisions in all the stages by the CP and ISPs,  listed in the previous subsections,    and discuss the implications of each being an SPNE.




\textbf{SPNE candidate (a):}  $p^{eq}_{NoN}=c+\kappa_u \tilde{q}_{p}-t_{NoN}$, $p^{eq}_N=c$, $z^{eq}=1$, 
$\tilde{p}^{eq}=\tilde{p}_{t,1}=\kappa_{ad}(1-\frac{\tilde{q}_f}{\tilde{q}_p})$, $(q^{eq}_N,q^{eq}_{NoN})=(0,\tilde{q}_p)\in F^L_1$, $n^{eq}_N=0, n^{eq}_{NoN}=1$. 

\begin{remark}
\label{r2}
The access fees correspond to  Theorem~\ref{theorem:NE_stage1_new_q>}-1. 
 Item 1 of Theorem~\ref{theorem:p_tilde_new} 
yields that  $(q^{eq}_N,q^{eq}_{NoN})=(0,\tilde{q}_p)\in F^L_1$. Thus,  $n^{eq}_N=0, n^{eq}_{NoN}=1$. Also, by Theorem~\ref{theorem:NE_stage2_new_suff}, 
$\tilde{p}^{eq}=\tilde{p}_{t,1}=\kappa_{ad}(1-\frac{\tilde{q}_f}{\tilde{q}_p})$. From Theorems~\ref{theorem:neutralnotexists_q>}, \ref{theorem:NE_stage1_new_q>}, Remark~\ref{r0},  
   candidate (a) is  the unique SPNE  if $t_N$ and $t_{NoN}$ are very small, and hence, EUs are not locked-in with the ISPs - specifically, if $t_N+2t_{NoN} \leq \tilde{q}_{p}(\kappa_u+\kappa_{ad})$ and $t_{NoN} < \kappa_u \tilde{q}_p+\kappa_{ad}(\tilde{q}_p-\tilde{q}_f)$. 
\end{remark}

  For this candidate, the CP offers the content with premium quality on ISP NoN, and does not offer her content on ISP N. 
 In general, the EUs can receive a better quality of content (the premium quality) only  on  NoN, and that yields a higher advertisement revenue for the CP. For small inertias, referring to \eqref{equation:CP_2},  the qualities offered by the CP on the two ISPs play important roles in the utilities of the EUs. Thus, by offering her content only on NoN, and that too with premium quality, the CP can motivate all EUs  to subscribe to NoN, and enhance her advertisement revenue. This also drives  N  out of the market, even when he offers the minimum possible access fee equaling the service cost $c$.

The  access fee chosen by NoN, $p^{eq}_{NoN}$,  increases with (i) increasing the sensitivity of EUs to the quality ($\kappa_u$), (ii) increasing the value of the premium quality ($\tilde{q}_p$), and (iii) decreasing NoN's marginal transport cost, $t_{NoN}$ (recall that NoN's  market power is inversely proportional to $t_{NoN}$).

  The CP pays NoN a positive side-payment, $\tilde{p}^{eq}\tilde{q}_p$,  which increases with (i) the CP's sensitivity to the quality of the advertisement, i.e. $\kappa_{ad}$, and (ii) the difference between the premium and free qualities, i.e. $\tilde{q}_p-\tilde{q}_f$, which represents the additional value the premium quality creates for the CP.


\textbf{SPNE Candidate  (b):} $p^{eq}_{NoN}=c+\frac{t_{NoN}+2t_N+\tilde{q}_{p}(\kappa_u -2\kappa_{ad})}{3}$, $p^{eq}_{N}=c+\frac{2t_{NoN}+t_N-\tilde{q}_{p}(\kappa_u +\kappa_{ad})}{3}$, $z^{eq}=1$, $\tilde{p}^{eq}=\tilde{p}_{t,2}=\kappa_{ad}(n^{eq}_{NoN}-\frac{\tilde{q}_f}{\tilde{q}_p})$, $(q^{eq}_N,q^{eq}_{NoN})=(0,\tilde{q}_p)\in F^I_1$, $n^{eq}_N=\frac{t_N+2t_{NoN}-\tilde{q}_p(\kappa_u+\kappa_{ad})}{3(t_N+t_{NoN})}, n^{eq}_{NoN}=\frac{2t_N+t_{NoN}+\tilde{q}_p(\kappa_u+\kappa_{ad})}{3(t_N+t_{NoN})}.$ 

\begin{remark}
The access fees correspond to Theorem~\ref{theorem:NE_stage1_new_q>}-2 and Theorem~\ref{theorem:NE_stage1_new_q<}-2. These access fees were identified as candidate SPNEs  such that $\Delta p$ satisfies item  3 of Theorem~\ref{theorem:p_tilde_new}. Thus, $(q^{eq}_N,q^{eq}_{NoN})=(0,\tilde{q}_p)\in F^I_1$. By Theorem~\ref{theorem:NE_stage2_new_suff}, $\tilde{p}^{eq}=\tilde{p}_{t,2}=\kappa_{ad}(n_{NoN}-\frac{\tilde{q}_f}{\tilde{q}_p})$. From $\Delta p=p^{eq}_{NoN}-p^{eq}_N$, and \eqref{equ:EUs_linear}, the expressions for $n^{eq}_N$ and $n^{eq}_{NoN}$ follow. From Theorems~\ref{theorem:neutralnotexists_q>}, \ref{theorem:NE_stage1_new_q>}-2, Remark~\ref{r0},
when inertias are small to moderate, i.e.,  $t_N+t_{NoN}\leq \kappa_u \tilde{q}_p$, $t_N+2t_{NoN} > \tilde{q}_{p}(\kappa_u+\kappa_{ad})$, the only possible SPNE is (b) (though it need not be the SPNE).
\end{remark}

Similar to (a), with candidate (b), the CP does not offer her content over  ISP N. But, unlike in (a), the EUs do not all migrate to NoN, owing to the different selections of the access fees by the ISPs. Thus, the CP can not always force N out of the market by her selection of qualities over the two SPs, particularly when the inertias are not very small and  non-negligible fractions of EUs may still be locked-in with the ISPs in effect, referring to \eqref{equation:CP_2} (like the cases in Theorem~\ref{theorem:NE_stage1_new_q>}-2, Theorem~\ref{theorem:NE_stage1_new_q<}-2). But,  since $n^{eq}_{NoN}<1$, NoN's side-payment to the CP in this case  is lower than that in  (a), which
 compensates for her loss in revenue in EU subscription. 

 $p^{eq}_{NoN}$ (respectively, $p^{eq}_N$) increases with a rate $\frac{2}{3}$rd the rate of the growth of $t_N$ (respectively, $t_{NoN}$). Intuitively,  the higher $t_N$ (while $t_{NoN}$ fixed) is, NoN can retain more EUs despite charging them high access fees (the ``lock-in'' effect), thus the higher is her market power, and  the higher is $p^{eq}_{NoN}$. But, counter-intuitively, $p_N$ (respectively, $p_{NoN}$) also increases with a rate $\frac{1}{3}$rd of the rate of growth of $t_N$ (respectively, $t_{NoN}$). This happens because of competition between ISPs: with the increase of $t_{NoN}$, more EUs  prefer  N to NoN. Thus N can select a higher access fee for EUs. This allows her competitor, NoN, to increases her access fee, but with a rate lower than the rate by which N increases her access fee. 

 Next,  ISP N's  access fee, $p^{eq}_N$ is a decreasing function of the premium quality, $\tilde{q}_p$, the sensitivity of EUs and the CP to the quality, $\kappa_u$ and $\kappa_{ad}$. On the other hand, the relationship between NoN's access fee, $p^{eq}_{NoN}$ and these parameters  is more complicated.  $p^{eq}_{NoN}$ increases (respectively, decreases) with respect to the sensitivity of EUs (respectively, CP) to the quality. Thus, the more the CP is sensitive to the quality,  NoN can provide cheaper access fees, and recoup the loss by charging the CP more. 
  It is for this reason that we see that if the sensitivity of EUs to the quality does not dominate the sensitivity of the CP ($\kappa_u < 2\kappa_{ad}$), then $p^{eq}_{NoN}$ decreases with respect to $\tilde{q}_p$. If not, then  $p^{eq}_{NoN}$ increases with respect to $\tilde{q}_p$. 

 Next, $n_{NoN}$  increases with the premium quality, $\tilde{q}_p$, and the sensitivity of the CP and EUs to the quality,  $\kappa_u$ and $\kappa_{ad}$. The more complex relationship between $n_{NoN}$ (and thus $n_N$) and $t_N$ and $t_{NoN}$ will be discussed in Section~\ref{section:numerical_more_results}.


\textbf{SPNE Candidate  (c):}  $p^{eq}_{NoN}=c+\frac{t_{NoN}+2t_N+  (\tilde{q}_{p}-\tilde{q}_f)(\kappa_u -2\kappa_{ad})}{3}$, $p^{eq}_{N}=c+\frac{2t_{NoN}+t_N-(\tilde{q}_{p}-\tilde{q}_f)(\kappa_u +\kappa_{ad})}{3}$, $z^{eq}=1$, $\tilde{p}^{eq}=\tilde{p}_{t,3}=\kappa_{ad}n^{eq}_{NoN}(1-\frac{\tilde{q}_f}{\tilde{q}_p})$, $(q^{eq}_N,q^{eq}_{NoN})=(\tilde{q}_f,\tilde{q}_p)\in F^I_1$, $n^{eq}_N=\frac{t_N+2t_{NoN}-(\tilde{q}_p-\tilde{q}_f)(\kappa_u+\kappa_{ad})}{3(t_N+t_{NoN})}, n^{eq}_{NoN}=\frac{2t_N+t_{NoN}+(\tilde{q}_p-\tilde{q}_f)(\kappa_u+\kappa_{ad})}{3(t_N+t_{NoN})}.$  

 Candidate (c) represents the case that EUs are split between the two ISPs, and the CP  offers her content  with free quality on N and with premium quality on NoN.
  The dependence of  $p^{eq}_N, p^{eq}_{NoN}, n_N, n_{NoN}$ on various parameters are similar to that in (b), with the difference that these depend on the difference in the qualities, i.e. $\tilde{q}_p-\tilde{q}_f$, in the current candidate instead of only $\tilde{q}_p$ as in (b). This difference is expected as  the CP offers quality $\tilde{q}_f$ ($0$, respectively)  on N in the current candidate (candidate (b) respectively). Also, the side-payment is guaranteed to be positive in this case.



   Recall that by  Theorem~\ref{theorem:bigt}, when either of $t_N$ or $t_{NoN}$ is large, then (c) is the only candidate  by which $z^{eq}=1$.   We can now provide the intuition for this. An ISP's subscription revenue depends on both the number of EUs and the access fee she charges. When either of $t_N$ or $t_{NoN}$ is large, then the access fees are  large in both candidates (b) and (c), and hence ISPs  prefer these to the candidates by which they discount the access fee heavily to attract EUs ((a) and (d)). Also, large $t_{NoN}$ or $t_N$ decreases the effect of quality of the content on the decision of EUs.  Thus, the CP cannot control the number of EUs with each ISP by strategically controlling her quality. Therefore she simply chooses  maximum possible quality on both ISPs instead of choosing strategic qualities to control the equilibrium. Thus, (c) arises rather than (b).


\textbf{SPNE Candidate  (d):} $p^{eq}_{NoN}=c$, $p^{eq}_N=c-\kappa_u(2\tilde{q}_{p}-\tilde{q}_f)+t_{NoN}$, $z^{eq}=1$, $\tilde{p}^{eq}=\tilde{p}_{t,3}=\kappa_{ad}n^{eq}_{NoN}(1-\frac{\tilde{q}_f}{\tilde{q}_p})$, $(q^{eq}_N,q^{eq}_{NoN})=(\tilde{q}_f,\tilde{q}_p)\in F^I_1$,  $n^{eq}_N=\frac{\kappa_u \tilde{q}_p}{t_N+t_{NoN}}, n^{eq}_{NoN}=\frac{t_N+t_{NoN}-\kappa_u \tilde{q}_p}{t_N+t_{NoN}}$. 

   To ensure that $p^{eq}_N\geq c$, this candidate can be an SPNE only when $t_{NoN}$ is large, i.e., when EUs are reluctant to join NoN, which explains why NoN  selects a low access fee  of $c.$  


The side-payment is similar to the one in candidate (c).

Next,  $p^{eq}_N$  decreases with $\tilde{q}_p$ and $\kappa_u$, at a rate that is twice  the rate of increase of the utility of EUs from $\kappa_u$ and $\tilde{q}_p$ when connecting to  NoN. Thus, the rate of increase in the utility of EUs for  N is higher than that for  NoN; hence, $n^{eq}_N$ increases with $\tilde{q}_p$ and $\kappa_u.$ 
 Next,  $p^{eq}_N$ increases with $t_{NoN}$. Thus,  $n^{eq}_N$ decreases with respect to $t_{NoN}$\footnote{Note that the utility of EUs connecting to ISP NoN also decreases with $t_{NoN}$ \eqref{equation:CP_2}. However, the rate of decrease in the utility of EUs connecting to  NoN ($t_{NoN}$ is multiplied to a coefficient smaller than one) is lower than the rate of increase of the access fee offered by N  (multiplied by one). Thus, overall, the number of EUs with N (respectively, NoN)  decreases (respectively, increasing) with $t_{NoN}$.}. Finally, note that the access fees are independent of $t_N$, but the utility of EUs connecting to  N  decreases with $t_N$ \eqref{equation:CP_2}. Thus,  the fraction of EUs with $N$, i.e. $n^{eq}_{N}$,  decreases also with respect to  $t_N$.

\begin{remark}
The access fees in the SPNE candidates (c) and (d) correspond to  Theorem \ref{theorem:NE_stage1_new_q<}-3 and \ref{theorem:NE_stage1_new_q<}-4.
 From their constructions, $\Delta p$ satisfies  2-a-i of Theorem~\ref{theorem:p_tilde_new}. 
  Thus,   $(q^{eq}_N,q^{eq}_{NoN})=(\tilde{q}_f,\tilde{q}_p)\in F^I_1$. Also, by Theorem~\ref{theorem:NE_stage2_new_suff}, 
 $\tilde{p}^{eq}=\tilde{p}_{t,3}=\kappa_{ad}n^{eq}_{NoN}(1-\frac{\tilde{q}_f}{\tilde{q}_p})$.
From $\Delta p=p^{eq}_{NoN}-p^{eq}_N$, and \eqref{equ:EUs_linear}, 
the expressions for $n^{eq}_N$ and $n^{eq}_{NoN}$ follow.
\end{remark}

\textbf{SPNE Candidate (e):}  $p^{eq}_{NoN}=c+\frac{1}{3}(2t_N+t_{NoN})$, $p^{eq}_N=c+\frac{1}{3}(2t_{NoN}+t_N)$, $z^{eq}=0$, $(q^{eq}_N,q^{eq}_{NoN})=(\tilde{q}_f,\tilde{q}_f)\in F^L_0$, $n^{eq}_N=\frac{2t_{NoN}+t_N}{3(t_{NoN}+t_N)}, n^{eq}_{NoN}=\frac{2t_N+t_{NoN}}{3(t_N+t_{NoN})}$, and since $z^{eq}=0$, $\tilde{p}^{eq}$ is of no importance. 

\begin{remark} The access fees are obtained from Theorem~\ref{theorem:neutralz=0_q<}.
 From Theorem~\ref{lemma:CP_z=0_new}-1, $(q^{eq}_N,q^{eq}_{NoN})=(\tilde{q}_f,\tilde{q}_f)\in F^I_0$. From $\Delta p=p^{eq}_{NoN}-p^{eq}_N$, and \eqref{equ:EUs_linear}, the expressions for $n^{eq}_N$ and $n^{eq}_{NoN}$ follow.
\end{remark}

  Candidate (e) is the only possible SPNE  by which $z^{eq}=0$, and  is similar to the benchmark one (Section~\ref{section:summary_benchmark}).  Candidate (c) reduces to (e) if $\tilde{q}_p=\tilde{q}_f$. For (e),  the asymmetries between the ISPs only arise from those of   $t_N$ and $t_{NoN}$. Thus, EUs are divided between ISPs depending on $t_N$ and $t_{NoN}$, and the  access fees are a function of only  $t_N$ and $t_{NoN}$. 
 In contrast, in a non-neutral regime, i.e., when $z^{eq}=1$, the  access fee quoted by ISP NoN  depends on the value of the quality she provides, the premium quality $\tilde{q}_p$.  This is because, in the latter,  NoN differentiates her product from  ISP N by providing premium  quality to the EUs, and can accordingly charge the EUs; while in a neutral regime, both ISPs offer identical quality.

We conclude this subsection with a general observation on all the SPNE candidates.   Intuitively, we expect that high sensitivity of EUs and the CP to the quality, i.e.  large $\kappa_u$ and $\kappa_{ad}$, respectively, yields high payoff for ISP NoN,  since this ISP can provide a premium quality and charge the  EUs and CPs accordingly to increase her payoff. 
Results reveal that in all candidates  NoN charges the CP in proportion to $\kappa_{ad}$.  NoN's subscription revenue from the EUs, $n_{NoN}^{eq} (p_{NoN}^{eq}-c)$ is non-decreasing in $\kappa_u$: 1)  increase in $\kappa_u$ in candidates (a), (b), (c), 2) is $0$ in (d) and 3) does not change with $\kappa_u$ in (e). 

\subsection{Determining the SPNE(s), if any, given parameter values}
\label{determining}
If SPNE exists, it would of the form of candidates (a)-(e) in Section~\ref{section:summaryof resutls}. Given specific values of parameters, the necessary conditions listed in Theorems \ref{theorem:NE_stage1_new_q>}, \ref{theorem:NE_stage1_new_q<}, and \ref{theorem:neutralz=0_q<} can be checked early on to rule out some candidates. Subsequently, one can determine which, if any of the surviving ones is the SPNE, given the parameter values, by checking for  profitable unilateral deviation from each. 
  For $z=1$, the quality choices and the marginal side-payment that maximize the payoffs of the CP, and ISP NoN, respectively,     are uniquely known, from Theorems~\ref{theorem:p_tilde_new} and \ref{theorem:NE_stage2_new_suff}, once the access  fees chosen by the ISPs in Stage $1$ are known. For $z=0,$ the CP offers quality  $\tilde{q}_f$ on each ISP, following the argument in Remark~\ref{r5}, and the marginal side-payment  does not matter. Thus,  we only need to  check for any profitable unilateral deviation  for each  ISP in Stage $1$. Considering each  region of $\Delta p$, characterized in Theorem \ref{theorem:p_tilde_new}, one can identify  potential profitable deviations (using first order condition for some regions, and the monotonicity of  payoffs of ISPs  in other regions). Thus, the search for the best deviation reduces to comparing the payoffs of a finite number of candidate deviations with the payoff of the candidate SPNE.

   Since the CP and ISPs know the values of all the parameters, e.g., $\kappa_u, \kappa_{ad}, c, t_N, t_{NoN}, \tilde{q}_f, \tilde{q}_p$, each can determine if an SPNE exists, and compute the SPNE strategy, should one exist,  for any given parameter values, using the above computation approach,  without any coordination with other players.

\subsection{Does SPNE always exist?}
 \label{existence}The computation approach in Section~\ref{determining} yields that there is no SPNE for instance if $\tilde{q}_f=1$, $\tilde{q}_p=1.5$, $c=1$, $\kappa_u=1$, $\kappa_{ad}=0.5$, $t_N=3$, and $t_{NoN}=2$. In this case, for each  of the candidates listed in Section~\ref{section:summaryof resutls}, there exists a profitable deviation for either ISP N or NoN.  

\section{Benchmark: Both ISPs are Neutral}\label{section:summary_benchmark}\label{section:proofsBenchamrk}

We seek  the SPNE when both ISPs are neutral as a benchmark case. 
This is equivalent to restricting the CP to choose $z^{eq}=0$. 
Recall that in Theorem~\ref{lemma:CP_z=0_new}, we characterize the equilibrium strategies within $F_0$ without  considering the strategies in $F_1$.
Next, note that Stage 2 is  now irrelevant  since with $z^{eq}=0$, the effect of $\tilde{p}$ would be eliminated in all analyses. One can show:  

 \begin{theorem}\label{lemma:NEz=0}
 	Outcome (e) in Section~\ref{section:summaryof resutls} is the unique SPNE when the CP is constrained to choose $z^{eq}=0$.  
 \end{theorem}



If $t_{NoN}\& t_N\rightarrow 0$, outcome (e) continues to be the SPNE in this benchmark case, but not otherwise (refer to Remark~\ref{r2}). In this case,  the access fees for the benchmark case,   $p^{eq}_{NoN,B}\&p^{eq}_{N,B}\rightarrow c$. Thus,  in the absence of inertias, since there is no differentiation between the quality offered by the ISPs in the neutral regime, price competition drives the access fees to the marginal cost. This implies that by removing the inertias ($t_N$ and $t_{NoN}$), the model would be similar to a Bertrand competition \cite{MWG}. Thus, considering the inertias brings a realistic dimension to the model.

\section{Numerical Results}\label{section:numericalresults}

In Section~\ref{section:simulation_findNE},  using the computation strategy in Section~\ref{determining}, we find the SPNE strategies for various parameters. 
  In Section~\ref{section:numerical_more_results},  we provide  insights about $n^{eq}_{NoN}$, $\tilde{p}^{eq}$, and the payoff of ISPs, emerging from the numerical results.	 We assess the benefits of non-neutrality by comparing   with the benchmark case in Section~\ref{section:compare_results}.  In Section~\ref{section:regulation}, we provide regulatory comments based on the results. We refer to SPNE as NE in the figures to preserve brevity of text.

\subsection{SPNE Strategies}\label{section:simulation_findNE}

\begin{figure}[t]
	\centering
	\includegraphics[width=0.35\textwidth]{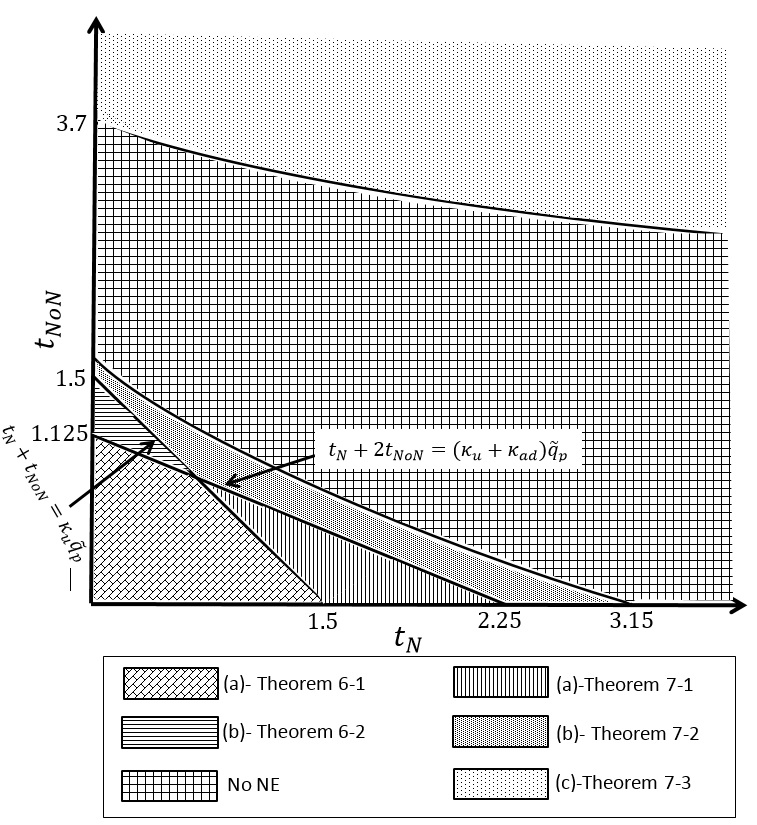}
	\caption{SPNE strategies when $\kappa_u=1$ and $\kappa_{ad}=0.5$. The region corresponding to Theorem~\ref{theorem:NE_stage1_new_q>}-1 indicates that of Theorem~\ref{theorem:NE_stage1_new_q>}-1 other than that of Theorem~\ref{theorem:NE_stage1_new_q<}-1. }
	\label{figure:NE_general_ku>kad}
\end{figure}

\begin{figure}[t]
	\centering
	\includegraphics[width=0.35\textwidth]{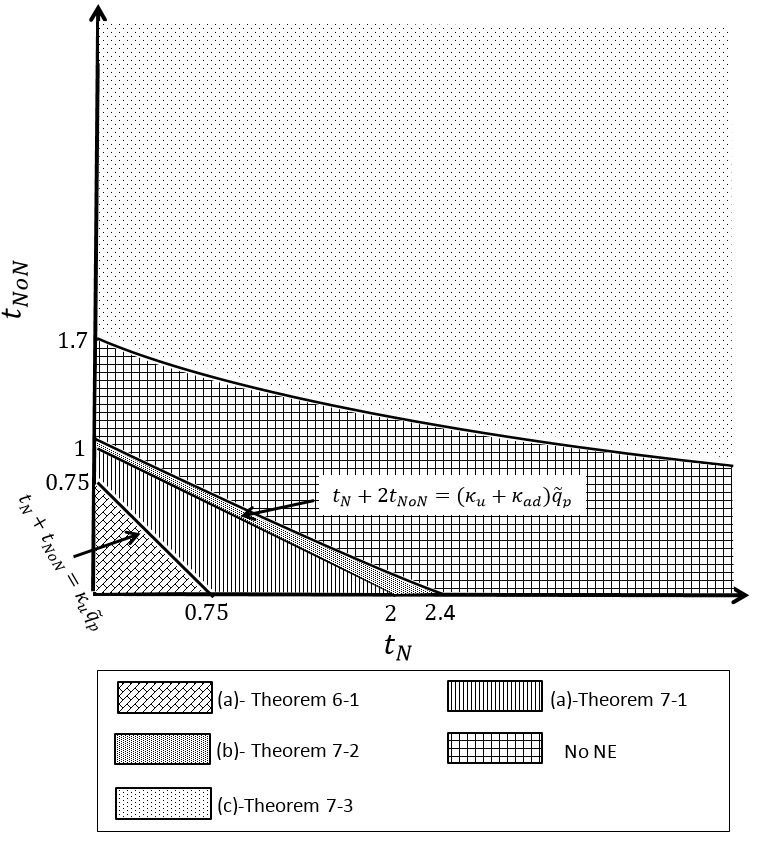}
	\caption{SPNE strategies when $\kappa_u=0.5$ and $\kappa_{ad}=1$. The region corresponding to Theorem~\ref{theorem:NE_stage1_new_q>}-1 indicates that of Theorem~\ref{theorem:NE_stage1_new_q>}-1 other than that of Theorem~\ref{theorem:NE_stage1_new_q<}-1.}
	\label{figure:NE_general_ku<kad}
\end{figure}

%
%
Recall that if SPNE exists, it would of the form of candidates (a)-(e) in Section~\ref{section:summaryof resutls}.
We present the SPNE strategies for different regions of $t_N$ and $t_{NoN}$ and $\tilde{q}_f=1, \tilde{q}_p=1.5, c=1$ and two illustrative examples:  $\kappa_u=1, \kappa_{ad}=0.5$ (Figure~\ref{figure:NE_general_ku>kad}) and  $\kappa_u=0.5, \kappa_{ad}=1$ (Figure~\ref{figure:NE_general_ku<kad}). 
In these, 1) candidate (a) is the SPNE for small $t_N, t_{NoN}$, 2) in the intermediate regions of $t_N, t_{NoN}$ there is no SPNE, 3) as $t_N, t_{NoN}$ further increase, depending on specific regions of these parameters, candidates (b) or (c) become the SPNE.  The SPNE is unique, if it were to exist (in Figures there exists at most one SPNE in each region). Candidates (d) and (e) are never SPNE.
Extensive numerical computations  reveal that the pattern of SPNE strategies for different values of parameters is similar to one of the two patterns presented in Figures  \ref{figure:NE_general_ku>kad} and \ref{figure:NE_general_ku<kad}. Thus, henceforth we do not investigate (d) and (e).


We note that for the parameter values chosen for the figures, $ t_N+2t_{NoN} \leq \tilde{q}_{p}(\kappa_u+\kappa_{ad})$ implies that $t_{NoN} < \kappa_u \tilde{q}_p+\kappa_{ad}(\tilde{q}_p-\tilde{q}_f)$\footnote{This is because the maximum value of $t_{NoN}$ in the first region, $\tilde{q}_{p}(\kappa_u+\kappa_{ad})/2$, is lower than $\kappa_u \tilde{q}_p+\kappa_{ad}(\tilde{q}_p-\tilde{q}_f)$ as $\kappa_{ad} \tilde{q}_f < \tilde{q}_{p}(\kappa_u+\kappa_{ad})/2$.}. Thus, consistent with Remark~\ref{r2}, the numerical results show that candidate (a) is the unique SPNE for $t_N+2t_{NoN} \leq \tilde{q}_{p}(\kappa_u+\kappa_{ad})$. Following Remark~\ref{rlast}, the region corresponding to Theorem~\ref{theorem:NE_stage1_new_q<}-1, is the intersection of that corresponding to Theorem~\ref{theorem:NE_stage1_new_q>}-1 with $\{t_N+t_{NoN} >  \tilde{q}_{p}\kappa_u\}.$  In the figures we mark this region separately from the rest of that corresponding to Theorem~\ref{theorem:NE_stage1_new_q>}-1. 

\subsection{Dependencies of $n^{eq}_{NoN}, n^{eq}_N$, $\tilde{p}^{eq}$, and Payoffs on $t_N$ and $t_{NoN}$}\label{section:numerical_more_results}
We consider $\tilde{q}_f=1, \tilde{q}_p=1.5, c=1, \kappa_u = 0.5, \kappa_{ad} = 1$ throughout in Section~\ref{section:numerical_more_results}.

\subsubsection{ $n^{eq}_{NoN}, n^{eq}_N$} Numerical computations reveal that $n^{eq}_{NoN}$ is non-increasing  (hence, $n^{eq}_{N}$ will be non-decreasing, since $n^{eq}_N=1-n^{eq}_{NoN}$) with respect to both $t_N, t_{NoN}$, e.g.,  Figure \ref{figure:nNoN_double}. For candidate  (a), $n^{eq}_{NoN}=1$. For candidates (b) and (c), different factors determine how $n^{eq}_{NoN}, n^{eq}_N$ depend on $t_N, t_{NoN}$; the overall relation turns out to be monotonic. 

%
%

\begin{figure}
	\begin{subfigure}{.25\textwidth}
		\centering
		\includegraphics[width=\linewidth]{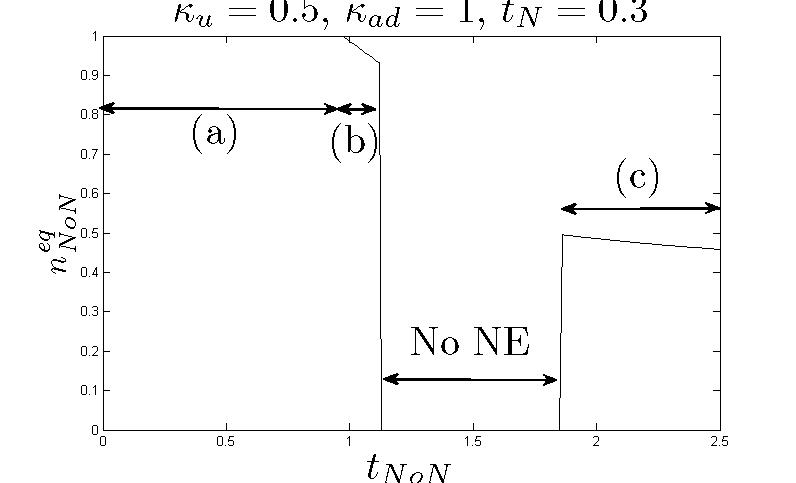}
		\label{fig:sfig1}
	\end{subfigure}%
	\begin{subfigure}{.25\textwidth}
		\centering
		\includegraphics[width=\linewidth]{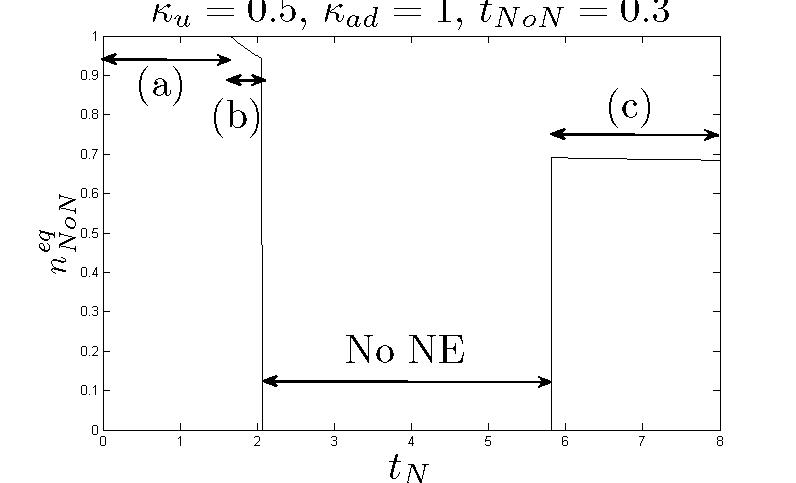}
		\label{fig:sfig2}
	\end{subfigure}
	\caption{$n^{eq}_{NoN}$ with respect to $t_N$ and $t_{NoN}$}\label{figure:nNoN_double}
\end{figure}

\subsubsection{$\tilde{p}^{eq}$}  $\tilde{p}^{eq}$ does not depend on $n^{eq}_{NoN}$ in candidate (a), it is linearly increasing with $n^{eq}_{NoN}$ in both candidates (b) and (c). This, given how $n^{eq}_{NoN}$ is non-increasing   in both $t_N, t_{NoN}$, the nature of the dependence of $\tilde{p}^{eq}$ on both $t_N, t_{NoN}$ in Figures \ref{figure:tildep_kadlessku} may be anticipated from Figure \ref{figure:nNoN_double}. 
  Expressions for candidates (a) and (c) guarantee that $\tilde{p}^{eq}$ is positive for these; numerical results for  a large set of parameters reveal the same for candidate (b) whenever it is SPNE, though the expression is not definitive in this regard. Note that for all candidates, (a) to (e), $\Delta p^{eq}<t_N+\kappa_u$. Thus,  Theorem~\ref{theorem:NE_stage2_new_suff} informs us that  $z^{eq}=1$ provided $\tilde{p}^{eq} > 0$, which holds as per the numerical computations. Thus, $z^{eq}$ turns out to be $1$ for all the SPNEs that occur, namely (a), (b), (c).
%

\begin{figure}
	\begin{subfigure}{.25\textwidth}
		\centering
		\includegraphics[width=\linewidth]{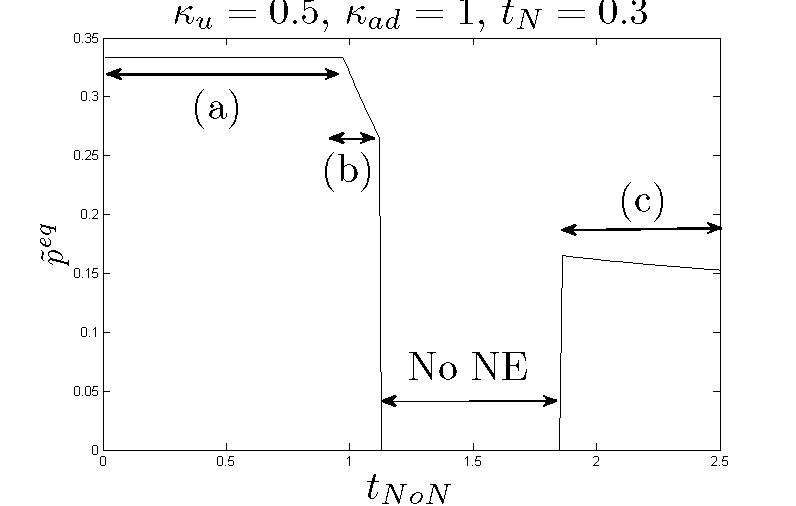}
		\label{figure:tildep_kadlessku_tNoN}
	\end{subfigure}%
	\begin{subfigure}{.25\textwidth}
		\centering
		\includegraphics[width=\linewidth]{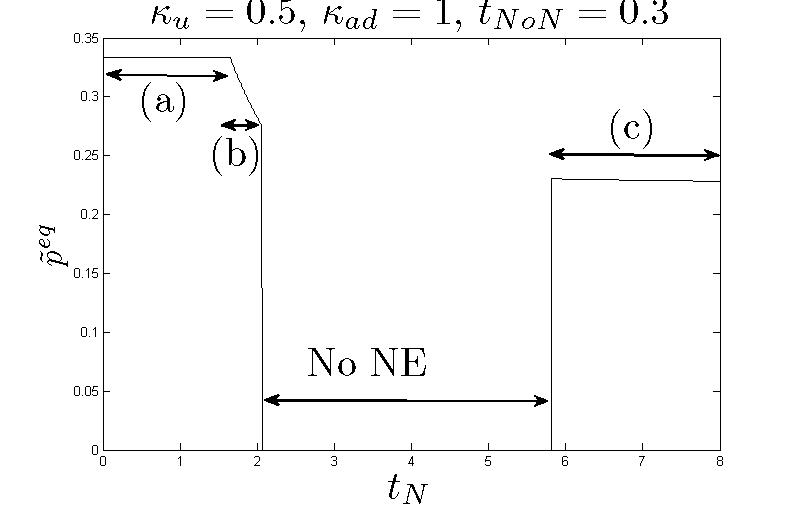}
		\label{figure:tildep_kadlessku_tN}
	\end{subfigure}
	\caption{$\tilde{p}^{eq}$ with respect to $t_N$ and $t_{NoN}$}\label{figure:tildep_kadlessku}
\end{figure}

\subsubsection{Payoffs of ISPs} When the market of power ISP NoN i.e., $\frac{t_{N}}{t_N+t_{NoN}}$, is small, then ISP  NoN's payoff is lower than that of ISP N (Figure~\ref{figure:payoffs_double}-left). Intuitively, we expect  the payoff of an ISP to be  decreasing (increasing, respectively) with respect to the marginal transport cost of that ISP (the competitor ISP). However, this intuition is belied for some candidates and parameter ranges (Figure \ref{figure:payoffs_double}); we explain in the next section, as to why. Expressions for candidate (a) tell us $\pi^{eq}_{NoN} = \kappa_u \tilde{q}_{p}-t_{NoN}, \pi^{eq}_{N} = 0$, which is consistent with Figure \ref{figure:payoffs_double}. Generally,
an ISP's payoff increases with (i) the number of EUs choosing her and  (ii) the  access fee she charges  the EUs. The expressions show that  both ISPs' access fees increase in both $t_N, t_{NoN}$ for candidates (b) and (c).  Numerical computations show that $n^{eq}_{NoN}, n^{eq}_N$ are respectively non-increasing and non-decreasing in $t_N, t_{NoN}$.  Thus, for (b) and (c), N's payoff is increasing with respect to both $t_N, t_{NoN}$, and  
 depending on which of these factors dominates, NoN's payoff  increases, remains constant or decreases with respect to $t_N, t_{NoN}$.


 \begin{figure}
 	\begin{subfigure}{.25\textwidth}
 		\centering
 		\includegraphics[width=\linewidth]{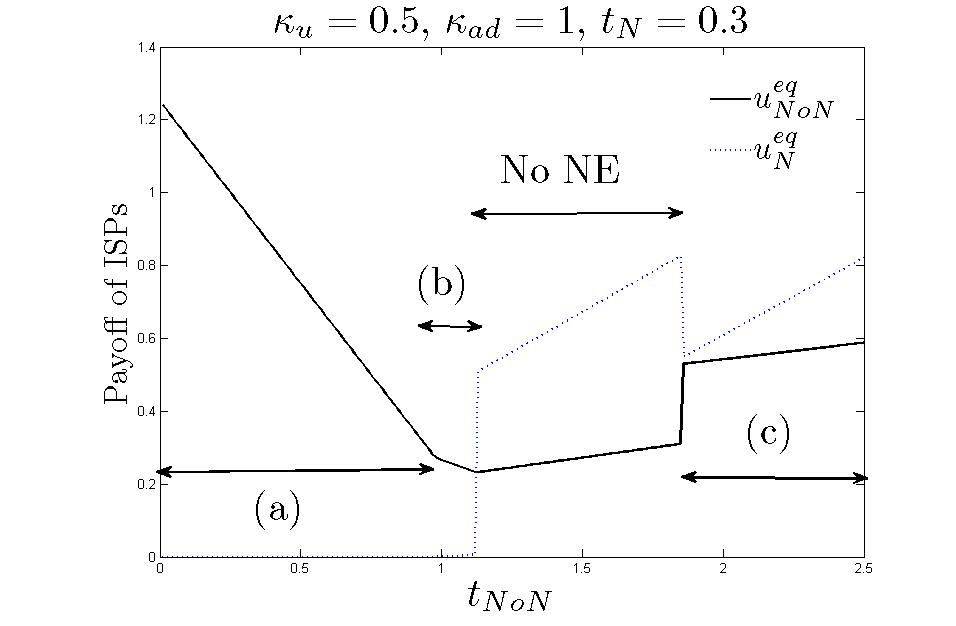}
 	\end{subfigure}%
 	\begin{subfigure}{.25\textwidth}
 		\centering
 		\includegraphics[width=\linewidth]{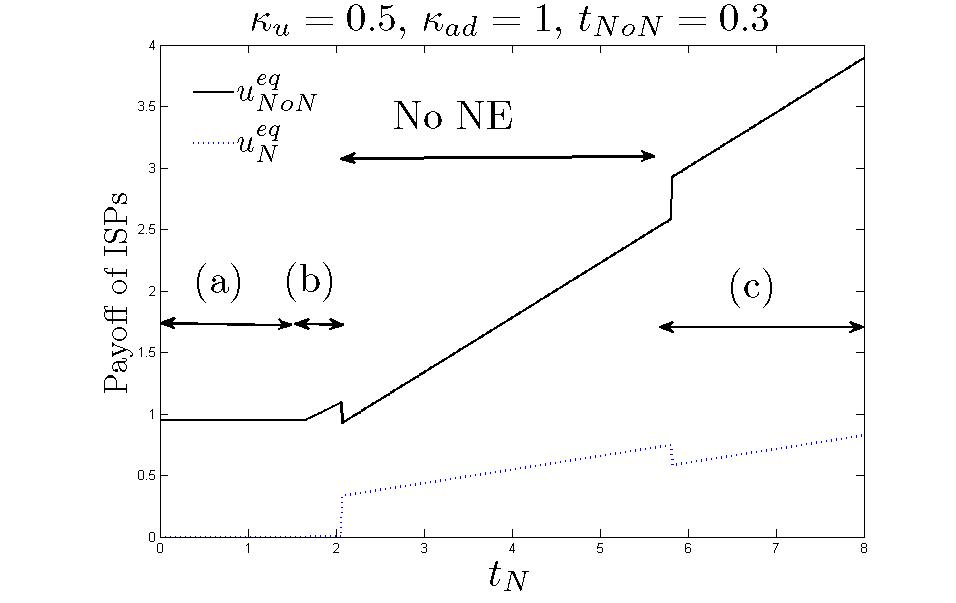}
 	\end{subfigure}
 	\caption{Payoffs of ISPs with respect to $t_N$ and $t_{NoN}$. If  there is no SPNE strategy,  we plot the payoffs in the benchmark case (Section~\ref{section:summary_benchmark}).}\label{figure:payoffs_double}
 \end{figure}





\subsection{Profits of Entities Due to Non-neutrality}\label{section:compare_results}

We compare  the results of the model and the benchmark case in which both ISPs are neutral. We compare   access fees,  payoff of ISPs, the welfare of EUs, and  the payoff of the CP in Sections   \ref{section:comp_access}, \ref{section:comp_CP}, \ref{section:comp_payoffs} and \ref{section:EUW} respectively.

\subsubsection{Access Fees}\label{section:comp_access}
From expressions for candidates (a), (b), (c), (e) and Theorem~\ref{lemma:NEz=0}, ISP N charges lower access fees in (a), (b), (c) as compared to the benchmark case (Theorem~\ref{lemma:NEz=0}). This is to compete with NoN who offers better quality than N on (a), (b), (c)  (per CP's quality choice), but the same free quality  in the benchmark case.

Next,   ISP NoN's access fee in the benchmark case exceeds that in  candidate (a) by $\frac{1}{3}(5t_{NoN}+t_N)-\kappa_u\tilde{q}_p$. This can be negative or positive, and decreases with  $\kappa_u$ and $\tilde{q}_p$, and increases with  $t_{NoN}$ and $t_N$. Thus, if (i) EUs are not sensitive to the quality, i.e. small $\kappa_u$, (ii) the premium quality is not high, i.e. small $\tilde{q}_p$, (iii) EUs cannot switch between ISPs easily, i.e. $t_N, t_{NoN}$ are large,  then NoN provides a cheaper access  for EUs compared to when she is forced to be neutral.

NoN's access fee in the benchmark case exceeds that in  candidate (b) (respectively, (c)) by
$\frac{1}{3}\tilde{q}_p(2\kappa_{ad}-\kappa_u)$ (respectively, $\frac{1}{3}(\tilde{q}_p-\tilde{q}_f)(2\kappa_{ad}-\kappa_u)$). Thus, if $2\kappa_{ad}>\kappa_u$, i.e., the sensitivity of the CP is high enough, then the discount in (b) (respectively, (c)) is positive and increases with the premium quality (respectively, the difference  between the premium and free qualities). On the other hand, if the sensitivity of the CP is low, then the discount is negative. The reason is that  in the former case,  NoN can charge higher marginal side-payments to the CP, and accordingly subsidize  the EUs even though they receive a premium quality. But, in the latter case, the CP would not offer premium quality unless the marginal side-payment is low; hence  NoN needs to charge  the premium quality to the EUs directly, leading to higher  access fees.

\subsubsection{Payoff of the CP}\label{section:comp_CP}  It follows from  \eqref{equ:payoffCP_new} that  the CP gets a payoff of $\kappa_{ad}\tilde{q}_f$ for each candidate SPNE (a)-(e)   and also in the benchmark scenario in Section~\ref{section:summary_benchmark}. This is because  ISP NoN  decides her strategy in Stages 1 and 2, followed by the CP in Stage 3. Thus, CP is the follower, and NoN is a leader in this leader-follower game. Thus, knowing the parameters of the game and the tie-breaking assumption 2 of the CP, NoN can extract all the profits of the CP and render the CP indifferent between offering premium and free qualities on NoN.

\subsubsection{Payoffs of ISPs}\label{section:comp_payoffs}
We plot the differences between the SPNE payoffs of ISPs N and NoN (Section~\ref{section:discussion}) with the corresponding ones in the benchmark case (Section~\ref{section:summary_benchmark}): Figure~\ref{figure:diff_payoffs_NoN_double} for  $\kappa_u=0.5, \kappa_{ad}=1, \tilde{q}_f=1, \tilde{q}_p=1.5, t_{NoN}=0.3$ (intuition remains the same for other values of the parameters).

\emph{Results reveal that  ISP N will lose payoff in all of the non-neutral SPNE strategies}, i.e., those that yield $z^{eq}=1$ (Figure \ref{figure:diff_payoffs_NoN_double}-right). This is because  for candidate  (a),  N would be driven out of the market, while for  (b) and (c),  she has to lower her access fee compared to the benchmark case (Section~\ref{section:comp_access}). 

\emph{Results also reveal that for a wide range of parameters, ISP NoN receives a higher payoff under a non-neutral scenario} (Figure \ref{figure:diff_payoffs_NoN_double}-left). To see why, recall that NoN extracts the additional profit of the CP in a non-neutral scenario. Also,   that for some parameters ($\kappa_u>2\kappa_{ad}$), NoN  charges higher access fee (Section~\ref{section:comp_access}). Even when NoN subsidizes the  access for EUs ($2\kappa_{ad}>\kappa_u$), the deficit is recovered through the side-payment charged to the CP (high $\kappa_{ad}$ yields a high side-payment). Besides,  NoN can potentially attract more EUs by providing cheaper access or  a premium quality (or both). 

However, NoN loses payoff by switching to non-neutrality for some parameter values -  for candidate (a), and when (i) EUs are not sensitive to the quality, i.e., low $\kappa_u$, (ii)  the CP is not sensitive to the quality her EUs receive, i.e., low $\kappa_{ad}$,  (iii) NoN does not offer enough differentiation in the quality, i.e. low $\tilde{q}_p-\tilde{q}_f$,  (iv) NoN's market power,  $\frac{t_{N}}{t_N+t_{NoN}}$,   is low. 
For example, when  $\kappa_u=\kappa_{ad}=0.85, \tilde{q}_f=1, \tilde{q}_p=1.03, t_N=0.05, t_{NoN}=0.8$, then $\pi^{eq}_{NoN}<\pi^{eq}_{NoN,B}$. 
This counter-intuitive phenomenon  results from competition between the ISPs. Recall that N lowers her access fee compared to the benchmark case (Section~\ref{section:comp_access}). To successfully compete for the EUs, then, NoN  also has to decrease her  access fee (exacerbated by her low market power,  and low sensitivity of EUs to the quality), while not generating enough revenue through the side-payments received from the CP (for candidates (a), (b), (c), marginal side-payment is low when $\kappa_{ad}, \tilde{q}_p-\tilde{q}_f$ are low)\footnote{ISP NoN still extracts the additional profit she creates for EUs.}. All these reduce NoN's payoff. A non-neutral regime is unlikely to emerge in these cases. 
Recall that  ISPs N and NoN first decide on the access fees, and then NoN decides on the marginal side-payment in a subsequent stage. We had chosen this order because the access fees are expected to change less frequently than   the marginal side-payment. But, if the order of these two stages were swapped, then by virtue of becoming the only leader of the game,  NoN would never lose payoff by switching to non-neutrality. 

\begin{figure}
 	\begin{subfigure}{.25\textwidth}
 		\centering
 		\includegraphics[width=\linewidth]{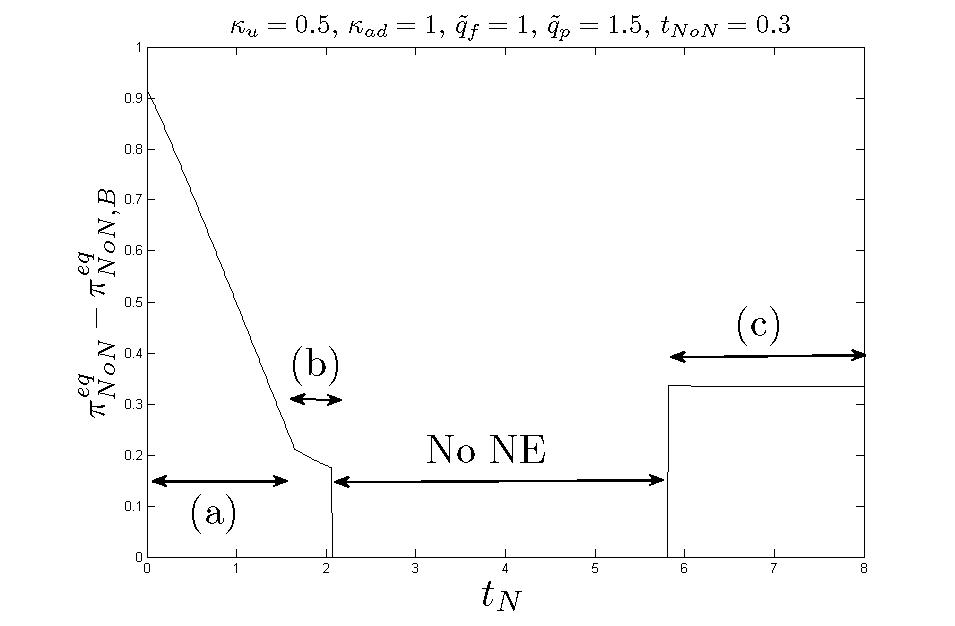}
 	\end{subfigure}%
 	\begin{subfigure}{.25\textwidth}
 		\centering
 		\includegraphics[width=\linewidth]{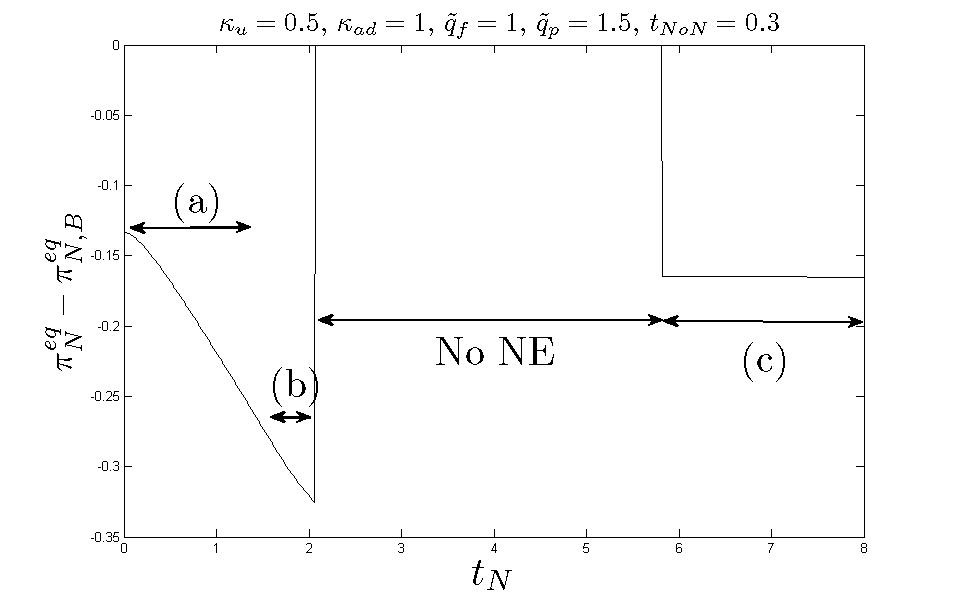}
 	\end{subfigure}
 	\caption{The difference between the payoffs of ISPs with respect to $t_N$ and $t_{NoN}$}\label{figure:diff_payoffs_NoN_double}
 \end{figure}

\subsubsection{EUs'
 Welfare}\label{section:EUW}
 \begin{figure}
 	\begin{subfigure}{.25\textwidth}
 		\centering
 		\includegraphics[width=\linewidth]{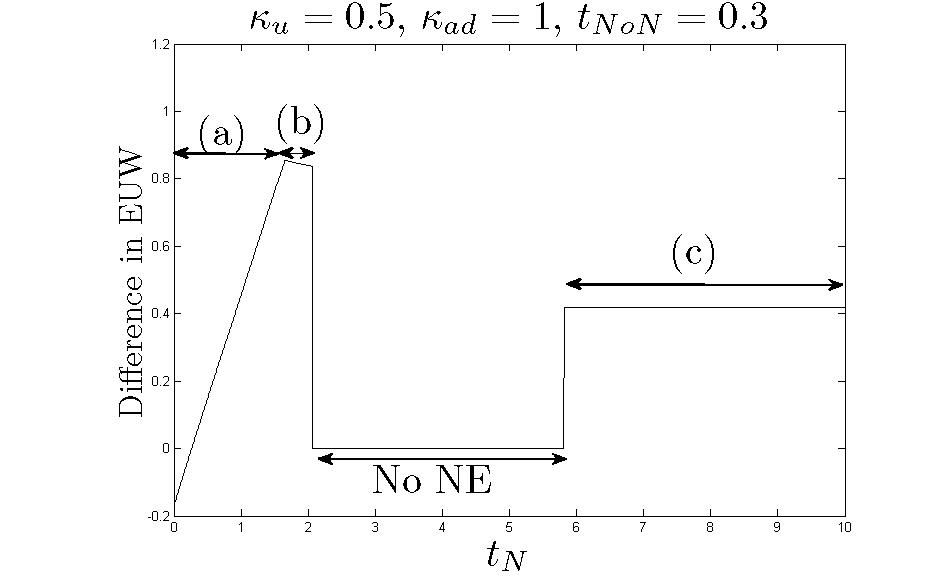}
 	\end{subfigure}%
 	\begin{subfigure}{.25\textwidth}
 		\centering
 		\includegraphics[width=\linewidth]{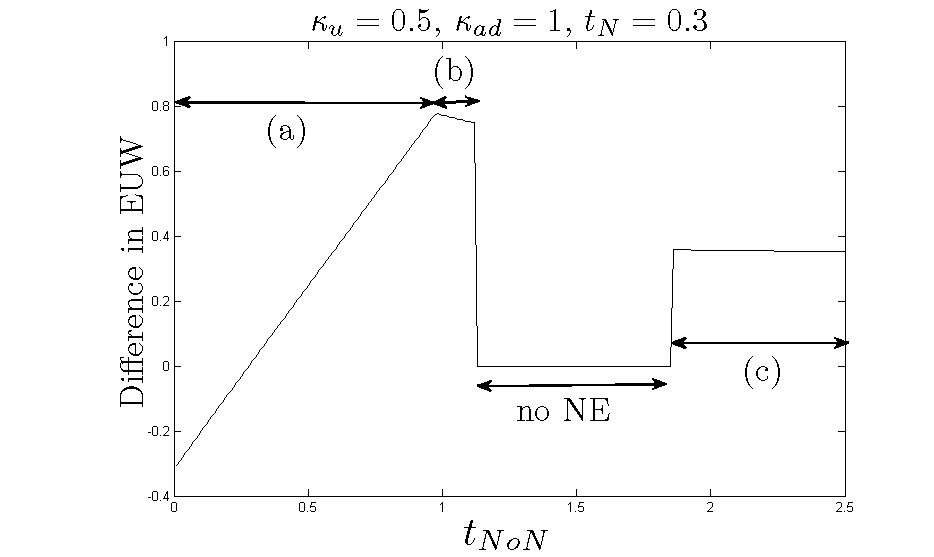}
 	\end{subfigure}
 	\caption{Difference in EUW between the non-neutral and benchmark scenarios with respect to $t_N$ and $t_{NoN}$}\label{figure:EUW_double}
 \end{figure}
 \begin{figure}
 	\begin{subfigure}{.25\textwidth}
 		\centering
 		\includegraphics[width=\linewidth]{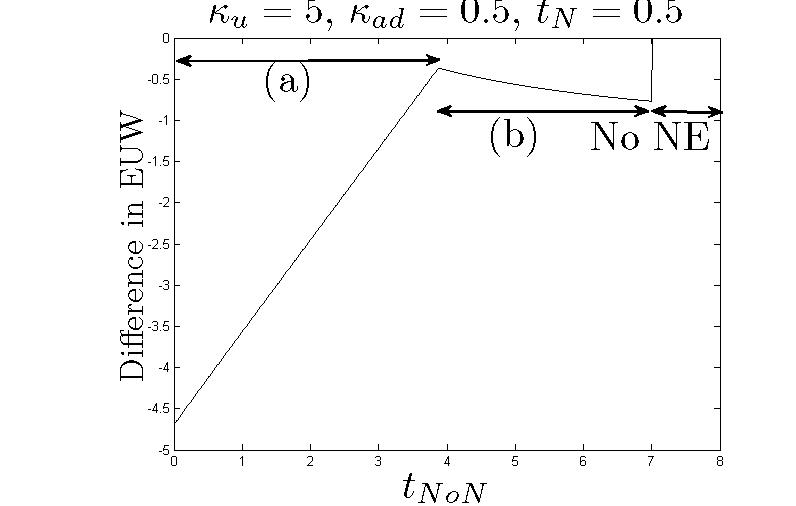}
 	\end{subfigure}%
 	\begin{subfigure}{.25\textwidth}
 		\centering
 		\includegraphics[width=\linewidth]{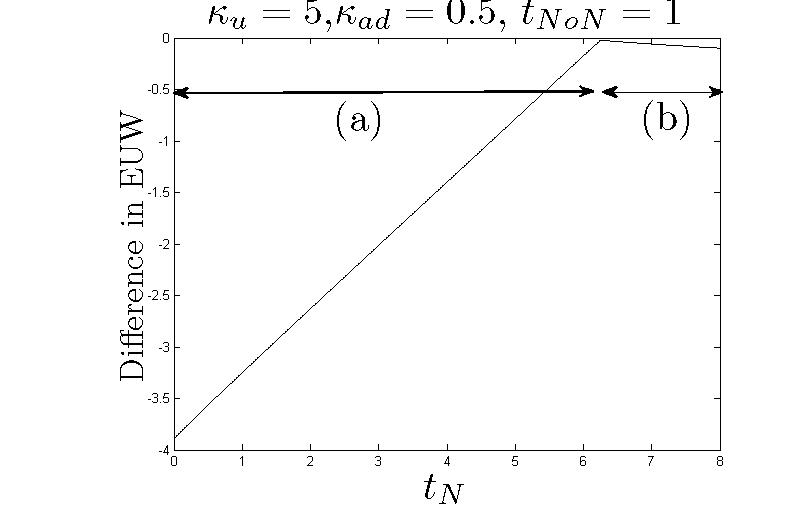}
 	\end{subfigure}
 	\caption{Difference in EUW between the non-neutral and benchmark scenarios with respect to $t_N$ and $t_{NoN}$}\label{figure:EUW_double_kularge}
 \end{figure}
%
%
 %
  Using, \eqref{equation:CP_2},  we define the total welfare of EUs, EUW, as:
{\footnotesize
	\be\label{equ:SW}
	\ba
&\int_0^{x_n} u_{EU,N}(x) dx + \int_{x_n}^1 u_{EU,NoN}(x) dx = \\
	& v^*+ (\kappa_u q_N-p_N)n_N-\frac{t_N}{2}n^2_N\\
	&\qquad \qquad +(\kappa_u q_{NoN}-p_{NoN})n_{NoN}	-\frac{t_{NoN}}{2}n^2_{NoN}
	\ea
	\ee}
\normalsize
 We plot the difference in the EUW of the non-neutral  (Section~\ref{section:discussion}) with the benchmark case (Section~\ref{section:summary_benchmark})  for $\tilde{q}_f=1, \tilde{q}_p=1.5$ (Figures \ref{figure:EUW_double}, \ref{figure:EUW_double_kularge}). Other parameters of the figures are $\kappa_u=0.5$, $\kappa_{ad}=1$, and $t_{NoN}=0.3$ for Figure~\ref{figure:EUW_double}-left, $\kappa_u=0.5$, $\kappa_{ad}=1$, and $t_{N}=0.3$ for Figure~\ref{figure:EUW_double}-right,  $\kappa_u=5$, $\kappa_{ad}=0.5$, and $t_{N}=0.5$ for Figure~\ref{figure:EUW_double_kularge}-left, and $\kappa_u=5$, $\kappa_{ad}=0.5$, and $t_{NoN}=1$ for Figure~\ref{figure:EUW_double_kularge}-right.
We observe that EUW is higher in a non-neutral setting if (i) the market power of ISP NoN $t_N/(t_N + t_{NoN})$ is low, (ii) $\kappa_{ad}$  is high, or (iii) $\kappa_u$ is low. However, when $t_N, t_{NoN}$  are small, and (ii) and (iii) do not hold,   the benchmark case yields a higher EUW (Figure \ref{figure:EUW_double_kularge}). This may be explained by the fact that  these are also the parameter values for which NoN provides  more expensive access in the  non-neutral setting for candidates (a), (b), (c)   (Section~\ref{section:comp_access}). Clearly, many factors influence the difference, and in opposing directions. But that the difference  may either be positive or negative may be expected since at least one of the ISPs, ISP NoN, provides both more expensive and cheaper access to the EUs, for different values of the parameters,  in the benchmark case compared to the non-neutral one. 


\subsection{Does the Market Need to be Regulated?}\label{section:regulation}


ISP N is driven out of the market for candidate SPNE (a), and loses payoff in other SPNEs (namely (b), (c)) as compared to a fully neutral market, i.e., the benchmark scenario (Section~\ref{section:comp_payoffs}).
 A regulator like FCC may provide incentives to the neutral ISPs (eg, monetary subsidies or tax deductions) to retain her in the market. Thereby, the market retains a mixture of  neutral and non-neutral ISPs. Our research has shown that in such markets the non-neutral ISP gets a higher payoff and EUs a higher overall welfare for a wide range of parameters, as compared to in a fully neutral scenario. On the other hand, if all  ISPs are non-neutral, then few powerful CPs may acquire inordinate control over the market by offering premium quality on all ISPs, or  one non-neutral ISP may become a monopoly driving the rest out of the market.  



\section{Generalization and Future Research}\label{section:implicationAssum}\label{future}

If the EUs are non-uniformly distributed in $[0, 1]$, we start with a cumulative distribution function $F(x)$ which gives the fraction of EUs in $(0, x)$.  In \eqref{equ:EUs_linear}, for $x_n \in (0, 1)$,  $n_N$ will now be $F(x_n)$, where $x_n$ is given by \eqref{equ:xn}. The results that do not depend on the specific expression for $n_N$ extend (eg, those in Sections~\ref{section:CPdecides}, \ref{section:stage2}), but others must be derived using specific expressions for $F(\cdot)$ (eg, those in Sections~\ref{section:stage1}, \ref{section:summaryof resutls}, \ref{section:summary_benchmark}). 


We may consider that  NoN incurs an additional marginal cost $\kappa$   each time she offers CP's content with premium quality, since NoN could presumably have utilized the associated bandwidth elsewhere.    
 Then   $\pi_{NoN}(p_{NoN}, \tilde{p})$ in \eqref{equ:payoffISPsGeneral_new} must be decremented by $z\kappa\tilde{q}_p$, and  (a)-(e) of Section~\ref{section:summaryof resutls} would still constitute the possible SPNEs, as payoffs for all non-neutral choices would be reduced by a fixed constant $\kappa\tilde{q}_p$ compared to that for the neutral choices. For small $\kappa$,  all results remain the same.  As $\kappa$ increases, for larger ranges of parameters, 1) the SPNE candidate (e) becomes the actual SPNE  even outside the Benchmark case,  and 2)   NoN  loses payoff by switching to a non-neutral regime. When $\kappa$ becomes large, (e) becomes the unique SPNE. As to specific results, Theorems~\ref{lemma:CP_z=0_new}, \ref{theorem:p_tilde_new}, \ref{theorem:neutralz=0_q<}-1, \ref{theorem:NE_stage2_new_suff} hold as is.  
  The proofs in  Theorems~\ref{lemma:CP_z=0_new}, \ref{theorem:p_tilde_new}, \ref{theorem:NE_stage2_new_suff}, \ref{theorem:neutralz=0_q<}-1 do not use expressions for NoN's payoff. Theorems~\ref{theorem:neutralnotexists_q>}-2, \ref{theorem:NE_stage1_new_q>} and \ref{theorem:bigt} hold only for small $\kappa.$ But, the SPNEs in Theorem~\ref{theorem:NE_stage1_new_q>} and \ref{theorem:NE_stage1_new_q<} remain possible SPNEs regardless of the value of $\kappa$, and as before Theorem~\ref{theorem:NE_stage1_new_q<} provides all possible SPNEs. In each part, additional necessary condition  $\pi_{NoN}(\tilde{p}^{eq}_{NoN},\tilde{p}_t)>\pi_{NoN,z=0}(\tilde{p}^{eq}_{NoN},\tilde{p})$ applies (which is naturally satisfied for small $\kappa$). These follow from a perusal of the corresponding proofs.  
 Theorem~\ref{lemma:NEz=0} holds because the proof only considers that $z = 0.$ In Appendix~\ref{appendix:fixedcost}, we outline the proof 
 that (e) is the unique SPNE when $\kappa$ is large.

 In Section~\ref{section:comp_CP} we observed that the CP receives a payoff of $\kappa_{ad} \tilde{q}_f$ regardless of her choice of quality on the ISPs. In practice, however, monopolistic CPs (eg, Google)  earn large advertising revenues, which would presumably increase if they are allowed to control the qualities of their transmissions on the ISPs. Our model can incorporate the same with a minor modification, that in which, there exists an upper-bound $\alpha > 0$ on the marginal side-payment that NoN can not exceed. This upper-bound will be decided through prior negotiations between the CP and ISPs, and its value will depend on the popularity of the CP in the market. If $\alpha$ is small  (which is likely to happen when the CP is popular among EUs), the CP will pay a low side-payment to the ISP, and thereby retain most of her advertising revenue. Theorems~\ref{lemma:CP_z=0_new}, \ref{theorem:p_tilde_new} hold as is. A modified threshold $\tilde{p}_t^{mod} = \min(\tilde{p}_t, \alpha)$ will need to be defined.    Theorem~\ref{theorem:NE_stage2_new_suff}-1 holds with $\tilde{p}_t^{mod}$ replacing $\tilde{p}_t$. The SPNEs in Theorem~\ref{theorem:NE_stage1_new_q>} and \ref{theorem:NE_stage1_new_q<} remain possible SPNEs, but a few  other possible SPNEs arise. We outline their derivation, as also the proof of a modified version of Theorem~\ref{theorem:NE_stage2_new_suff},  in Appendix~\ref{appendix:upperbound}.


Assessing the predicaments of the stake-holders in a potentially non-neutral Internet market with $m$ ISPs and $n$ CPs remain open for $m > 1, n> 2.$ These intensify the competition between the ISPs and introduce competition between the CPs (eg, 1) Netflix, Hulu and Amazon 2) Facebook, Snapchat Twitter, compete respectively  for delivering movies and social media exchanges). The competition between the CPs  would for example decrease their market power and profit margins; these may have opposite effects on the acceptance of high marginal side-payments. Internet may also be fragmented if each CP offers her content exclusively or with a better quality on a different ISP to decrease competition. 

 Formulation and solution of a repeated interaction game  in which the CP and the ISPs determine their strategies based on prior choices and corresponding payoffs remain open, eg, if N is driven out of the market, corresponding to the parameter choices in which SPNE candidate (a) emerges as the SPNE, should she upgrade to acquire the capability to deliver CP's content with premium quality and reappear at a later time? How do ISPs strategize, eg, invoking predatory pricing,  in presence of repeated interactions?

The ISPs can influence some of the factors that determine the marginal transport costs, $t_N, t_{NoN}$,  but only over a long time-horizon. Generalizing the formulation and the analysis to consider these as decision variables remain open.


\bibliographystyle{IEEEtran}
\bibliography{bmc_article}

\appendices

\section{Proofs of Section~\ref{section:CPdecides}- Stage 3}\label{section:Stage3_proofs}

We start from Remark~\ref{r5} that identifies the candidates  for  $(q^{eq}_N,q^{eq}_{NoN})$ that had not been ruled out until the Remark.

Next, if $0<n_{NoN}<1$, then the expression for the payoff in \eqref{equ:payoffCP_new}, would be (using~\eqref{equ:EUs_linear}):
\be\label{equ:payoffCP_new_2}
\footnotesize
\ba
&\pi_{CP}(q_N,q_{NoN},z)=\frac{t_{NoN}+\kappa_u(q_N-q_{NoN})+p_{NoN}-p_N}{t_N+t_{NoN}} \kappa_{ad} q_N\\
&+\frac{t_N+\kappa_u(q_{NoN}-q_N)+p_N-p_{NoN}}{t_N+t_{NoN}}\kappa_{ad} q_{NoN}-z \tilde{p} q_{NoN}
\ea
\ee
\normalsize

The following lemmas are used in proving the main results of this section:

\begin{lemma} \label{lemma:thresh_on_delta_p_z1_new}
	Let $(\tilde{q}_f,\tilde{q}_{p})$ and $(0,\tilde{q}_{p})$ belong to the set $F^I$, i.e. for them $0< x_N< 1$. Then $\pi_G(\tilde{q}_f, \tilde{q}_{p},z=1) \geq \pi_G(0, \tilde{q}_{p},z=1)$ if and only if $\Delta p\geq \Delta p_{t}$, where $\Delta p_{t}=\kappa_{u}(2\tilde{q}_{p}-\tilde{q}_f)-t_{NoN}$.
\end{lemma}
\begin{proof}
	The proof is done by comparing the payoffs (note that in both cases $0< x_N<1$). We use \eqref{equ:payoffCP_new_2} to write the expression of $\pi_G({q}_N,{q}_{NoN},z)$:
	$$
	\ba
	&\pi_G(\tilde{q}_f, \tilde{q}_{p},z=1) \geq \pi_G(0, \tilde{q}_{p},z=1) \iff \\
	& \big{(}t_{NoN}-\kappa_{u}(\tilde{q}_{p}-\tilde{q}_f)+p_{NoN}-p_{N}\big{)} \kappa_{ad}\tilde{q}_f+\\
	&\qquad \qquad \qquad 	\big{(} t_{N}+\kappa_{u}(\tilde{q}_{p}-\tilde{q}_f)+p_{N}-p_{NoN}\big{)}\kappa_{ad} \tilde{q}_{p} \\
	&\geq \big{(} t_{N}+\kappa_{u}\tilde{q}_{p}+p_{N}-p_{NoN}\big{)} \kappa_{ad}  \tilde{q}_{p}\\
	&\qquad \iff t_{NoN}-\kappa_{u}(2\tilde{q}_{p}-\tilde{q}_f)+p_{NoN}-p_{N} \geq 0 \\
	&\qquad\iff \Delta p \geq \kappa_{u}(2\tilde{q}_{p}-\tilde{q}_f)-t_{NoN}=\Delta p_{t}
	\ea
	$$
	The result follows.
\end{proof}

\begin{lemma} \label{lemma:thresh_on_delta_p_z10_new_1}
	Let $(0,\tilde{q}_p)\in F^L_1$, i.e. by which $n_{NoN}=1$. Then, $\pi_{CP}(0, \tilde{q}_p,z=1) \geq \kappa_{ad} \tilde{q}_f$ if and only if $\tilde{p}\leq \tilde{p}_{t,1}$, where $ \tilde{p}_{t,1}=\kappa_{ad}(1-\frac{\tilde{q}_{f}}{\tilde{q}_p})$.
\end{lemma}

\begin{proof}
	We use \eqref{equ:payoffCP_new} to write the expression of the payoff of the CP:
	$$
	\small
	\ba
	\pi_{CP}(0, \tilde{q}_p,z=1) \geq \kappa_{ad} \tilde{q}_f &\iff \kappa_{ad}\tilde{q}_p-\tilde{p}\tilde{q}_p\geq \kappa_{ad}\tilde{q}_f \\
	&\iff \tilde{p}\leq \kappa_{ad}(1-\frac{\tilde{q}_{f}}{\tilde{q}_p})=\tilde{p}_{t,1}
	\ea
	$$
	\normalsize
\end{proof}

\begin{lemma} \label{lemma:thresh_on_delta_p_z10_new_2}
	Let $(0,\tilde{q}_p)\in F^I_1$, i.e. by which $0< n_{NoN}< 1$. Then, $\pi_{CP}(0, \tilde{q}_p,z=1) \geq \kappa_{ad} \tilde{q}_f$ if and only if $\tilde{p}\leq \tilde{p}_{t,2}$, where $\tilde{p}_{t,2}=\kappa_{ad} (n_{NoN}-\frac{\tilde{q}_f}{\tilde{q}_p})$ and $n_{NoN}=\frac{t_N+\kappa_u \tilde{q}_p-\Delta p}{t_N+t_{NoN}}$.
\end{lemma}

\begin{proof}
	We compare the payoff with $\kappa_{ad}\tilde{q}_f$. We use \eqref{equ:payoffCP_new} to write the expression of the payoff of the CP:
	$$
	\small
	\ba
	\pi_{CP}(0, \tilde{q}_p,z=1) \geq \kappa_{ad} \tilde{q}_f &\iff n_{NoN}\kappa_{ad}\tilde{q}_p -\tilde{p}\tilde{q}_p\geq \kappa_{ad} \tilde{q}_f \\
	& \iff \tilde{p}\leq \kappa_{ad} (n_{NoN}-\frac{\tilde{q}_f}{\tilde{q}_p})=\tilde{p}_{t,2}
	\ea
	\normalsize
	$$
	where, by \eqref{equ:EUs_linear}, $n_{NoN}=\frac{t_N+\kappa_u \tilde{q}_{p}-\Delta p}{t_N+t_{NoN}}$. The result follows.
\end{proof}

\begin{lemma} \label{lemma:thresh_on_delta_p_z10_new_3}
	Let $(\tilde{q}_f,\tilde{q}_{p})\in F^I_1$, i.e. by which $0< n_{NoN}< 1$. Then, $\pi_G(\tilde{q}_f, \tilde{q}_{p},z=1) \geq \kappa_{ad} \tilde{q}_f$ if and only if $\tilde{p}\leq \tilde{p}_{t,3}$, where $ \tilde{p}_{t,3}= \kappa_{ad}n_{NoN}(1-\frac{\tilde{q}_f}{\tilde{q}_{p}})$ and $n_{NoN}=\frac{t_N+\kappa_u (\tilde{q}_{p}-\tilde{q}_f)-\Delta p}{t_N+t_{NoN}}$.
\end{lemma}

\begin{proof}
	We compare the payoff with $\kappa_{ad}\tilde{q}_f$. We use \eqref{equ:payoffCP_new} to write the expression of the payoff of the CP:
	$$
	\ba
	\pi_G(\tilde{q}_f, \tilde{q}_{p},z=1) \geq \kappa_{ad} \tilde{q}_f &\iff (1-n_{NoN}) \kappa_{ad} \tilde{q}_f\\
	&\qquad +n_{NoN}\kappa_{ad} \tilde{q}_{p}- \tilde{p} \tilde{q}_{p} \geq \kappa_{ad} \tilde{q}_f\\
	&\iff \tilde{p}\leq \kappa_{ad}n_{NoN}(1-\frac{\tilde{q}_f}{\tilde{q}_{p}})=\tilde{p}_{t,3}
	\ea
	$$
	where, by \eqref{equ:EUs_linear}, $n_{NoN}=\frac{t_N+\kappa_u (\tilde{q}_{p}-\tilde{q}_f)-\Delta p}{t_N+t_{NoN}}$. The result follows.
\end{proof}

\begin{remark}
	The values of $\Delta p_t$, $\tilde{p}_{t,1}$, $\tilde{p}_{t,2}$, and $\tilde{p}_{t,3}$ characterized in the above lemmas are used in Definition \ref{def:pt1,pt2}.
\end{remark}

We should distinguish between the solutions that maximize \eqref{equ:payoffCP_new}, i.e. $(q^*_N,q^*_{NoN})$ which is not unique, and the strategy that is chosen by the CP in the equilibrium, which is a unique choice among the optimum solutions. Thus, we denote the equilibrium strategy of the CP by $(q^{eq}_N,q^{eq}_{NoN})$, which subsequently yields the equilibrium fraction of EUs with each ISP, i.e. $x^{eq}_N$, $N^{eq}_N$, and $N^{eq}_{NoN}$.

Now, by comparing the payoffs of the candidate solutions and using tie-breaking assumptions,  we prove one of the main results of this section, Theorem~\ref{lemma:CP_z=0_new}:


\begin{proof}
	Note that an equilibrium strategy, i.e. $(q^{eq}_N,q^{eq}_{NoN})$, should be a global maxima of ~\eqref{equ:payoffCP_new}. Suppose $(q^*_N,q^*_{NoN})\in F_0$. First, in Part A, we separate the cases that $(q^*_N,q^*_{NoN})$ is in $F^L_0$, $F^I_0$, or $F^U_0$, characterize the  candidate optimum strategy, i.e. $(q^*_N,q^*_{NoN})$, chosen by the CP in each of these subsets\footnote{Note that $F^L_0\cup F^I_0\cup F^U_0=F_0$.}, and identify the necessary condition on $\Delta p$ for each of these candidate optimums to be in a particular subset. In Part B, we summarize the candidate optimum strategies. Finally, in Part C, by comparing the payoffs of the candidate strategies in different regions of $\Delta p$ and using the tie-breaking assumptions, we characterize the equilibrium strategies.
	
	\textbf{Part A:} First, consider $F^I_0$. If $(q^*_N,q^*_{NoN})\in F^I_0$, i.e. $z^*=0$, then $(q^*_N,q^*_{NoN})$,  by \eqref{equ:summarize_CP_candidate_new}, is $(a)$  $(0,\tilde{q}_{f})$, or $(b)$ $(\tilde{q}_f,0)$, or $(c)$ $(\tilde{q}_f,\tilde{q}_{f})$. Note that the necessary and sufficient condition for each of these candidate outcomes to be in $F^I_0$ is $\frac{\Delta p-t_N}{\kappa_u}< \Delta q^*< \frac{\Delta p + t_{NoN}}{\kappa_u}$ (Table~\ref{table:subsets}). First consider case (a).  Note that $\Delta q^*=\tilde{q}_f$. Thus, the necessary and sufficient condition for $(a)$ to be in $F^I_0$ becomes $\frac{\Delta p-t_N}{\kappa_u}< \tilde{q}_{f}< \frac{\Delta p + t_{NoN}}{\kappa_u}$, which yields $\kappa_u \tilde{q}_{f}-t_{NoN}   < \Delta p <
	\kappa_u \tilde{q}_{f}+t_N $. Similarly, For  cases (b), the necessary and sufficient condition is $
	-\kappa_u \tilde{q}_f-t_{NoN} < \Delta p< -\kappa_u \tilde{q}_f +t_N$ , and for (c) is $-t_{NoN}< \Delta p < t_{N}$.

	Now, consider $F^L_0$. If $(q^*_N,q^*_{NoN})\in F^L_0$, then   $(q^*_N,q^*_{NoN})$,  by \eqref{equ:summarize_CP_candidate_new}, is (d)  $(0,\tilde{q}_{f})\in F^L_0$. Note that, using the condition in Table~\ref{table:subsets}, the necessary and sufficient condition for $(0,\tilde{q}_{f})\in F^L_0$ is  $\Delta p \leq \kappa_u \tilde{q}_{f}-t_{NoN}$.
	
	Finally, consider $F^U_0$. If $(q^*_N,q^*_{NoN})\in F^U_0$, then $(q^*_N,q^*_{NoN})$,  by \eqref{equ:summarize_CP_candidate_new}, is (e) $(\tilde{q}_{f},0)\in F^U_0$. Using the condition in Table~\ref{table:subsets}, the necessary and sufficient condition for $(\tilde{q}_{f},0)\in F^U_0$ is  $\Delta p \geq t_N -\kappa_u \tilde{q}_f$.

	\textbf{Part B:} Note that, as mentioned before, the strategy that is chosen by the CP in the equilibrium is a unique choice among the possible optimum solutions. Thus, if $(q^{eq}_N,q^{eq}_{NoN})\in F_0$, then $(q^{eq}_N,q^{eq}_{NoN})$ is of the form of \emph{one} of the followings (the necessary condition for each to be optimum is also listed):
	\begin{description}
		\item [(a)]  $(0,\tilde{q}_{f})\in F^I_0$, if this is overall optimum then $\kappa_u \tilde{q}_{f}-t_{NoN}  < \Delta p < \kappa_u \tilde{q}_{f}+t_N $  (the necessary condition).
		\item [(b)]  $(\tilde{q}_f,0)\in F^I_0$, the necessary condition: $
		-\kappa_u \tilde{q}_f-t_{NoN} < \Delta p < -\kappa_u \tilde{q}_f +t_N$.
		\item [(c)]  $(\tilde{q}_f,\tilde{q}_{f})\in F^I_0$, the necessary condition: $-t_{NoN}< \Delta p < t_{N}$.
		\item [(d)]  $(0,\tilde{q}_{f})\in F^L_0$, the necessary condition: $\Delta p \leq \kappa_u \tilde{q}_{f}-t_{NoN}$.
		\item [(e)] $(\tilde{q}_{f},0)\in F^U_0$, the necessary condition:  $\Delta p \geq -\kappa_u \tilde{q}_f+t_N$.
	\end{description}
	
	\textbf{Part C:} Now, we compare the payoffs of the CP at each candidate solutions, and use tie-breaking assumptions whenever needed to get the equilibrium strategies of the CP. The payoff of the CP, for each candidate solution, is as follows (by \eqref{equ:payoffCP_new}):
	\be \label{equ:payoffCPinCandidates_new}
	\ba
	\pi_{CP,(a)}&= n_{NoN}\kappa_{ad}\tilde{q}_{f} \quad \text{  \& } 0< n_{NoN}< 1\\
	\pi_{CP,(b)}&= n_{N}\kappa_{ad}\tilde{q_{f}} \quad \text{ \& } 0< n_{N}< 1\\
	\pi_{CP,(c)}&= \kappa_{ad}\tilde{q_{f}} \\
	\pi_{CP,(d)}&= \kappa_{ad}\tilde{q_{f}} \\
	\pi_{CP,(e)}&= \kappa_{ad}\tilde{q_{f}}
	\ea
	\ee
	
	Next, we characterize the equilibrium strategies in different intervals of $\Delta p$. First consider $-t_{NoN}< \Delta p< t_N$. Note that in this case, $\Delta p$ satisfies the necessary condition of (c) being a candidate strategy, and also the necessary and sufficient condition of (c) being in $F_0^I$.  In addition, $\pi_{CP,(c)}>\pi_{CP,(a)}$ and  $\pi_{CP,(c)}>\pi_{CP,(b)}$. Thus, (a) and (b)  cannot be overall optimum solutions. Moreover, $\pi_{CP,(c)}=\pi_{CP,(d)}$ and  $\pi_{CP,(c)}=\pi_{CP,(e)}$. Using tie-breaking assumption \ref{assumption:tie} yields that the CP prefers (c) to (d) and (e). Thus,  $(\tilde{q}_f,\tilde{q}_f)\in F^I_0$ is chosen as the equilibrium strategy in this interval, and case 1 of the lemma follows.
	
	Now, consider $\Delta p\geq t_N$. Note that in this case, $\Delta p$ satisfies the necessary condition of (e) being a candidate strategy, and also the necessary and sufficient condition of (e) to be in $F_0^U$. In addition, this condition rules out (b) and (c). However, for certain intervals of $\Delta p \geq t_N$, the necessary condition of candidate strategies (a) and (d) can be satisfied. We now compare the payoff of (e) to (a) and (d). First note that $\pi_{CP,(e)}>\pi_{CP,(a)}$. Thus candidate strategy (a) can be discarded. Also, $\pi_{CP,(e)}=\pi_{CP,(d)}$. Since $\Delta p=p_{NoN}-p_N\geq t_N>0$,\footnote{Note that $p_N$ and $p_{NoN}$ are Internet access fees.} and by using tie-breaking assumption~\ref{assumption:tie_lowerp}, candidate strategy (e), i.e. $(\tilde{q}_f,0)\in F^U_0$ is chosen as the equilibrium strategy in this interval by the CP. Thus, case 2 of the lemma follows.
	
	Finally, consider $\Delta p\leq -t_{NoN}$. Note that in this case, $\Delta p$ satisfies the necessary condition of  (d) to be a candidate strategy, and also the necessary and sufficient condition of (d) to be in $F_0^L$. In addition, this condition rules out (a) and (c). However, for certain intervals of $\Delta p \leq -t_{NoN}$, the necessary condition of candidate strategies (b) and (e) can be satisfied. We now compare the payoff of (d) to (b) and (e). First note that $\pi_{CP,(d)}>\pi_{CP,(b)}$. Thus candidate strategy (b) can be discarded. Also, $\pi_{CP,(d)}=\pi_{CP,(e)}$. Since $\Delta p=p_{NoN}-p_N\leq -t_{NoN}<0$, and by using tie-breaking assumption~\ref{assumption:tie_lowerp}, candidate strategy (d), i.e. $(0,\tilde{q}_f)\in F^L_0$ is chosen as the equilibrium strategy in this interval by the CP. Thus, case 3 of the lemma follows.
	
	Note that by \eqref{equ:payoffCPinCandidates_new}, $\pi_{CP,(a)}=\pi_{CP,(b)}=\pi_{CP,(c)}=\kappa_{ad}\tilde{q}_f$ and these are all the candidate solutions.Thus, the utility of the CP by each candidate equilibrium strategy would be $\kappa_{ad} \tilde{q}_f$. The result follows.
\end{proof}

Now, we focus on characterizing the candidate strategies and the necessary condition for each of them when $z^{eq}=1$.

\begin{theorem}\label{lemma:summarycandidatesz=1_new}
	If $(q^{eq}_N,q^{eq}_{NoN})\in F_1$, then $(q^{eq}_N,q^{eq}_{NoN})$ is of the form of one of the followings:
	\begin{description}
		\item [(a)]  $(0,\tilde{q}_{p})$, the necessary condition: $\kappa_u \tilde{q}_{p}-t_{NoN}  < \Delta p < \kappa_u \tilde{q}_{p}+t_N$. In addition, $(0,\tilde{q}_{p})\in F^I_1$ if and only if  $\kappa_u \tilde{q}_{p}-t_{NoN}  < \Delta p < \kappa_u \tilde{q}_{p}+t_N$.
		\item [(b)]  $(\tilde{q}_f,\tilde{q}_{p})$, the necessary condition: $\kappa_u (\tilde{q}_{p}-\tilde{q}_f)-t_{NoN}< \Delta p < \kappa_u (\tilde{q}_{p}-\tilde{q}_f)+t_{N}$. In addition, $(\tilde{q}_f,\tilde{q}_{p})\in F^I_1$ iff  $\kappa_u (\tilde{q}_{p}-\tilde{q}_f)-t_{NoN}< \Delta p < \kappa_u (\tilde{q}_{p}-\tilde{q}_f)+t_{N}$.
		\item [(c)]  $(0,\tilde{q}_{p})$, the necessary condition: $\Delta p \leq \kappa_u \tilde{q}_{p}-t_{NoN}$. In addition, $(0,\tilde{q}_{p})\in F^L_1$  iff  $\Delta p \leq \kappa_u \tilde{q}_{p}-t_{NoN}$.
	\end{description}
\end{theorem}

\begin{proof}
	Suppose $(q^*_N,q^*_{NoN})\in F_1$. We separate the cases that $(q^*_N,q^*_{NoN})$ is in $F^L_1$, $F^I_1$, or $F^U_1$, characterize the candidate  optimum  solutions chosen by the CP in each of these subsets, and identify the necessary condition on $\Delta p$ for each of these candidate optimum strategies to be in a particular subset.
	
	
	Note that by~\eqref{equ:summarize_CP_candidate_new}, no optimum strategy is chosen in the set $F^U_1$. Thus, we characterize the optimum strategies chosen in $F^I_1$ and $F^L_1$ by the CP.
	
	
	Now, consider $F^I_1$. By~\eqref{equ:summarize_CP_candidate_new}, if $(q^*_N,q^*_{NoN})\in F^I_1$, then $(q^*_N,q^*_{NoN})$ is $(a)$  $(0,\tilde{q}_{p})$ or $(b)$ $(\tilde{q}_f,\tilde{q}_{p})$. The necessary condition for each of them to be optimum is to be in $F^I_1$. In addition, the necessary and sufficient condition for each of these candidate outcomes to be in $F^I_1$ is   $\frac{\Delta p-t_N}{\kappa_u}< \Delta q^*< \frac{\Delta p + t_{NoN}}{\kappa_u}$ (by Table~\ref{table:subsets}). Thus, for case (a), the necessary and sufficient condition is $\kappa_u \tilde{q}_{p}-t_{NoN}  < \Delta p <
	\kappa_u \tilde{q}_{p}+t_N $ (note that $\Delta q^*=\tilde{q}_{p}$), and for case (b) is  $\kappa_u ( \tilde{q}_{p}-\tilde{q}_f)-t_{NoN} <\Delta p < \kappa_u  ( \tilde{q}_{p}-\tilde{q}_f) +t_N$. These yields candidate strategies (a) and (b) and their conditions in the lemma.
	
	Consider $F^L_1$. By~\eqref{equ:summarize_CP_candidate_new},  if $(q^*_N,q^*_{NoN})\in F^L_1$, then $(q^*_N,q^*_{NoN})$ is $(c)$  $(0,\tilde{q}_{p})$. Note that the necessary and sufficient condition for  $(0,\tilde{q}_{p})\in F_1^L$ is $\Delta p \leq \kappa_u \tilde{q}_{p}-t_{NoN}$ (by the condition in Table~\ref{table:subsets} and $\Delta q=\tilde{q}_{p}$). The lemma follows.
\end{proof}

The payoff of the CP in each candidate solution of Theorem~\ref{lemma:summarycandidatesz=1_new} is as follows (using \eqref{equ:payoffCP_new}):

\be \label{equ:payoffCPinCandidatesz=1_new}
\small
\ba
\pi_{CP,(a)}&= n_{NoN}\kappa_{ad}\tilde{q}_{p}-\tilde{p}\tilde{q}_{p} \quad \text{  \& } 0< n_{NoN}< 1\\
\pi_{CP,(b)}&= (1-n'_{NoN})\kappa_{ad}\tilde{q_{f}} + n'_{NoN}\kappa_{ad}\tilde{q}_{p}-\tilde{p}\tilde{q}_{p}\ \text{ \& } 0< n'_{N}< 1\\
\pi_{CP,(c)}&=\kappa_{ad}\tilde{q}_{p}-\tilde{p}\tilde{q}_{p}
\ea
\ee
\normalsize

Thus, the payoffs are a function of $\tilde{p}$ and $\Delta p$. Now, to get the second main result of this section, we compare the payoff of the candidate answers with the payoff of the candidate strategies when $z=0$ considering different values of $\tilde{p}$ and $\Delta p$, and pick the maximum as the equilibrium strategy of the CP. Thus Theorem~\ref{theorem:p_tilde_new} is proved as follows:

\begin{figure*}[t]
	\centering
	\setlength\fboxsep{0pt}
	\setlength\fboxrule{0.25pt}
	\fbox{
		\begin{tikzpicture}{ht}
		\draw [ultra thick] [<->](0,0) -- (13,0);
		\draw [thick] (2,-.1) node[below]{\footnotesize{$-t_{NoN}+\kappa_u (\tilde{q}_p-\tilde{q}_f)$}} -- (2,.1) ;
		\draw [thick] (9.5,-.1) node[below]{\footnotesize{$\kappa_u \tilde{q}_{p}-t_{NoN}$}} -- (9.5,.1) ;
		\draw [thick] (6.5,-.1) node[below]{\footnotesize{$t_N+\kappa_u (\tilde{q}_{p}-\tilde{q}_f)$}} -- (6.5,.1) ;
		\draw [thick] (12.5,-.1) node[below]{\footnotesize{$\kappa_u \tilde{q}_{p}+t_N$}} -- (12.5,.1) ;
		\draw [thick,dotted] [<->] (2,.2) -- node[above]{\scriptsize{$(\tilde{q}_f,\tilde{q}_{p}) \in F^I_1$}}(6.5,0.2);
		\draw [thick,dotted] [<->] (9.5,0.2) -- node[above]{\scriptsize{$(0,\tilde{q}_{p}) \in F^I_1$}}(12.5,0.2);
		\draw [thick,dotted] [<->] (0,-0.8) -- node[below]{\scriptsize{$(0,\tilde{q}_{p}) \in F^L_1$}}(9.5,-0.8);
		\draw [thick] (13,-.1) node[above]{\footnotesize{$\Delta p$}}  (13,.2) ;
		\end{tikzpicture}}
	\caption{A schematic view of the ordering of different candidate equilibrium stratwgies characterized in Theorem~\ref{lemma:summarycandidatesz=1_new} with respect to  $\Delta p$ when $\tilde{q}_f>\frac{t_N+t_NoN}{\kappa_u}$ and $z=1$.}\label{figure:schematic_z=1_q<>}
\end{figure*}

\begin{proof}
	Now, for different regions of $\Delta p$, we compare the payoffs of the candidate equilibrium strategies characterized in Theorem~\ref{lemma:summarycandidatesz=1_new}  to each other and also to the equilibrium strategies in Theorem~\ref{lemma:CP_z=0_new},  and use tie-breaking assumptions (whenever needed) to characterize the equilibrium strategies of the CP.
	
	First consider  $\Delta p\leq \kappa_u \tilde{q}_{p}-t_{NoN}$. In this case, $\Delta p$ satisfies the necessary condition of candidate strategy (c) in Theorem~\ref{lemma:summarycandidatesz=1_new}. In addition, note that by \eqref{equ:payoffCPinCandidatesz=1_new}, $\pi_{CP,(c)}>\pi_{CP,(a)}$ and  $\pi_{CP,(c)}>\pi_{CP,(b)}$ (by $\tilde{q}_{p}>\tilde{q}_f$). Thus, for this region, (c) is chosen if and only if this strategy yields a higher  or equal (by tie-breaking assumption \ref{assumption:tie4}) payoff than the payoff when $z^{eq}=0$, that is $\kappa_u \tilde{q}_f$ (by Theorem~\ref{lemma:CP_z=0_new}). Thus, using Lemma~\ref{lemma:thresh_on_delta_p_z10_new_1}, $z^{eq}=1$, and
	$(q^{eq}_N,q^{eq}_{NoN})=(0,\tilde{q}_{p})\in F^L_1$ if  $\tilde{p}\leq \tilde{p}_{t,1}$. Otherwise $z^{eq}=0$, since the payoff of (c) and subsequently (a) and (b) are smaller than the payoff when $z^{eq}=0$. Thus, in this case, the equilibrium strategy can be found using Theorem~\ref{lemma:CP_z=0_new}. This is item 1 of the theorem.

For $\Delta p\geq t_N+\kappa_u \tilde{q}_{p}$,  the necessary condition of none of the candidate strategies in Theorem~\ref{lemma:summarycandidatesz=1_new} can be satisfied. Therefore, $z^{eq}=0$. This is item 4 of the theorem.
	
Now, for the rest of the proof, we consider $\kappa_u \tilde{q}_{p}-t_{NoN}<\Delta p<t_N+\kappa_u \tilde{q}_{p}$.  In this case, the necessary condition of candidate strategy (c) of Theorem~\ref{lemma:summarycandidatesz=1_new} cannot be satisfied.  Therefore, we can eliminate (c). On the other hand, the necessary and sufficient condition of (a) of Theorem~\ref{lemma:summarycandidatesz=1_new} can be met. Now, consider two different  cases, $\tilde{q}_f\leq \frac{t_N+t_{NoN}}{\kappa_u}$ and $\tilde{q}_f> \frac{t_N+t_{NoN}}{\kappa_u}$:

	\begin{itemize}
		\item $\tilde{q}_f\leq \frac{t_N+t_{NoN}}{\kappa_u}$.  This yields that $\kappa_u \tilde{q}_{p}-t_{NoN}\leq t_N+\kappa_u (\tilde{q}_{p}-\tilde{q}_f)$. For this case, consider two sub-cases:
		\begin{itemize}
			\item $\kappa_u \tilde{q}_{p}-t_{NoN}< \Delta p < t_N+\kappa_u (\tilde{q}_{p}-\tilde{q}_f)$. In this case, $\Delta p$ satisfies the necessary and sufficient condition of (b) in Theorem~\ref{lemma:summarycandidatesz=1_new}.
			Now, we should compare $\pi_{G,(a)}$ and $\pi_{G,(b)}$. In Lemma~\ref{lemma:thresh_on_delta_p_z1_new}, we compare the payoff of the two solutions. In addition, by tie breaking assumption~\ref{assumption:tie_diversify}, when the payoffs are equal the CP chooses (b) over (a). Thus, if $\Delta p\geq \Delta p _{t}$, (b), i.e. $(\tilde{q}_f,\tilde{q}_{p})$ would be chosen versus (a). Otherwise (a), i.e. $(0,\tilde{q}_{p})$ would be chosen. Now, we compare the payoff of the one chosen with the payoff of the case $z=0$, i.e. $\kappa_{ad}\tilde{q}_f$:
			\begin{itemize}
				\item If  $\Delta p\geq \Delta p _{t}$, then by Lemma~\ref{lemma:thresh_on_delta_p_z10_new_3} and tie-breaking assumption~\ref{assumption:tie4}, $z^{eq}=1$ and $(q^{eq}_N,q^{eq}_{NoN})=(\tilde{q}_f,\tilde{q}_{p})\in F^I_1$ if $\tilde{p}\leq \tilde{p}_{t,3}$ (by comparing the payoff of strategy (b) by the payoff when $z=0$, i.e. $\kappa_{ad}\tilde{q}_f$). Otherwise $z^{eq}=0$, and the equilibrium strategy can be found using Theorem~\ref{lemma:CP_z=0_new}. This is item 2-a-i of the theorem.
				\item If  $\Delta p< \Delta p_{t}$, then by Lemma~\ref{lemma:thresh_on_delta_p_z10_new_2} and tie-breaking assumption~\ref{assumption:tie4}, $z^{eq}=1$ and$(q^{eq}_N,q^{eq}_{NoN})=(0,\tilde{q}_{p})\in F^I_1$ if $\tilde{p}\leq \tilde{p}_{t,2}$ (by comparing the payoff of strategy (a) by the payoff when $z=0$, i.e. $\kappa_{ad}\tilde{q}_f$). Otherwise $z^{eq}=0$, and the equilibrium strategy can be found using Theorem~\ref{lemma:CP_z=0_new}. This is item 2-a-ii of the theorem.
			\end{itemize}

			\item $t_N+\kappa_u (\tilde{q}_{p}-\tilde{q}_f)\leq \Delta p< t_N+\kappa_u \tilde{q}_{p}$:   In this range, the necessary condition of (b)  of Theorem~\ref{lemma:summarycandidatesz=1_new} cannot be satisfied. Thus, the only candidate solution  by which $z=1$, whose necessary and sufficient conditions can be satisfied, is  $(a)$ (as stated before).  Therefore, we should compare the payoff of (a) with that of when $z^{eq}=0$, i.e. $\kappa_{ad}\tilde{q}_f$. Using Lemma~\ref{lemma:thresh_on_delta_p_z10_new_2} and Assumption~\ref{assumption:tie4}, if $\tilde{p}\leq\tilde{p}_{t,2}$ then $z^{eq}=1$ and $(q^{eq}_N,q^{eq}_{NoN})=(0,\tilde{q}_{p})\in F^I_1$. Otherwise $z^{eq}=0$, and the equilibrium strategy can be found using Theorem~\ref{lemma:CP_z=0_new}. This is item 2-b of the theorem.
		\end{itemize}
		\item  $\tilde{q}_f> \frac{t_N+t_{NoN}}{\kappa_u}$:  In this case, $\kappa_u \tilde{q}_{p}-t_{NoN}> t_N+\kappa_u (\tilde{q}_{p}-\tilde{q}_f)$. Thus, the necessary condition of (b) cannot be satisfied. Therefore, we can eliminate (c) (eliminated before) and (b). On the other hand, the necessary and sufficient condition of (a) of Theorem~\ref{lemma:summarycandidatesz=1_new} can be met.  Therefore, we should compare the payoff of (a) with that of when $z^{eq}=0$, i.e. $\kappa_{ad}\tilde{q}_f$. Using Lemma~\ref{lemma:thresh_on_delta_p_z10_new_2} and Assumption~\ref{assumption:tie4},  if $\tilde{p}\leq\tilde{p}_{t,2}$ then  $z^{eq}=1$ and $(q^{eq}_N,q^{eq}_{NoN})=(0,\tilde{q}_{p})\in F^I_1$. Otherwise $z^{eq}=0$, since the payoff of (a) is smaller than the payoff when $z^{eq}=0$. Thus, in this case, the equilibrium strategy can be found using Theorem~\ref{lemma:CP_z=0_new}. This is item 3 of the theorem. 		
	\end{itemize}
	The result follows.
\end{proof}

The following lemma simplify item 2-a of Theorem~\ref{theorem:p_tilde_new}, and is useful in the next stages:

\begin{lemma}\label{lemma:deltap_t}
	Consider  $\kappa_u \tilde{q}_{p}-t_{NoN}< \Delta p < t_N+\kappa_u(\tilde{q}_{p}-\tilde{q}_f)$. If $\tilde{q}_{p}\geq \frac{t_N+t_{NoN}}{\kappa_u}$, then $\Delta p<\Delta p_t$. If $\tilde{q}_{p}< \frac{t_N+t_{NoN}}{\kappa_u}$, then $\kappa_u \tilde{q}_{p}-t_{NoN} <  \Delta p_t < t_N+\kappa_u(\tilde{q}_{p}-\tilde{q}_f)$, where $\Delta p_{t}=\kappa_{u}(2\tilde{q}_{p}-\tilde{q}_f)-t_{NoN}$ characterized in Lemma~\ref{lemma:thresh_on_delta_p_z1_new}.
\end{lemma}

\begin{proof}
	First, consider $\tilde{q}_{p}\geq \frac{t_N+t_{NoN}}{\kappa_u}$. Note that:
	$$
	\ba
	\Delta p_t-(t_N+\kappa_u(\tilde{q}_{p}-\tilde{q}_f))&=\kappa_u \tilde{q}_{p}-t_N-t_{NoN}\geq 0
	\ea
	$$
	Thus for every $\Delta p$ such that $\Delta p < t_N+\kappa_u(\tilde{q}_{p}-\tilde{q}_f)$, $\Delta p_t>\Delta p$. This establish the first part of the lemma.
	
	Now, consider $\tilde{q}_{p}< \frac{t_N+t_{NoN}}{\kappa_u}$. In this case:
	$$
	\Delta p_t-(t_N+\kappa_u(\tilde{q}_{p}-\tilde{q}_f))=\kappa_u \tilde{q}_{p}-t_N-t_{NoN}< 0
	$$
	and
	$$
	\Delta p_t-(\kappa_u \tilde{q}_{p}-t_{NoN})=\kappa_u (\tilde{q}_{p}-\tilde{q}_f)>0 \qquad \text{ (since $\tilde{q}_p>\tilde{q}_f$)}
	$$
	Thus, $\kappa_u \tilde{q}_{p}-t_{NoN} <  \Delta p_t < t_N+\kappa_u(\tilde{q}_{p}-\tilde{q}_f)$. The result follows.
\end{proof}

Theorem~\ref{theorem:p_tilde_new} and Lemma~\ref{lemma:deltap_t} yields the following corollary:

\begin{corollary}\label{corollary:qNoN>}
	Let $\tilde{q}_p\geq \frac{t_N+t_{NoN}}{\kappa_u}$. Then the structure of the optimum answers of the CP (results in Theorem~\ref{theorem:p_tilde_new}) for the case that $\tilde{q}_f\leq \frac{t_N+t_{NoN}}{\kappa_u}$  is the same as the results when $\tilde{q}_f>\frac{t_N+t_{NoN}}{\kappa_u}$.
\end{corollary}
\begin{proof}
Note that items 1 and 4 of Theorem~\ref{theorem:NE_stage2_new_suff} are the same for both cases, regardless of $\tilde{q}_f$.  In addition by Lemma~\ref{lemma:deltap_t}, when $\tilde{q}_{p}\geq \frac{t_N+t_{NoN}}{\kappa_u}$, then $\Delta p<\Delta p_t$. Thus, 2-a-i in Theorem~\ref{theorem:p_tilde_new} would not happen. Note that 2-a-ii and 2-b yields is similar to 3. Thus, the two structures are similar, and the corollary follows.
\end{proof}

\section{Proof of  Theorem~\ref{theorem:NE_stage2_new_suff} - Stage 2}\label{appendix:side}

\begin{proof}
%
Consider a SPNE in which $z^{eq} = 1$. If $\tilde{p} >  \tilde{p}_t$, $z^{eq}=0$, by Theorem~\ref{theorem:p_tilde_new}. So $\tilde{p} \leq \tilde{p}_t$. The payoff of ISP NoN is  equal to $(p_{NoN}-c)n_{NoN}+\tilde{p}\tilde{q}_f$, by \eqref{equ:payoffISPsGeneral_new},  and is a strictly increasing function of $\tilde{p}$ (note that $p_{NoN}$ is fixed, and by \eqref{equ:EUs_linear}, $n_{NoN}$ is independent of $\tilde{p}$). Thus, for NoN to maximize her payoff, $\tilde{p}$ must assume the maximum value in the range that  $\tilde{p} \leq \tilde{p}_t.$ Thus, $\tilde{p} = \tilde{p}_t.$ Thus Theorem~\ref{theorem:NE_stage2_new_suff}-1 follows.


Now, we prove Theorem~\ref{theorem:NE_stage2_new_suff}-2.  Again, consider a SPNE in which $z^{eq} = 1$. By Theorem~\ref{theorem:p_tilde_new}, when $\Delta p\geq  t_N+\kappa_u \tilde{q}_{p}$, $z^{eq}=0$. Thus, 	$\Delta p<t_N+\kappa_u \tilde{q}_{p}$. By Theorem~\ref{theorem:NE_stage2_new_suff}-1, $\tilde{p} = \tilde{p}_t.$ Thus, $\pi_{NoN}(p_{NoN},\tilde{p}_{t}) = \pi_{NoN}(p_{NoN},\tilde{p})$. Thus, if $\pi_{NoN}(p_{NoN},\tilde{p}_{t}) \leq \pi_{NoN,z=0}(p_{NoN},\tilde{p})$, then 	$\pi_{NoN}(p_{NoN},\tilde{p}) \leq \pi_{NoN,z=0}(p_{NoN},\tilde{p}).$ Note that by Theorem~\ref{theorem:p_tilde_new}, NoN can ensure $z^{eq}=0$, by choosing $\tilde{p} > \tilde{p}_{t}.$ Thus, by tie-breaking assumption~\ref{assumption:ISP_p_tilde}, $z^{eq}=0$, which is a contradiction. Thus, $\pi_{NoN}(p_{NoN},\tilde{p}_{t}) > \pi_{NoN,z=0}(p_{NoN},\tilde{p}).$

Next, consider that $\pi_{NoN}(p_{NoN},\tilde{p}_{t}) > \pi_{NoN,z=0}(p_{NoN},\tilde{p})$ and $\Delta p<t_N+\kappa_u \tilde{q}_{p}.$ By Theorem~\ref{theorem:NE_stage2_new_suff}-1, if $z^{eq}=1$, $\pi_{NoN}(p_{NoN},\tilde{p}_{t}) = \pi_{NoN}(p_{NoN},\tilde{p})$. By Theorem~\ref{theorem:p_tilde_new}, NoN can ensure $z^{eq}=1$, by choosing $\tilde{p} = \tilde{p}_{t}.$ This fetches NoN a payoff, $\pi_{NoN}(p_{NoN},\tilde{p})$, which exceeds the payoff at $z=0$,  $\pi_{NoN,z=0}(p_{NoN},\tilde{p}).$
Thus, $z^{eq}=1.$ Theorem~\ref{theorem:NE_stage2_new_suff}-2 follows.

\end{proof}

\section{Proof of Theorem~\ref{theorem:neutralz=0_q<}}\label{appendix:theorem:neutralz=0_q<}

\begin{proof}
	We first prove the first part of the Theorem, considering  different regions of $\Delta p$ in Theorem~\ref{lemma:CP_z=0_new}. 
	
	
	
\textbf{Case 1:$-t_{NoN}<\Delta p<t_N$ } 
 By Theorem \ref{lemma:CP_z=0_new}-1, $(q^{eq}_N,q^{eq}_{NoN})=(\tilde{q}_f,\tilde{q}_f)\in F^I_0$. Note that in this region, $0<x_N<1$, and an SPNE strategy for ISPs  should satisfy the first order optimality conditions. Thus, using \eqref{equ:UN_new} and \eqref{equ:UNoN_new}:

\small	
	$$
	\small
	\pi_N(p_N)=(p_N-c)\frac{t_{NoN}+p_{NoN}-p_N}{t_N+t_{NoN}}
	$$
	
	$$
	\ba
	\pi_{NoN}(p_{NoN},\tilde{p})&=(p_{NoN}-c)\frac{t_N+p_N-p_{NoN}}{t_N+t_{NoN}}
	\ea
	$$
	\normalsize
	Solving the first order optimality condition yields:

\small
	\be\label{equ:ISPcandidateMultihome_z=0}
	\ba
	p^{eq}_N&= c+\frac{1}{3}(2t_{NoN}+t_N)\\
	p^{eq}_{NoN}&=c+\frac{1}{3}(2t_N+t_{NoN})
	\ea
	\ee
	\normalsize
	which is unique. Note that $p^{eq}_N\geq c$  and $p^{eq}_{NoN}\geq c$. First, note that $-t_{NoN}< \Delta p^{eq}=p^{eq}_{NoN}-p^{eq}_N=\frac{t_N-t_{NoN}}{3}< t_N$.
	
	The necessary condition for this strategy to be an SPNE is $\pi_{NoN,z=0}(p^{eq}_{NoN}, \tilde{p}^{eq})\geq \pi_{NoN}(p^{eq}_{NoN},\tilde{p}_t)$ (by Theorem~\ref{theorem:NE_stage2_new_suff}). The candidate strategies and this necessary  condition is listed in the statement of the theorem.
	
	\textbf{Case 2: $\Delta p\geq t_{N}$ } 
 We consider two cases $\Delta p=t_N$ and $\Delta p>t_N$ in Cases 2-i and 2-ii, respectively. \\
	\textbf{Case 2-i:$\Delta p = t_{N}$ } 
Using Theorem~\ref{lemma:CP_z=0_new}-2, CP chooses $(q_N^{eq},q_{NoN}^{eq})=(\tilde{q}_f,0)\in F_0^U$ to maximize her payoff. Thus, $n_{NoN}^{eq}=0$ and $\pi_{NoN,z=0}(p_{NoN}^{eq}, \tilde{p})=0$, i.e. the payoff of  ISP NoN is zero.  Consider $\epsilon>0$ such that $p'_{NoN}=p_{NoN}^{eq}-\epsilon> c$. In this case, $p'_{NoN}-p_N^{eq}<t_{N}$. Thus, by  Theorem~\ref{lemma:CP_z=0_new},  CP chooses her qualities in $ F^I_0 \cup F^L_0$. Thus, $n_{NoN}'>0$, and $\pi_{NoN,z=0}(p'_{NoN}, \tilde{p})>0$. Thus, $p'_{NoN}$ is a profitable deviation from $p_{NoN}^{eq}$ for  ISP NoN. Therefore,  $p_{NoN}^{eq}$ and $p_N^{eq}$ such that $\Delta p=t_{N}$ cannot be SPNE.\\
\textbf{Case 2-ii: $\Delta p> t_{N}$ } 
 By  Theorem~\ref{lemma:CP_z=0_new}-2,  $n^{eq}_N=1$.   Consider a unilateral deviation by neutral ISP such that $p'_N=p^{eq}_N+\epsilon$ in which $\epsilon>0$ such that $p^{eq}_{NoN}-p'_N>t_N$. Note that the values of $q^{eq}_N$ and $q^{eq}_{NoN}$ is the same as before, since still $\Delta p'=p^{eq}_{NoN}-p'_N>t_N$. Thus, again $n^{eq}_N=1$, and by \eqref{equ:payoffISPsGeneral_new}, the payoff of neutral ISP is an increasing function of $p_N$. Thus, $p'_N$ is a profitable unilateral deviation. This contradicts the assumption that $p^{eq}_N$ and $p^{eq}_{NoN}$ is SPNE. 

\textbf{Case 3:$\Delta p \leq -t_{NoN}$ }  By  Theorem~\ref{lemma:CP_z=0_new}-3, $n_{NoN}^{eq} = 1, n_{N}^{eq} = 0.$ 
We consider two cases $\Delta p < -t_{NoN}$ and $\Delta p=-t_{NoN}$ in Cases 3-i and 3-ii, respectively. \\
\textbf{Case 3-i: $\Delta p < -t_{NoN}$ }If $z^{eq}=0$, NoN's payoff is $p^{eq}_{NoN}-c$,  by~\eqref{equ:payoffISPsGeneral_new}. This is strictly increasing in $p^{eq}_{NoN}.$ Thus, there is no SPNE in this region. \\
\textbf{Case 3-ii: $\Delta p = -t_{NoN}$ } Since $n_{N}^{eq} = 0,$ N gets a payoff of $0.$ Recall that $p^{eq}_{NoN}\geq c$ since $z^{eq}=0$ (Section~\ref{section:theory}, second paragraph). Thus, $p^{eq}_{N} = p^{eq}_{NoN} + t_{NoN} > c.$
Now, if we decrease $p_{N}$ by a small $\epsilon > 0$ from  $p^{eq}_{N}$, then $-t_{NoN} < \Delta p < t_N$. Then, by Theorem \ref{lemma:CP_z=0_new}-1, CP chooses $(q_N,q_{NoN})=(\tilde{q}_f,\tilde{q}_f)\in F^I_0$ to optimize her payoff,   and $n_{N}$ becomes positive. Also, $p_N > c.$ Thus, N gets a positive payoff, by~\eqref{equ:payoffISPsGeneral_new}. This contradicts the assumption that $p^{eq}_N$ and $p^{eq}_{NoN}$ is SPNE. The first part of the Theorem follows.

We now consider the special case that $t_N+t_{NoN}\leq \kappa_u \tilde{q}_p. $ First, let $p^{eq}_N, p^{eq}_{NoN}$ be such that $\Delta p^{eq}>\kappa_u\tilde{q}_p-t_{NoN}$. It follows from both these that    
$\Delta p >t_N$. We know from Case 2-ii of the general case that there is no SPNE such that $z^{eq}=1$ in this case.
 Now, consider $\Delta p\leq \kappa_u \tilde{q}_p-t_{NoN}$. In this region, if $z^{eq}=0$ NoN's payoff is at most $p^{eq}_{NoN}-c$ (by~\eqref{equ:payoffISPsGeneral_new}). On the other hand, by Theorem \ref{theorem:p_tilde_new}, by choosing $\tilde{p}'=\tilde{p}_{t,1}$,  NoN can ensure that the CP chooses $z^{eq}=1$. In this case, by~\eqref{equ:payoffISPsGeneral_new}, NoN's payoff is $p'_{NoN}-c+\tilde{p}_{t,1}\tilde{q}_{NoN}=p^{eq}_{NoN}-c+\kappa_{ad}(\tilde{q}_{p}-\tilde{q}_f)>p^{eq}_{NoN}-c$. Thus, $\pi_{NoN}(p^{eq}_{NoN},\tilde{p}_{t,1})>\pi_{NoN,z=0}(p^{eq}_{NoN},\tilde{p})$, and by Theorem~\ref{theorem:NE_stage2_new_suff}, $z^{eq}=1$. Thus, in this case, there is no SPNE by which $z^{eq}=0$. The second part follows.


\end{proof}
\section{Proof of Theorem~\ref{theorem:NE_stage1_new_q>}}
\label{appendix:firstproof}
 Before proving the theorem, we state and prove two lemmas,  which are used in the proof of the theorem:

\begin{lemma}\label{lemma:appendix_caseA}
	If $p_{NoN}=c+\kappa_u \tilde{q}_{p}-t_{NoN}$ and $p_N=c$, then $z^{eq}=1$.
\end{lemma}
\begin{proof}
	Note that in this case, $\Delta p=\kappa_u \tilde{q}_{p}-t_{NoN}$. Thus, $\tilde{p}_t=\tilde{p}_{t,1}$. Therefore, using Theorem~\ref{theorem:NE_stage2_new_suff}, it is sufficient to prove that $\pi_{NoN}(p_{NoN},\tilde{p}_{t,1})>\pi_{NoN,z=0}(p_{NoN},\tilde{p})$, where $\pi_{NoN,z=0}(\tilde{p}_{NoN},\tilde{p})$ is the payoff of ISP NoN when $z^{eq}=0$. Note that $\pi_{NoN,z=0}(p_{NoN},\tilde{p})\leq p_{NoN}-c=\kappa_u \tilde{q}_p-t_{NoN}$ and $\pi_{NoN}(p_{NoN},\tilde{p}_{t,1})=\kappa_u \tilde{q}_p-t_{NoN}+\kappa_{ad}(\tilde{q}_p-\tilde{q}_f)$  (since by Theorem~\ref{theorem:p_tilde_new}, $n_{NoN}=1$, and by \eqref{equ:payoffISPsGeneral_new}). In addition, note that,  $\tilde{q}_p>\tilde{q}_f$. Thus, this condition holds, and the result follows.
\end{proof}

\begin{lemma}\label{lemma:appendix_caseB}
	If 	$p_{NoN}=c+\frac{t_{NoN}+2t_N+  \tilde{q}_{p}(\kappa_u -2\kappa_{ad})}{3}$, $p_{N}=c+\frac{2t_{NoN}+t_N-\tilde{q}_{p}(\kappa_u +\kappa_{ad})}{3}$,  $\tilde{q}_{p}<\frac{t_N+2t_{NoN}}{\kappa_u +\kappa_{ad}}$, and $\kappa_u \tilde{q}_p\geq t_N+t_{NoN}$, then $z^{eq}=1$.
\end{lemma}
\begin{proof}
	Note that if  $\kappa_u \tilde{q}_{p}-t_{NoN}<\Delta p<t_N+\kappa_u \tilde{q}_{p}$, by definition of $\tilde{p}_t$ (Definition~\ref{def:pt}), $\tilde{p}_t=\tilde{p}_{t,2}$. Thus, by Theorem \ref{theorem:NE_stage2_new_suff}, it is enough to prove that $\pi_{NoN}(p_{NoN},\tilde{p}_{t,2})>\pi_{NoN,z=0}(\tilde{p}_{NoN},\tilde{p})$, where $\pi_{NoN,z=0}(\tilde{p}_{NoN},\tilde{p})$ is the payoff of ISP NoN when $z^{eq}=0$.

	First, we prove that $\pi_{NoN}(p_{NoN},\tilde{p}_{t,2})>{p}_N-c+ \kappa_u\tilde{q}_{p}-t_{NoN}+\kappa_{ad}(\tilde{q}_{p}-\tilde{q}_f)$:
	$$
	\footnotesize
	\ba
	&	\pi_{NoN}(p_{NoN},\tilde{p}_{t,2})\geq  {p}_N-c+ \kappa_u\tilde{q}_{p}-t_{NoN}+\kappa_{ad}(\tilde{q}_{p}-\tilde{q}_f)\\
	&\iff \frac{\big{(}t_{NoN}+2t_N+  \tilde{q}_{p}(\kappa_u +\kappa_{ad})\big{)}^2}{9(t_N+t_{NoN})} \geq  \frac{t_N-t_{NoN}+2\tilde{q}_{p}(\kappa_u +\kappa_{ad})}{3}\\
	& \iff (\tilde{q}_{p}(\kappa_u+\kappa_{ad})-t_N-2t_{NoN})^2\geq 0
	\ea
	$$
	\normalsize
	In addition, note that $ {p}_N-c+ \kappa_u\tilde{q}_{p}-t_{NoN}+\kappa_{ad}(\tilde{q}_{p}-\tilde{q}_f)>0$, since $p_N\geq c$ (under the condition $\tilde{q}_{p}<\frac{t_N+2t_{NoN}}{\kappa_u +\kappa_{ad}}$), $\kappa_u\tilde{q}_{p}-t_{NoN}\geq t_N>0$ (by the assumption of the lemma), and $\tilde{q}_p>\tilde{q}_f$. Thus, $\pi_{NoN}(p_{NoN},\tilde{p}_{t,2})>0$.
	
	Now, consider $\pi_{NoN,z=0}(\tilde{p}_{NoN},\tilde{p})$. Note that by the assumption of the lemma $\kappa_u\tilde{q}_p\geq t_N+t_{NoN}$. Thus, $\Delta p>t_N$, and by item 2 of Theorem~\ref{lemma:CP_z=0_new}, if $z^{eq}=0$, $n_{NoN}=0$. Thus, by \eqref{equ:payoffISPsGeneral_new}, $\pi_{NoN,z=0}(\tilde{p}_{NoN},\tilde{p})=0$. Therefore, $\pi_{NoN}(p_{NoN},\tilde{p}_{t,2})>\pi_{NoN,z=0}(\tilde{p}_{NoN},\tilde{p})$, and the result follows.
\end{proof}

Now, we prove Theorem~\ref{theorem:NE_stage1_new_q>}:

\begin{proof}
	We use the optimum strategies of the CP characterized in Theorem~\ref{theorem:p_tilde_new} to characterize Nash equilibria. Note that for the case that $\kappa_u\tilde{q}_p\geq t_N+t_{NoN}$, by Corollary~\ref{corollary:qNoN>}, the structure of the equilibrium strategies chosen by the CP is similar to the case that $\kappa_u \tilde{q}_f>t_N+t_{NoN}$. Thus, in this case, items 1, 3, and 4 of Theorem~\ref{theorem:p_tilde_new} characterizes the NE strategies chosen by the CP. Thus, henceforth we assume $\kappa_u \tilde{q}_p\geq t_N+t_{NoN}$, and  use these items to prove the theorem.
	
	We denote $\Delta p\leq \kappa_u \tilde{q}_{p}-t_{NoN}$ by region A, $\kappa_u \tilde{q}_{p}-t_{NoN}<\Delta p<t_N+\kappa_u \tilde{q}_{p}$ by region B, and $\Delta p\geq t_N+\kappa_u \tilde{q}_{p}$ by region C. Using Theorem~\ref{theorem:p_tilde_new}, if $z^{eq}=1$, then $\Delta p<t_N+\kappa_u \tilde{q}_{p}$. Thus, to characterize NE strategies by which $z^{eq}=1$, we should characterize any possible NE strategies in regions A and B. In Case A, we prove that the only possible NE in region A is $p^{eq}_{NoN}=c+\kappa_u \tilde{q}_{p}-t_{NoN}$ and $p^{eq}_N=c$. In addition, we prove that these strategies are NE if $\tilde{q}_{p}\geq \frac{t_N+2t_{NoN}}{\kappa_u+\kappa_{ad}}$. If not, then there is no NE in region A. In Case B, we prove that the only possible NE in region B is $p^{eq}_{NoN}=c+\frac{t_{NoN}+2t_N+\tilde{q}_{p}(\kappa_u -2\kappa_{ad})}{3}$ and $p^{eq}_{N}=c+\frac{2t_{NoN}+t_N-\tilde{q}_{p}(\kappa_u +\kappa_{ad})}{3}$. In addition, we prove that these strategies can be NE strategies if $\tilde{q}_{p}\geq \frac{t_N+2t_{NoN}}{\kappa_u+\kappa_{ad}}$. If not, then there is no NE in region B.
	
	\textbf{Case A:} We characterize the NE strategies $p^{eq}_N$ and $p^{eq}_{NoN}$ such that $\Delta p^{eq}=p^{eq}_{NoN}-p^{eq}_N\leq \kappa_u \tilde{q}_{p}-t_{NoN}$. First, in Case A-1, we prove that if $z^{eq}=1$ the only possible NE in this region is $p^{eq}_{NoN}=c+\kappa_u \tilde{q}_{p}-t_{NoN}$ and $p^{eq}_N=c$, and with these strategies, $z^{eq}$ is indeed equal to 1. In Case A-2,  we characterize the necessary and sufficient conditions by which there is no unilateral profitable deviation for ISPs. This provides the necessary and sufficient condition for these strategies to be NE.\\
\textbf{Case A-1:} First, we consider $\Delta p\leq \kappa_u \tilde{q}_{p}-t_{NoN}$. In this case, we show that the payoff of ISP NoN is an increasing function of $\Delta p$. Then, we characterize the NE as  $p^{eq}_{NoN}=c+\kappa_u \tilde{q}_{p}-t_{NoN}$ and $p^{eq}_N=c$, using the fact that when choosing an NE, no player can increase her payoff by unilaterally changing her strategy.

Note that by Theorem~\ref{theorem:p_tilde_new}, for region A, $(q^{eq}_N,q^{eq}_{NoN})=(0,\tilde{q}_{p})\in F^L_1$ if and only if $\tilde{p}\leq \tilde{p}_{t,1}=\kappa_{ad}(1-\frac{\tilde{q}_f}{\tilde{q}_{p}})$. In addition, by Theorem~\ref{theorem:NE_stage2_new_suff}, if $z^{eq}=1$ then $\tilde{p}^{eq}=\tilde{p}_{t,1}=\kappa_{ad}(1-\frac{\tilde{q}_f}{\tilde{q}_{p}})$ (Definition~\ref{def:pt1,pt2}).  Thus, in  this region, if $z^{eq}=1$, the payoff of ISP NoN is equal to $p_{NoN}-c+\tilde{q}_{p}\tilde{p}_{t,1}$ (by \eqref{equ:payoffISPsGeneral_new}) since $n_{NoN}=1$. Therefore, the payoff is an increasing function of $p_{NoN}$. In addition, note that in region A, $n_N=0$ and regardless of $p_N$, the neutral ISP receives a payoff of zero (by \eqref{equ:payoffISPsGeneral_new}). Thus, $p^{eq}_{NoN}$, i.e. the equilibrium Internet access fee, should be such that the neutral ISP cannot get a positive payoff  by increasing or decreasing $p_N$, and changing the region of $\Delta p$ to $B_1$, $B_2$, or $C$. Using this condition, we find the equilibrium strategy.

First consider a unilateral deviation by ISP N. Note that increasing $p_N$ decreases $\Delta p$, and cannot change the region of $\Delta p$.\footnote{Recall that in this region, $n_N=0$, and ISP N fetch a payoff of zero.} Thus, a deviation of this kind would not be profitable. We claim that by decreasing $p_N$ to $p'_N$ such that $p_{NoN}-p'_N>\kappa_u \tilde{q}_{p}-t_{NoN}$, the ISP N can fetch a positive payoff as long as $p'_N>c$ (the claim is proved in the next paragraph). Therefore, in the equilibrium, $p^{eq}_{NoN}$ is such that even with $p'_N=c$ (the minimum plausible price), $\Delta p \leq \kappa_u \tilde{q}_{p}-t_{NoN}$. Thus, $p^{eq}_{NoN}\leq c+ \kappa_u \tilde{q}_{p}-t_{NoN}$.\footnote{Otherwise, there exists a $p'_N>c$ by which  $\Delta p > \kappa_u \tilde{q}_{p}-t_{NoN}$.} Given that the payoff of ISP NoN is an increasing function of $p_{NoN}$, we get  $p^{eq}_{NoN}=c+\kappa_u \tilde{q}_{p}-t_{NoN}$. In addition, we claim that $p^{eq}_N=c$. If not, then $p^{eq}_N>c$. In this case, $\Delta p=p^{eq}_N-p^{eq}_{NoN}<\kappa_u \tilde{q}_{p}-t_{NoN}$.  We argued that the payoff of ISP NoN is an increasing function of $p_{NoN}$. Thus, by increasing $p_{NoN}$ such that $\Delta p=\kappa_u \tilde{q}_{p}-t_{NoN}$, ISP NoN can increase her payoff, which is a contradiction with $p^{eq}_N$ and $p^{eq}_{NoN}$ being NE strategies.


To prove the claim, note that if  $p_{NoN}-p'_N>\kappa_u \tilde{q}_{p}-t_{NoN}$, then either (i) $z^{eq}=0$ or (ii) $z^{eq}=1$. For case (i), since  $\kappa_u \tilde{q}_{p}-t_{NoN}>-t_{NoN}$, when $\Delta p>\kappa_u \tilde{q}_{p}-t_{NoN}$, then $(q^{eq}_N,q^{eq}_{NoN})$ is of the form of items 1 or  2 of Theorem~\ref{lemma:CP_z=0_new}.  Thus, $n_N>0$. Therefore ISP N can fetch a positive payoff as long as $p_N>c$ (by \eqref{equ:payoffISPsGeneral_new}). Now consider case (ii), i.e. $z^{eq}=1$. In this case, if $z^{eq}=1$, then by using item 2 of Thoerem~\ref{theorem:p_tilde_new}, $n_N>0$ (since solutions that yield $z^{eq}=1$ are in $F^I$.). Thus, ISP N can fetch a positive payoff as long as $p_N>c$ (by \eqref{equ:payoffISPsGeneral_new}). This completes the proof of the claim that by decreasing $p_N$ to $p'_N$ such that $p_{NoN}-p'_N>\kappa_u \tilde{q}_{p}-t_{NoN}$, the ISP N can fetch a positive payoff as long as $p'_N>c$.

Therefore, the NE strategies are $p^{eq}_{NoN}=c+\kappa_u \tilde{q}_{p}-t_{NoN}$ and $p^{eq}_N=c$, and the payoff of the ISP NoN at this price by \eqref{equ:payoffISPsGeneral_new} and $\tilde{p}_{t,1}=\kappa_{ad}(1-\frac{\tilde{q}_f}{\tilde{q}_{p}})$ is equal to  (note that $n_{NoN}=1$), and
\be \label{equ:payoff_NoN_eq_1_new}
\pi^{eq}_{NoN}=\kappa_u \tilde{q}_{p}-t_{NoN}+\tilde{q}_{p} \tilde{p}_{t,1}=\kappa_u \tilde{q}_p-t_{NoN}+\kappa_{ad}(\tilde{q}_p-\tilde{q}_f)
\ee
which is strictly positive since $t_{NoN} < \kappa_u \tilde{q}_p-t_{NoN}+\kappa_{ad}(\tilde{q}_p-\tilde{q}_f)$. 
The first item of the theorem follows.

	\textbf{Case A-2:} Now, in order to prove that $p^{eq}_N$ and $p^{eq}_{NoN}$ are indeed NE strategies, we show that there is no unilateral profitable deviation for ISPs. First, in Case (A-2-i) we rule out the possibility of a unilateral profitable deviation for ISP N. Then, in Case (A-2-ii) we rule out a possibility of a downward unilateral  profitable deviation, i.e.  $p_{NoN}<p^{eq}_{NoN}$, for ISP NoN. Finally, in Case (A-3-iii), we consider a deviation of the form  $p_{NoN}>p^{eq}_{NoN}$ for ISP NoN, and prove that the necessary and sufficient condition for this deviation to be not profitable is  $\tilde{q}_{p}\geq \frac{t_N+2t_{NoN}}{\kappa_u+\kappa_{ad}}$.
	
	\textbf{Case A-2-i: } The construction of strategies $p^{eq}_N$ and $p^{eq}_{NoN}$ yields that there is no profitable deviation for ISP N. To prove this formally, note that the only deviation for ISP N that might be profitable is $p_N>c$. With this deviation, $\Delta p$ would be still in region A, in which $n_N=0$, and the payoff of ISP N is zero. Thus, such a deviation is not profitable.
	
	\textbf{Case A-2-ii: } Now, consider a deviation by ISP NoN such that  $p_{NoN}<p^{eq}_{NoN}$. In this case, $\Delta p$ is in region A, and the payoff of ISP NoN is equal to $p_{NoN}-c+\tilde{q}_{p}\tilde{p}_{t,1}$ (by \eqref{equ:payoffISPsGeneral_new} and $n_{NoN}=1$). Thus, the payoff of ISP NoN is strictly increasing in region $A$. Therefore, $p^{eq}_{NoN}$ dominates all prices $p_{NoN}<p^{eq}_{NoN}$. Thus, this kind of deviation is not profitable for ISP NoN.
	
	\textbf{Case A-2-iii: } In this case, we consider a deviation such that  $p_{NoN}>p^{eq}_{NoN}$. Thus, $\Delta p>\kappa_u \tilde{q}_{p}-t_{NoN}$. Therefore, $\Delta p$ is either in Region B or C. First, in Case A-2-iii-a we rule out the possibility of a profitable unilateral deviation in region C. Then, in Case A-2-iii-b, we rule out the possibility of a profitable unilateral deviation in region B if  $z^{eq}=0$. Finally, in Case A-2-iii-c, we prove that a deviation to region B if $z^{eq}=1$ is not profitable if and only if  $\tilde{q}_{p}\geq \frac{t_N+2t_{NoN}}{\kappa_u+\kappa_{ad}}$.
	
	\textbf{Case A-2-iii-a: } Using item 4 of  Theorem~\ref{theorem:p_tilde_new}, if $\Delta p$ in region C, i.e. $\Delta p\geq t_N+\kappa_u \tilde{q}_{p}$, then $z^{eq}=0$.   In this case,  $(q^{eq}_N,q^{eq}_{NoN})$ is of the form of part 2 of Theorem~\ref{lemma:CP_z=0_new} (note that $\kappa_u \tilde{q}_p\geq t_N+t_{NoN}$). Thus, $n_{NoN}=0$. Therefore, the ISP NoN receives a payoff of zero, and a deviation of this kind in not profitable for this ISP (since the equilibrium payoff is positive.).
	
	\textbf{Case A-2-iii-b: } Consider a deviation to Region B by ISP NoN by which $z^{eq}=0$.  then by item 2 of Theorem~\ref{lemma:CP_z=0_new}, $n_{NoN}=0$. Therefore, the ISP NoN receives a payoff of zero, and a deviation of this kind in not profitable for this ISP.
	
	\textbf{Case A-2-iii-c: } Now, consider  Consider a deviation to Region B by ISP NoN by which   $z^{eq}=1$. In this case, by item 3 of Theorem~\ref{theorem:p_tilde_new}, $(0,\tilde{q}_{p})\in F^I_1$, and by Theorem~\ref{theorem:NE_stage2_new_suff} and Lemma~\ref{lemma:thresh_on_delta_p_z10_new_2}, $\tilde{p}^{eq}=\tilde{p}_{t,2}=\kappa_{ad} (n_{NoN}-\frac{\tilde{q}_f}{\tilde{q}_{p}})$ and $n_{NoN}=\frac{t_N+\kappa_u \tilde{q}_{p}-\Delta p}{t_N+t_{NoN}}$. Therefore, using \eqref{equ:payoffISPsGeneral_new}:
	\be\label{equ:payoff_deviation_A}
	\ba
	\pi_{NoN}(\tilde{p}'_{NoN},\tilde{p}_{t,2})&=(p'_{NoN}-c)n_{NoN}+ \kappa_{ad} (n_{NoN}\tilde{q}_{p}-\tilde{q}_f)\\
	& =(p'_{NoN}-c+\kappa_{ad}\tilde{q}_{p})n_{NoN}-\kappa_{ad}\tilde{q}_f
	\ea
	\ee
	in which $n_{NoN}=\frac{t_N+\kappa_u \tilde{q}_{p}-p'_{NoN}+c}{t_N+t_{NoN}}$. The maximum $\pi_{NoN}(\tilde{p}'_{NoN},\tilde{p}_{t,2})$ can be found by applying the first order condition on the payoff, which gives us:
	\be \label{equ:equ:p'}
	p^*_{NoN}=c+\frac{1}{2}(t_N+\tilde{q}_{p}(\kappa_u -\kappa_{ad}))
	\ee
	
	This deviation is a profitable deviation in region B if (i) $\pi_{NoN}(\tilde{p}^*_{NoN},\tilde{p}_{t,2})>\pi^{eq}_{NoN}$ and (ii)  $\kappa_u \tilde{q}_{p}-t_{NoN}<p^*_{NoN}-c<t_N+\kappa_u \tilde{q}_{p}$. We also claim (claim is proved in the next two paragraphs) that if any deviation to region B is profitable, then (i) $\pi_{NoN}(\tilde{p}^*_{NoN},\tilde{p}_{t,2})>\pi^{eq}_{NoN}$ and (ii)  $\kappa_u \tilde{q}_{p}-t_{NoN}<p^*_{NoN}-c<t_N+\kappa_u \tilde{q}_{p}$. Thus, a deviation to this region is profitable if and only if (i) $\pi_{NoN}(\tilde{p}^*_{NoN},\tilde{p}_{t,2})>\pi^{eq}_{NoN}$ and (ii)  $\kappa_u \tilde{q}_{p}-t_{NoN}<p^*_{NoN}-c<t_N+\kappa_u \tilde{q}_{p}$.

	Now, we prove the claim that  (i) $\pi_{NoN}(\tilde{p}^*_{NoN},\tilde{p}_{t,2})>\pi^{eq}_{NoN}$ and (ii)  $\kappa_u \tilde{q}_{p}-t_{NoN}<p^*_{NoN}-c<t_N+\kappa_u \tilde{q}_{p}$.
	are necessary condition for a profitable deviation. First, we prove that (ii) is a necessary condition. Suppose (ii) is not true. We claim that no  $p'_{NoN}$ such that $\kappa_u \tilde{q}_{p}-t_{NoN}<p'_{NoN}-c<t_N+\kappa_u \tilde{q}_{p}$ can be a profitable deviation. To prove this, note that by concavity of \eqref{equ:payoff_deviation_A}, if $p^*_{NoN}$ is not such that   $\kappa_u \tilde{q}_{p}-t_{NoN}<p^*_{NoN}-c<t_N+\kappa_u \tilde{q}_{p}$, then all $p'_{NoN}$ such that $\kappa_u \tilde{q}_{p}-t_{NoN}<p'_{NoN}-c<t_N+\kappa_u \tilde{q}_{p}$ yields a strictly lower payoff than the maximum of payoffs at the  boundary points. Note that  with the upper boundary point, $\Delta p=p'_{NoN}-c=t_N+\kappa_u \tilde{q}_p$. In this case, by item 4 of Theorem~\ref{theorem:p_tilde_new}, $z^{eq}=0$, and by item 2 of Theorem~\ref{lemma:CP_z=0_new}, $n_{{NoN}}=0$. Thus, the payoff of ISP NoN is zero (by \eqref{equ:payoffISPsGeneral_new}). On the other hand, in the lower boundary point, i.e. $p'_{NoN}=\kappa_u \tilde{q}_{p}-t_{NoN}+c$ is equal to $p^{eq}_{NoN}$. Thus, the  maximum payoff at the boundary points is equal to the equilibrium payoff. Therefore,  if $p^*_{NoN}$ is not such that
	$\kappa_u \tilde{q}_{p}-t_{NoN}<p^*_{NoN}-c<t_N+\kappa_u \tilde{q}_{p}$, then all $p'_{NoN}$ such that $\kappa_u \tilde{q}_{p}-t_{NoN}<p'_{NoN}-c<t_N+\kappa_u \tilde{q}_{p}$, yields a payoff which is strictly less than the equilibrium payoff. The proof of (ii) being a necessary condition is complete.
	
	Now, we prove that (i) is a necessary condition. Suppose (i) is not true and  $\pi_{NoN}(\tilde{p}^*_{NoN},\tilde{p}_{t,2})\leq \pi^{eq}_{NoN}$. Then, either (ii) is true or not. If (ii) is not true, in the previous paragraph, we prove that  no  $p'_{NoN}$ if Region B can be a profitable deviation, which yields the result. Now, consider the case that (ii) holds. In this case, by concavity of the payoff, $p^*_{NoN}$ yields the highest payoff among $p_{NoN}$'s in Region B. Thus, $\pi_{NoN}(\tilde{p}^*_{NoN},\tilde{p}_{t,2})\leq \pi^{eq}_{NoN}$ yields that a deviation to Region B cannot be profitable. This completes the proof of the claim.
	
	Thus, a deviation to region B is profitable if and only if (i) $\pi_{NoN}(\tilde{p}^*_{NoN},\tilde{p}_{t,2})>\pi^{eq}_{NoN}$ and (ii)  $\kappa_u \tilde{q}_{p}-t_{NoN}<p^*_{NoN}-c<t_N+\kappa_u \tilde{q}_{p}$. First we check (i) and then (ii). Using \eqref{equ:payoff_deviation_A}, \eqref{equ:equ:p'}, and the expressions of $n_{NoN}$, we find the payoff of ISP NoN after deviation and compare it to the value of \eqref{equ:payoff_NoN_eq_1_new}. We claim that (i) is always true unless $\tilde{q}_{p}=\frac{t_N+2t_{NoN}}{\kappa_u+\kappa_{ad}}$. Note that:
	$$
	\ba
	&\pi_{NoN}(\tilde{p}^*_{NoN},\tilde{p}_{t,2})\geq \pi^{eq}_{NoN}\\
	&\iff \frac{(t_N+\tilde{q}_{p}(\kappa_{ad}+\kappa_u))^2}{4(t_N+t_{NoN})} \geq  \tilde{q}_{p}(\kappa_u+\kappa_{ad})-t_{NoN} \\
	&\iff \big{(}(\kappa_u+\kappa_{ad})\tilde{q}_{p}-t_N-2t_{NoN}\big{)}^2\geq 0
	\ea
	$$
	Thus, (i) is true if and only if $\tilde{q}_{p}\neq \frac{t_N+2t_{NoN}}{\kappa_u+\kappa_{ad}}$.

	Now, we check (ii). Note that $p^*_{NoN}-c<t_N+\kappa_u \tilde{q}_{p}$ since:
	$$
	\ba
	p^*_{NoN}-c<t_N+\kappa_u \tilde{q}_{p} \iff \tilde{q}_{p}(\kappa_u+\kappa_{ad})>-t_N
	\ea
	$$
	is always true. Now, we should check the lowerbound, i.e. $\kappa_u \tilde{q}_{p}-t_{NoN}<p^*_{NoN}-c$:
	$$
	\kappa_u \tilde{q}_{p}-t_{NoN}<p^*_{NoN}-c\iff \tilde{q}_{p}(\kappa_u+\kappa_{ad})<t_N+2t_{NoN}
	$$
	which is true if and only if $\tilde{q}_{p}<\frac{t_N+2t_{NoN}}{\kappa_u+\kappa_{ad}}$.
	
	Now, using the conditions for (i) and (ii) to be true, we can say that (i) and (ii) are true if and only if  $\kappa_u \tilde{q}_{p}-t_{NoN}<p^*_{NoN}-c$. Thus, there is no profitable deviation to region B if and only if $\tilde{q}_{p}\geq \frac{t_N+2t_{NoN}}{\kappa_u+\kappa_{ad}}$.
	
	
	This completes the proof of item 1 of theorem that $p^{eq}_{NoN}=c+\kappa_u \tilde{q}_{p}-t_{NoN}$ and $p^{eq}_N=c$ are NE strategies  if and only if $\tilde{q}_{p}\geq \frac{t_N+2t_{NoN}}{\kappa_u+\kappa_{ad}}$.
	
	\textbf{Case B:} Now, we characterize any possible NE strategies in region B, i.e. $\kappa_u \tilde{q}_{p}-t_{NoN}<\Delta p<t_N+\kappa_u \tilde{q}_{p}$, by which $z^{eq}=1$. First, in case B-1 we prove that if $z^{eq}=1$, the only possible NE in this region is  $p^{eq}_{NoN}=c+\frac{t_{NoN}+2t_N+  \tilde{q}_{p}(\kappa_u -2\kappa_{ad})}{3}$ and $p^{eq}_{N}=c+\frac{2t_{NoN}+t_N-\tilde{q}_{p}(\kappa_u +\kappa_{ad})}{3}$. We also prove that the necessary condition for these strategies to be a NE is $\tilde{q}_{p}< \frac{t_N+2t_{NoN}}{\kappa_u+\kappa_{ad}}$, and verify that these strategies yield $z^{eq}=1$.
	In case B-2, we characterize the necessary and sufficient condition by which
	these is no unilateral profitable deviation for ISPs.
	
	\textbf{Case B-1:} Note that in this region,  by item 3 of Theorem~\ref{theorem:p_tilde_new}, if $z^{eq}=1$, then $(q^{eq}_N,q^{eq}_{NoN})=(0,\tilde{q}_{p})\in F^I_1$. In addition, by Theorem~\ref{theorem:NE_stage2_new_suff}, $\tilde{p}^{eq}=\tilde{p}_{t,2}=\kappa_{ad} (n_{NoN}-\frac{\tilde{q}_f}{\tilde{q}_{p}})$ and $n_{NoN}=\frac{t_N+\kappa_u \tilde{q}_{p}-\Delta p}{t_N+t_{NoN}}$ (by \eqref{equ:EUs_linear}). Thus, by \eqref{equ:payoffISPsGeneral_new}, the payoff of ISP NoN in this region is $\pi_{NoN,B}(p_{NoN},\tilde{p}_{t,2})=(p_{NoN}-c)n_{NoN}+\tilde{p}_{t,2} \tilde{q}_{p}$, and the payoff of ISP N is  $\pi_{N,B}=(p_{N}-c)(1-n_{NoN})$. Note that $\tilde{p}_{t,2} \tilde{q}_{p}=\kappa_{ad}(\tilde{q}_{p} n_{NoN}-\tilde{q}_f)$. Thus, using the expression of $n_{NoN}$, the payoffs are:
	\be \label{equ:help:theorem_B}
	\small
	\pi_{NoN,B}=(p_{NoN}-c+\kappa_{ad}\tilde{q}_{p})(\frac{t_N+\kappa_u \tilde{q}_{p}+p_N-p_{NoN}}{t_N+t_{NoN}})-\kappa_{ad}\tilde{q}_f
	\ee
	\be
	\pi_{N,B}=(p_N-c)(\frac{t_{NoN}-\kappa_u \tilde{q}_{p}+p_{NoN}-p_N}{t_N+t_{NoN}})
	\ee
	Note that any NE inside this region should satisfy the first order optimality condition (note that the payoffs are concave). Thus,
	\be
	\ba
	\frac{d \pi_N}{d p_{N}}=0& \Rightarrow t_{NoN}-\kappa_u \tilde{q}_{p} + p_{NoN}-2 p_N+c=0\\
	\frac{d \pi_{NoN,B}}{d p_{NoN}}=0& \Rightarrow t_N+\tilde{q}_{p}(\kappa_u-\kappa_{ad})+p_N-2 p_{NoN}+c=0
	\ea
	\ee
	\normalsize
	
	Solving the equation yields:
	\be \label{equ:equ:B>_eq_pNoN_new}
	p^{eq}_{NoN}=c+\frac{t_{NoN}+2t_N+  \tilde{q}_{p}(\kappa_u -2\kappa_{ad})}{3}
	\ee
	\be \label{equ:equ:B>_eq_pN_new}
	p^{eq}_{N}=c+\frac{2t_{NoN}+t_N-\tilde{q}_{p}(\kappa_u +\kappa_{ad})}{3}
	\ee

	The equilibrium payoffs for ISP are:
	\be \label{equ:equ:B>_eq_piNoN_new}
	\pi^{eq}_{NoN}=\frac{\big{(}t_{NoN}+2t_N+  \tilde{q}_{p}(\kappa_u +\kappa_{ad})\big{)}^2}{9(t_N+t_{NoN})}-\kappa_{ad}\tilde{q}_f
	\ee
	\be \label{equ:equ:B>_eq_piN_new}
	\pi^{eq}_{N}=\frac{\big{(}2t_{NoN}+t_N-\tilde{q}_{p}(\kappa_u +\kappa_{ad})\big{)}^2}{9(t_N+t_{NoN})}
	\ee

	Now, we check the necessary conditions for these strategies to be NE. First, note that if $\tilde{q}_{p}>\frac{2t_{NoN}+t_N}{\kappa_u+\kappa_{ad}}$, then $p^{eq}_N<c$, and $p^{eq}_N$ cannot be an NE. Thus, the first necessary condition for these strategies to be NE is $\tilde{q}_{p}\leq \frac{2t_{NoN}+t_N}{\kappa_u+\kappa_{ad}}$. The next necessary condition is that $\Delta p^{eq}=p^{eq}_{NoN}-p^{eq}_N$ to be in region B, i.e.  $\kappa_u \tilde{q}_{p}-t_{NoN}<\Delta p^{eq}<t_N+\kappa_u \tilde{q}_{p}$. We claim that the upperbound always holds. To prove this consider:
	$$
	\ba
	\Delta p^{eq}<t_N+\kappa_u \tilde{q}_{p} &\iff 2t_N+t_{NoN}+\tilde{q}_{p}(\kappa_u +\kappa_{ad})>0
	\ea
	$$
	which is always true. Now, we check the lower bound:
	$$
	\ba
	\kappa_u \tilde{q}_{p}-t_{NoN}<\Delta p^{eq} &\iff \kappa_u \tilde{q}_{p}-t_{NoN}\\
	&\qquad <\frac{1}{3}(t_N-t_{NoN}+\tilde{q}_{p}(2\kappa_u-\kappa_{ad}))\\
	& \iff \tilde{q}_{p}<\frac{t_N+2t_{NoN}}{\kappa_u +\kappa_{ad}}
	\ea
	$$
	Thus, this necessary condition together with the previous necessary condition yields that if $p^{eq}_N$ and $p^{eq}_{NoN}$, NE strategies, then $\tilde{q}_{p}<\frac{t_N+2t_{NoN}}{\kappa_u +\kappa_{ad}}$.
	
	In addition, note that by Lemma~\ref{lemma:appendix_caseB}, $p^{eq}_N$ and $p^{eq}_{NoN}$ indeed yields $z^{eq}=1$.
	
	Thus, if  $\tilde{q}_{p}<\frac{t_N+2t_{NoN}}{\kappa_u +\kappa_{ad}}$, then the  NE  strategies \emph{can} be $p^{eq}_N$ and $p^{eq}_{NoN}$. To prove that these strategies are NE, we should rule out the possibility of a unilateral profitable deviation by both ISPs which we proceed to do in the next case.
	
	\textbf{Case B-2:} In this case, we consider the possibility of a unilateral deviation by ISPs. First, in Case B-2-i, we rule out the possibility of a profitable deviation by the non-neutral ISP, and then in Case B-2-ii, we provide necessary and sufficient condition for a unilateral deviation to be not profitable for the neural ISP.
	
	\textbf{Case B-2-i: } A deviation by the non-neutral ISP can be to regions A, C, and other prices in region B. In the following cases, we prove that a deviation by ISP NoN to each of these regions si not profitable:
	
	\textbf{Case B-2-i-A:} Consider  $p^{eq}_N$ fixed and decreasing $p_{NoN}$ such that $\Delta p$  in regions A, i.e. $\Delta p\leq \kappa_u \tilde{q}_{p}-t_{NoN}$. Note that in A the payoff of the ISP NoN is an increasing function of $p_{NoN}$ (as discussed in Case A). Thus, all other prices are dominated by $p'_{NoN}={p}^{eq}_N+\kappa_u \tilde{q}_{p}-t_{NoN}$. The payoff in this case is $\pi'_{NoN}={p}^{eq}_N+\kappa_u \tilde{q}_{p}-t_{NoN}-c+z\tilde{q}_{p}\tilde{p}_{t,1}$ (by \eqref{equ:payoffISPsGeneral_new}), and $\tilde{p}_{t,1}=\kappa_{ad}(1-\frac{\tilde{q}_f}{\tilde{q}_p})$ (by definition~\ref{def:pt1,pt2}). We claim that this deviation is not profitable for ISP NoN, since:
	$$
	\footnotesize
	\ba
	&	\pi_{NoN}(p_{NoN},\tilde{p}_{t,2})\geq  {p}^{eq}_N-c+ \kappa_u\tilde{q}_{p}-t_{NoN}+\kappa_{ad}(\tilde{q}_{p}-\tilde{q}_f)\\
	&\iff \frac{\big{(}t_{NoN}+2t_N+  \tilde{q}_{p}(\kappa_u +\kappa_{ad})\big{)}^2}{9(t_N+t_{NoN})} \geq  \frac{t_N-t_{NoN}+2\tilde{q}_{p}(\kappa_u +\kappa_{ad})}{3}\\
	& \iff (\tilde{q}_{p}(\kappa_u+\kappa_{ad})-t_N-2t_{NoN})^2\geq 0
	\ea
	$$
	\normalsize
	which is true always. Thus, no deviation is profitable for ISP NoN.
	
	\textbf{Case B-2-i-B:} Now, consider a deviation by ISP NoN inside region B. By optimality of the solution inside B, if $z^{eq}=1$, since $p_N={p}^{eq}_N$ is fixed, all other $p_{NoN}$ such that $\Delta p$ in region B is dominated in payoff by $p_{NoN}={p}^{eq}_{NoN}$. If $p_{NoN}$ is such that $z^{eq}=0$, then $n_{NoN}=0$ (by item 2 of Theorem~\ref{lemma:CP_z=0_new} and $\kappa_u \tilde{q}_{p}-t_{NoN}\geq t_N$). Thus, the payoff of ISP NoN is zero and this deviation is also not profitable.
	
	\textbf{Case B-2-i-C:} In this case, consider a deviation to region C, i.e. $\Delta p\geq t_N+\kappa_u \tilde{q}_p$. Fixing $p^{eq}_N$ and increasing $p_{NoN}$ such that $\Delta p$  in regions C yields a payoff of zero to ISP NoN (since by item 4 of Theorem~\ref{theorem:p_tilde_new}, $z^{eq}=0$ in this region, and by Theorem~\ref{lemma:CP_z=0_new}, $n^{eq}_{NoN}=0$.). Thus, this deviation is also not profitable.
	
	\textbf{Case B-2-ii:} Now, consider a unilateral deviation by the non-neutral ISP. Similar to the case B-2-i, this deviation can be to regions A, C, and other prices in region B:
	
	\textbf{Case B-2-ii-A: } In this case, we consider the possibility of a deviation by ISP N to region A, i.e. $\Delta p\leq \kappa_u \tilde{q}_p-t_{NoN}$. Note that in region A, $\pi_{NoN}(p_{NoN},\tilde{p}_{t,1})>\pi_{NoN,z=0}(p_{NoN},\tilde{p})$, where $\pi_{NoN,z=0}(p_{NoN},\tilde{p})$ is the payoff of ISP NoN when $z^{eq}=0$. To prove this note that by $\tilde{q}_p\tilde{p}_{t,1}=\kappa_{ad}(\tilde{q}_p-\tilde{q}_f)>0$, we can write:
	$$
	\ba
	\pi_{NoN}(p_{NoN},\tilde{p}_{t,1})&=p_{NoN}-c+\tilde{q}_p\tilde{p}_{t,1}\\
	&>p_{NoN}-c>\pi_{NoN,z=0}(p_{NoN},\tilde{p})
	\ea
	$$
	Thus, in region A, $\pi_{NoN}(p_{NoN},\tilde{p}_{t,1})>\pi_{NoN,z=0}(p_{NoN},\tilde{p})$, and by Theorem~\ref{theorem:NE_stage2_new_suff}, $z^{eq}=1$. Thus, using Theorem~\ref{theorem:p_tilde_new}, ib this region $n_{NoN}=1$. Therefore, $n_N=0$, and by \eqref{equ:payoffISPsGeneral_new}, the payoff of ISP N is zero. Thus, a deviation to this region is not profitable.
	
	\textbf{Case B-2-ii-B:} Now, consider a deviation inside region B by ISP N. If $z^{eq}=1$, by optimality of the solution inside B (since $p_N={p}^{eq}_N$ is fixed) all other $p_{N}$ such that $\Delta p$ in region B is dominated in payoff by $p_{N}={p}^{eq}_{N}$.

	Now, consider the case that $p_{N}$ is such that $z^{eq}=0$. In this case, $n_{NoN}=0$ (by item 2 of Theorem~\ref{lemma:CP_z=0_new} and $\kappa_u \tilde{q}_{p}-t_{NoN}\geq t_N$), and such a deviation might be profitable.
	
	
	In order to have $z^{eq}=0$, by Theorem~\ref{theorem:NE_stage2_new_suff}, $\pi_{NoN}(p^{eq}_{NoN},\tilde{p}_{t,2})\leq \pi_{NoN,z=0}(p^{eq}_{NoN},\tilde{p})$, where $\pi_{NoN,z=0}(p^{eq}_{NoN},\tilde{p})$ is the payoff when $z^{eq}=0$. Note that by the assumption of the theorem ($\kappa_u\tilde{q}_p\geq t_N+t_{NoN}$), and in this region $\Delta p>\kappa_u \tilde{q}_p-t_{NoN}\geq t_N$. Thus, by Theorem~\ref{lemma:CP_z=0_new}, if $z^{eq}=0$, then $n_{NoN}=0$. Therefore, by \eqref{equ:payoffISPsGeneral_new}, $\pi_{NoN,z=0}(p^{eq}_{NoN},\tilde{p})=0$. Using \eqref{equ:help:theorem_B}, we can find $\pi_{NoN}(p^{eq}_{NoN},\tilde{p}_{t,2})$, and compare the payoffs:
	\small
	$$
	\ba
	&\pi_{NoN}(p^{eq}_{NoN},\tilde{p}_{t,2})\leq \pi_{NoN,z=0}(p^{eq}_{NoN},\tilde{p})\iff\\
	&(p^{eq}_{NoN}-c+\kappa_{ad}\tilde{q}_{p})(\frac{t_N+\kappa_u \tilde{q}_{p}+p'_N-p^{eq}_{NoN}}{t_N+t_{NoN}})-\kappa_{ad}\tilde{q}_f\leq 0\\
	& \iff p'_N\leq \frac{\kappa_{ad}\tilde{q}_f (t_N+t_{NoN})}{p^{eq}_{NoN}-c+\kappa_{ad} \tilde{q}_p}+p^{eq}_{NoN}-t_{NoN}-\kappa_u \tilde{q}_p=p^d_{t}
	\ea
	$$
	\normalsize
	
	Therefore, a deviation is only profitable if $p'_N\leq p^d_t$. If this condition holds, we need to check whether this deviation is indeed profitable. Note that in region B, if $z^{eq}=0$, (as explained before) by Theorem~\ref{lemma:CP_z=0_new}, $n_N=1$. Thus, by \eqref{equ:payoffISPsGeneral_new}, the payoff of ISP N is an increasing function of $p_N$, and is equal to $p'_N-c$. Thus, $p'_N=p^d_t$ yields the maximum payoff after deviation. Therefore, such a  deviation is not profitable if and only if $p^d_t-c\leq \pi_N(p^{eq}_N)$.

	\textbf{Case B-2-ii-C:} Now, consider a deviation by ISP N to region C, i.e. $\Delta p\geq \kappa_u \tilde{q}_p+t_N$. Note that in region C, $z^{eq}=0$, and by item 2 of Theorem~\ref{lemma:CP_z=0_new}, $n_N=1$. Thus, the payoff of ISP N \eqref{equ:payoffISPsGeneral_new} is an increasing function of $p_N$. Thus, $p'_N=p^{eq}_{NoN}-\kappa_u \tilde{q}_p-t_N$ (by definition of region C) yields the highest payoff after deviation. Note that by \eqref{equ:equ:B>_eq_pNoN_new}, $p^{eq}_{NoN}=c+\frac{t_{NoN}+2t_N+  \tilde{q}_{p}(\kappa_u -2\kappa_{ad})}{3}
	$. Therefore, $p'_N=c+\frac{t_{NoN}-t_N-2\tilde{q}_p (\kappa_u+\kappa_{ad}) }{3}$. In addition, note that by the assumption of the theorem, $\kappa_u\tilde{q}_p\geq t_N+t_{NoN}$. Thus, $p'_N<c$, and by \eqref{equ:payoffISPsGeneral_new}, the payoff of neutral ISP is negative. Thus, this deviation is  not profitable.
	
	Therefore, we only need to check the condition in Case B-2-ii-B for ruling out profitable deviations.  This is item 2 of the theorem. The theorem follows.
\end{proof}

\section{Proof of Theorem~\ref{theorem:NE_stage1_new_q<}}\label{appendix:theorem:NE_stage1_new_q<}
We first state and prove a Lemma that we use to prove Theorem~\ref{theorem:NE_stage1_new_q<}.


We now prove Theorem~\ref{theorem:NE_stage1_new_q<}:
\begin{proof}
Part~(\ref{i1}) follows from Theorem~\ref{theorem:NE_stage1_new_q>}-1 because the region   $ t_N+2t_{NoN} \leq \tilde{q}_{p}(\kappa_u+\kappa_{ad})$ and $t_{NoN} < \kappa_u \tilde{q}_p+\kappa_{ad}(\tilde{q}_p-\tilde{q}_f)$ in general has a non-empty intersection with that represented by $t_N+t_{NoN} >  \kappa_u \tilde{q}_p$. Clearly the access fees in part~(\ref{i1}) satisfy  $\Delta p\leq \kappa_u \tilde{q}_{p}-t_{NoN}$, and  were obtained in the proof of Theorem~\ref{theorem:NE_stage1_new_q>}-1 (in Appendix ) by considering   $\Delta p\leq \kappa_u \tilde{q}_{p}-t_{NoN}$, which we denote as  region A.

  Moving from region A, we denote $\kappa_u \tilde{q}_{p}-t_{NoN}<\Delta p<\kappa_u (2\tilde{q}_{p}-\tilde{q}_f)-t_{NoN}$ by region $B_1$, $\kappa_u (2\tilde{q}_{p}-\tilde{q}_f)-t_{NoN}\leq \Delta p<t_N+\kappa_u(\tilde{q}_{p}-\tilde{q}_f)$ by region C, $t_N+\kappa_u(\tilde{q}_{p}-\tilde{q}_f)\leq \Delta p<t_N+\kappa_u \tilde{q}_{p}$ by set $B_2$, and $\Delta p\geq t_N+\kappa_u \tilde{q}_{p}$ by D. Using Theorem~\ref{theorem:p_tilde_new}, if $z^{eq}=1$, then $\Delta p<t_N+\kappa_u \tilde{q}_{p}$. Thus, we characterize any possible SPNE strategies by which $z^{eq}=1$, in regions $B_1$, $C$, and $B_2$:


Note that by Lemma~\ref{lemma:deltap_t}, since $\tilde{q}_{p}< \frac{t_N+t_{NoN}}{\kappa_u}$,  $\kappa_u \tilde{q}_{p}-t_{NoN} <  \Delta p_t < t_N+\kappa_u(\tilde{q}_{p}-\tilde{q}_f)$, where $\Delta p_{t}=\kappa_{u}(2\tilde{q}_{p}-\tilde{q}_f)-t_{NoN}$. Also, $\tilde{q}_f<\tilde{q}_{p}<\frac{t_N+t_{NoN}}{\kappa_u}$. Thus, using this result, we characterize  the regions characterized in item 2  of Theorem~\ref{theorem:p_tilde_new}.

\textbf{Cases $B_1$ and $B_2$:} Now, consider regions $B_1$ and $B_2$, i.e. $\kappa_u \tilde{q}_{p}-t_{NoN}<\Delta p<\kappa_u (2\tilde{q}_{p}-\tilde{q}_f)-t_{NoN}$ and $t_N+\kappa_u(\tilde{q}_{p}-\tilde{q}_f)\leq \Delta p<t_N+\kappa_u \tilde{q}_{p}$, respectively.

Note that in these regions,  by items 2-a-ii and 2-b of Theorem~\ref{theorem:p_tilde_new}, if $z^{eq}=1$, then $(q^{eq}_N,q^{eq}_{NoN})=(0,\tilde{q}_{p})\in F^I_1$. In addition, by Theorem~\ref{theorem:NE_stage2_new_suff}, $\tilde{p}^{eq}=\tilde{p}_{t,2}=\kappa_{ad} (n_{NoN}-\frac{\tilde{q}_f}{\tilde{q}_{p}})$ and $n_{NoN}=\frac{t_N+\kappa_u \tilde{q}_{p}-\Delta p}{t_N+t_{NoN}}$ (Definition~\ref{def:pt1,pt2}). Thus, by \eqref{equ:payoffISPsGeneral_new}, the payoff of ISP NoN in this region is $\pi_{NoN,B}(p_{NoN},\tilde{p}_{t,2})=(p_{NoN}-c)n_{NoN}+\tilde{p}_{t,2} \tilde{q}_{p}$, and the payoff of ISP N is  $\pi_{N,B}=(p_{N}-c)(1-n_{NoN})$. Note that $\tilde{p}_{t,2} \tilde{q}_{p}=\kappa_{ad}(\tilde{q}_{p} n_{NoN}-\tilde{q}_f)$. Thus, using the expression of $n_{NoN}$, the payoffs are:

\footnotesize
\be\label{equ:equ:Theorem7_help_2}
\ba
\pi_{NoN,B}(p_{NoN},\tilde{p}_{t,2})&=(p_{NoN}-c+\kappa_{ad}\tilde{q}_{p})(\frac{t_N+\kappa_u \tilde{q}_{p}+p_N-p_{NoN}}{t_N+t_{NoN}})\\
&\qquad \qquad \qquad -\kappa_{ad}\tilde{q}_f
\ea
\ee
\be
\pi_{N,B}(p_N)=(p_N-c)(\frac{t_{NoN}-\kappa_u \tilde{q}_{p}+p_{NoN}-p_N}{t_N+t_{NoN}})
\ee
\normalsize

First, we rule out any SPNE such that $\Delta p^{eq}=t_N+\kappa_u(\tilde{q}_{p}-\tilde{q}_f)$. Suppose that $\Delta p^{eq}=p^{eq}_{NoN}-p^{eq}_{N}=t_N+\kappa_u(\tilde{q}_{p}-\tilde{q}_f)$. Consider a deviation by ISP N such that $p'_N=p^{eq}_{N}+\epsilon>c$ for $\epsilon>0$ such that $\Delta p'=p^{eq}_{NoN}-p'_N$ to be in region C. Note that by item 2-a-i of Theorem~\ref{theorem:p_tilde_new}, in region C, $(q^{eq}_N,q^{eq}_{NoN})=(\tilde{q}_f,\tilde{q}_p)\in F_1^I$. Thus, the payoff of this ISP with this deviation is (by \eqref{equ:UN_new}):

\small
\be \nonumber
\pi_{N}(p'_N)=(p^{eq}_N+\epsilon-c)(\frac{t_{NoN}-\kappa_u (\tilde{q}_{p}-\tilde{q}_f)+p^{eq}_{NoN}-p^{eq}_N-\epsilon}{t_N+t_{NoN}})
\ee
\normalsize

Note that $\lim_{\epsilon\downarrow 0} \pi_N(p'_N)>\pi_{N,B}(p^{eq}_N)$. Thus, for $\epsilon>0$ small enough, this deviation is profitable. Thus, the strategies by which $\Delta p^{eq}=t_N+\kappa_u(\tilde{q}_{p}-\tilde{q}_f)$ cannot be SPNE.

Now, we characterize any SPNE in    $\kappa_u \tilde{q}_{p}-t_{NoN}<\Delta p<\kappa_u (2\tilde{q}_{p}-\tilde{q}_f)-t_{NoN}$ and $t_N+\kappa_u(\tilde{q}_{p}-\tilde{q}_f)<\Delta p<t_N+\kappa_u \tilde{q}_{p}$. Any SPNE inside this region should satisfy the first order optimality condition (note that the payoffs are concave). Thus,

\small
\be
\ba
 \frac{d \pi_N}{d p_{N}}=0& \Rightarrow t_{NoN}-\kappa_u \tilde{q}_{p} + p_{NoN}-2 p_N+c=0\\
\frac{d \pi_{NoN,B}}{d p_{NoN}}=0& \Rightarrow t_N+\tilde{q}_{p}(\kappa_u-\kappa_{ad})+p_N-2 p_{NoN}+c=0
\ea
\ee
\normalsize

Solving the equation yields:

\small
\be 
p^{eq}_{NoN}=c+\frac{t_{NoN}+2t_N+  \tilde{q}_{p}(\kappa_u -2\kappa_{ad})}{3}
\ee
\be 
p^{eq}_{N}=c+\frac{2t_{NoN}+t_N-\tilde{q}_{p}(\kappa_u +\kappa_{ad})}{3}
\ee
\normalsize


If $\tilde{q}_{p}>\frac{2t_{NoN}+t_N}{\kappa_u+\kappa_{ad}}$, then $p^{eq}_N<c$, and $p^{eq}_N$ cannot be an SPNE. Thus, the first necessary condition for these strategies to be SPNE is  $\tilde{q}_{p}\leq \frac{2t_{NoN}+t_N}{\kappa_u+\kappa_{ad}}$.
 In addition, by Theorem~\ref{theorem:NE_stage2_new_suff},\footnote{Note that in Regions $B_1$ and $B_2$, $\Delta p<t_N+\kappa_u \tilde{q}_p$.} $\pi_{NoN}(p_{NoN}^{eq}, \tilde{p}_{t,2})>\pi_{NoN,z=0}(\tilde{p}^{eq}_{NoN},\tilde{p})$ (for these strategies to yield $z^{eq}=1$).  The second item of the theorem follows.

\textbf{Case C: } Now, consider region C, i.e. $\Delta p_t=\kappa_u (2\tilde{q}_{p}-\tilde{q}_f)-t_{NoN}\leq \Delta p<t_N+\kappa_u (\tilde{q}_{p}-\tilde{q}_f)$. Note that in this regions,  by items 2-a-i of Theorem~\ref{theorem:p_tilde_new}, if $z^{eq}=1$, then $(q^{eq}_N,q^{eq}_{NoN})=(\tilde{q}_f,\tilde{q}_{p})\in F^I_1$. In addition, by Theorem~\ref{theorem:NE_stage2_new_suff} and Definition \ref{def:pt}, $\tilde{p}^{eq}=\tilde{p}_{t,3}=\kappa_{ad}n_{NoN} (1-\frac{\tilde{q}_f}{\tilde{q}_{p}})$ and $n_{NoN}=\frac{t_N+\kappa_u (\tilde{q}_{p}-\tilde{q}_f)-\Delta p}{t_N+t_{NoN}}$ (Definition~\ref{def:pt1,pt2}). Thus, by \eqref{equ:payoffISPsGeneral_new}, the payoff of ISP NoN in this region is $\pi_{NoN,C}(p_{NoN},\tilde{p}_{t,3})=(p_{NoN}-c)n_{NoN}+\tilde{p}_{t,3} \tilde{q}_{p}$, and the payoff of ISP N is  $\pi_{N,B}=(p_{N}-c)(1-n_{NoN})$. 
Note that $\tilde{p}_{t,3} \tilde{q}_{p}=\kappa_{ad}n_{NoN}(\tilde{q}_{p} -\tilde{q}_f)$. Thus, using the expression of $n_{NoN}$, the payoffs are:

\footnotesize
\be\label{equ:equ:Theorem7_help}
\ba
\pi_{NoN,C}&(p_{NoN},\tilde{p}_{t,3})=\\
&(p_{NoN}-c+\kappa_{ad}(\tilde{q}_{p}-\tilde{q}_f))(\frac{t_N+\kappa_u (\tilde{q}_{p}-\tilde{q}_f)+p_N-p_{NoN}}{t_N+t_{NoN}})
\ea
\ee
\be\label{equ:equ:Theorem7_help_5}
\pi_{N,C}(p_N)=(p_N-c)(\frac{t_{NoN}-\kappa_u (\tilde{q}_{p}-\tilde{q}_f)+p_{NoN}-p_N}{t_N+t_{NoN}})
\ee
\normalsize

First, in Part C-1, we characterize any SPNE in region  $\kappa_u (2\tilde{q}_{p}-\tilde{q}_f)-t_{NoN}< \Delta p<t_N+\kappa_u (\tilde{q}_{p}-\tilde{q}_f)$. Later, in Part C-2, we consider the case that $\Delta p^{eq}=\kappa_u(2\tilde{q}_{p}-\tilde{q}_f)-t_{NoN}$.

\textbf{Part C-1:}
 Note that any SPNE in region  $\kappa_u (2\tilde{q}_{p}-\tilde{q}_f)-t_{NoN}< \Delta p<t_N+\kappa_u (\tilde{q}_{p}-\tilde{q}_f)$ should satisfy the first order optimality condition (note that the payoffs are concave). Thus,

 \small
\be
\ba \label{equ:equ:Theorem7_help_4}
 \frac{d \pi_{N,C}}{d p_{N}}=0& \Rightarrow t_{NoN}-\kappa_u (\tilde{q}_{p}-\tilde{q}_f) + p_{NoN}-2 p_N+c=0\\
\frac{d \pi_{NoN,C}}{d p_{NoN}}=0& \Rightarrow t_N+(\tilde{q}_{p}-\tilde{q}_f)(\kappa_u-\kappa_{ad})+p_N-2 p_{NoN}+c=0
\ea
\ee
\normalsize

Solving the equation yields:

\small
\be 
p^{eq}_{NoN}=c+\frac{t_{NoN}+2t_N+  (\tilde{q}_{p}-\tilde{q}_f)(\kappa_u -2\kappa_{ad})}{3}
\ee
\be 
p^{eq}_{N}=c+\frac{2t_{NoN}+t_N-(\tilde{q}_{p}-\tilde{q}_f)(\kappa_u +\kappa_{ad})}{3}
\ee
\normalsize

First, note that if $\tilde{q}_{p}-\tilde{q}_f>\frac{2t_{NoN}+t_N}{\kappa_u+\kappa_{ad}}$, then $p^{eq}_N<c$, and $p^{eq}_N$ cannot be an SPNE. Thus, the necessary condition for these strategies to be SPNE is $\tilde{q}_{p}-\tilde{q}_f\leq \frac{2t_{NoN}+t_N}{\kappa_u+\kappa_{ad}}$.
In addition, by Theorem~\ref{theorem:NE_stage2_new_suff},   $\pi^{eq}_{NoN}(\tilde{p}^{eq}_{NoN},\tilde{p}_{t,3})>\pi_{NoN,z=0}(\tilde{p}^{eq}_{NoN},\tilde{p})$ (in order for these strategies to yield $z^{eq}=1$).
 The third item of the theorem follows.


\textbf{Part C-2:} Now, consider $p^{eq}_N$ and $p^{eq}_{NoN}$ such that $\Delta p^{eq}=p^{eq}_{NoN}-p^{eq}_{N}=\kappa_u(2\tilde{q}_{p}-\tilde{q}_f)-t_{NoN}$. These strategies are not SPNE if ISP NoN can strictly increase her payoff by decreasing her price such that $\Delta p$ in region $B_1$. Note that using \eqref{equ:equ:Theorem7_help} and the expression for $
\Delta p^{eq}$, NoN's payoff  is:

\vspace{-3mm}
\small
\be\label{equ:equ:Theorem7_help_3}
\pi_{NoN}(p^{eq}_{NoN},\tilde{p}_{t,3})=(p_{NoN}-c+\kappa_{ad}(\tilde{q}_{p}-\tilde{q}_f))(\frac{t_N-\kappa_u \tilde{q}_{p}+t_{NoN}}{t_N+t_{NoN}})
\ee
\normalsize

By choosing $p'_{NoN}=p^{eq}_{NoN}-\epsilon$ such that $\epsilon\downarrow 0$, ISP NoN can get a limit payoff of (since $\Delta p=\Delta p^{eq}$ when $\epsilon\rightarrow 0$, and it is in  region $B_1$, and using \eqref{equ:equ:Theorem7_help_2}):

\footnotesize
 $$
 \ba
\pi'_{NoN}&=\lim_{\epsilon\downarrow 0} \pi^{eq}_{NoN}(p_{NoN}-\epsilon,\tilde{p}_{t,3})\\
&=(p_{NoN}^{eq}-c+\kappa_{ad}\tilde{q}_{p})(\frac{t_N-\kappa_u (\tilde{q}_{p}-\tilde{q}_f)+t_{NoN}}{t_N+t_{NoN}})-\kappa_{ad}\tilde{q}_f
\ea
$$
\normalsize

Thus, $p^{eq}_N$ and $p^{eq}_{NoN}$ such that $\Delta p^{eq}=p^{eq}_{NoN}-p^{eq}_{N}=\kappa_u(2\tilde{q}_{p}-\tilde{q}_f)-t_{NoN}$ are not SPNE if:
\footnotesize
$$
\ba
\pi'_{NoN}&>\pi_{NoN}(p^{eq}_{NoN},\tilde{p}_{t,3})\\
&\iff (p^{eq}_{NoN}-c+\kappa_{ad}\tilde{q}_{p})\frac{\kappa_u \tilde{q}_f}{t_N+t_{NoN}}-\frac{\kappa_{ad}\kappa_u\tilde{q}_f\tilde{q}_{p}}{t_N+t_{NoN}}>0\\
&\iff p^{eq}_{NoN}>c
\ea
$$
\normalsize
Thus, the necessary condition for these strategy to be SPNE is $p^{eq}_{NoN}\leq c$.
 Note that from \eqref{equ:equ:Theorem7_help} and  \eqref{equ:equ:Theorem7_help_5}, since $\Delta p$ is fixed, the payoffs of ISP NoN and N are an increasing function of $p_{NoN}$ and $p_{N}$, respectively. Thus, $p^{eq}_{NoN}=c$, and $p^{eq}_N=c-\kappa_u(2\tilde{q}_{p}-\tilde{q}_f)+t_{NoN}$. Note that a necessary condition for $p^{eq}_N$ to be an SPNE is that $p^{eq}_N\geq c$. Thus,   $2\tilde{q}_{p}-\tilde{q}_f\leq \frac{t_{NoN}}{\kappa_{u}}$ is a necessary condition.  In addition,  $\pi_{NoN}(\tilde{p}^{eq}_{NoN},\tilde{p}_{t,3})>\pi_{NoN,z=0}(\tilde{p}^{eq}_{NoN},\tilde{p})$ (using Theorem~\ref{theorem:NE_stage2_new_suff}, in order for these strategies to yield $z^{eq}=1$).
  The fourth item of the theorem follows.
\end{proof}

\section{Proof of Theorem~\ref{theorem:bigt}}\label{appendix:theorem:bigt}
\begin{proof}
	We use Theorem~\ref{theorem:NE_stage1_new_q<} to prove the result. First, in Part 1, we prove that when one of $t_N$ or $t_{NoN}$ is large, then  strategies 1), 2), and 4) listed in Theorem~\ref{theorem:NE_stage1_new_q<} are not NE. In Part 2, we prove that when one of $t_N$ or $t_{NoN}$ is high, then strategy 3) of Theorem~\ref{theorem:NE_stage1_new_q<} is an NE. This completes the proof of the theorem.
	
	\textbf{Part 1:} We prove that strategies 1), 2), and 4) listed in Theorem~\ref{theorem:NE_stage1_new_q<} are not NE in Parts 1-i, 1-ii, and 1-iii, respectively.
	
\textbf{Part 1-i: } In this part, we prove that, item 1 of Theorem~\ref{theorem:NE_stage1_new_q<}, i.e. $p^{eq}_{NoN}=c+\kappa_u \tilde{q}_p-t_{NoN}$ and $p^{eq}_N=c$ is not an NE. We do so in Parts 1-i-a and 1-i-b, by introducing a unilateral profitable deviation for ISP NoN for the cases that $t_{NoN}$ is large and  $t_{N}$ is large, respectively. Note that in this case, by item 1 of Theorem~\ref{theorem:p_tilde_new}, $(q^{eq}_N,q^{eq}_{NoN})\in (0,\tilde{q}_f)\in F^L_1$. Thus, $n_{NoN}=1$, and the payoff of ISP NoN  is (by \eqref{equ:payoffISPsGeneral_new}, Theorem~\ref{theorem:NE_stage2_new_suff}, and Definition \ref{def:pt1,pt2}):
\be \label{equ:theorem:larget1}
\pi_{NoN}=\kappa_u \tilde{q}_p-t_{NoN}+\kappa_{ad}(\tilde{q}_p-\tilde{q}_f)
\ee
\textbf{Part 1-i-a:} If $t_{NoN}$ is large, then \eqref{equ:theorem:larget1} would be less than zero. A deviation to price $p'_{NoN}=c$ yields a payoff of at least zero for the ISP NoN (by \eqref{equ:payoffISPsGeneral_new}). Thus, this is a profitable deviation. \\
\textbf{Part 1-i-b: } Now, consider $t_N$ to be large, and a deviation by ISP NoN such that $p'_{NoN}=\frac{1}{2}t_N$ (Thus, $\Delta p=p'_{NoN}-p^{eq}_N=\frac{1}{2}t_N-c$). Note that in this case, $\Delta p_t=\kappa_u(2\tilde{q}_p-\tilde{q}_f)-t_{NoN}<\Delta p<t_N+\kappa_u(\tilde{q}_p-\tilde{q}_f)$. Thus, by item 2-a-i of
Theorem~\ref{theorem:p_tilde_new}, $(q^{eq}_n,q^{eq}_{NoN})=(\tilde{q}_f,\tilde{q}_p)\in F^I_1$.
Thus, by \eqref{equ:payoffISPsGeneral_new}, the payoff of ISP NoN after deviation is  at least\footnote{Note that the payoff of NoN is equal to the maximum of the payoff when $\tilde{p}^{eq}=\tilde{p}_t$ and when $\tilde{p}^{eq}>\tilde{p}_t$, i.e. when $z^{eq}=0$.} (by the definition of $\tilde{p}_{t,3}$ in Definition~\ref{def:pt1,pt2} and Theorem~\ref{theorem:NE_stage2_new_suff}):
\be \label{equ:thoerem:larget2}
\pi'_{NoN}=\frac{1}{2}t_Nn_{NoN}+\kappa_{ad}n_{NoN}(\tilde{q}_p-\tilde{q}_f)
\ee
, where $n_{NoN}=\frac{\frac{1}{2}t_N+\kappa_u(\tilde{q}_p-\tilde{q}_f)+c}{t_N+t_{NoN}}$. Thus, for $t_N$ large, $n_{NoN}\rightarrow \frac{1}{2}$. Thus, comparing \eqref{equ:thoerem:larget2} with \eqref{equ:theorem:larget1} yields:
$$
\pi'_{NoN}=\frac{1}{4}t_N+\frac{1}{2}\kappa_{ad}(\tilde{q}_p-\tilde{q}_f)>\pi_{NoN} \qquad \text{since $t_N$ is large}
$$
Thus, this deviation is  profitable .

\textbf{Part 1-ii:} In this part, we prove that item 2 of Theorem~\ref{theorem:NE_stage1_new_q<}, i.e. $p^{eq}_{NoN}=c+\frac{t_{NoN}+2t_N+\tilde{q}_{p}(\kappa_u -2\kappa_{ad})}{3}$ and $p^{eq}_{N}=c+\frac{2t_{NoN}+t_N-\tilde{q}_{p}(\kappa_u +\kappa_{ad})}{3}$ is not an NE. We do so by proving that $\Delta p^{eq}$ does not satisfy   $\kappa_u \tilde{q}_{p}-t_{NoN}<\Delta p^{eq}<\kappa_u (2\tilde{q}_{p}-\tilde{q}_f)-t_{NoN}$ and $t_N+\kappa_u(\tilde{q}_{p}-\tilde{q}_f)<\Delta p^{eq}<t_N+\kappa_u \tilde{q}_{p}$,
in the cases that $t_{NoN}$ or $t_{N}$ is large.

First, note that:
\begin{equation}
\Delta p^{eq}=p^{eq}_{NoN}-p^{eq}_N=\frac{1}{3}(t_N-t_{NoN}+\tilde{q}_p(2\kappa_u -\kappa_{ad}))
\end{equation}
If $\Delta p^{eq}<\kappa_u(2\tilde{q}_p-\tilde{q}_f)-t_{NoN}$, then
$t_N+2t_{NoN}<3\kappa_u (2\tilde{q}_p-\tilde{q}_f)-\tilde{q}_p(2\kappa_u -\kappa_{ad})$, which is not correct when $t_{NoN}$  or $t_N$ is large. Thus, (a) $\Delta p^{eq}\geq \kappa_u(2\tilde{q}_p-\tilde{q}_f)-t_{NoN}$. In addition, if $t_N+\kappa_u(\tilde{q}_p-\tilde{q}_f)<\Delta p^{eq}$, then $2t_N+t_{NoN}<\tilde{q}_p(2\kappa_u -\kappa_{ad})-3\kappa_u(\tilde{q}_p-\tilde{q}_f)$,  which is not correct when $t_{NoN}$  or $t_N$ is large. Thus,   (b) $\Delta p^{eq}\leq t_N+\kappa_u(\tilde{q}_p-\tilde{q}_f)$. Therefore, (a) and (b) yields that $\Delta p^{eq}$ is not in the regions specified.  Thus, item 2 cannot be an NE.

\textbf{Part 1-iii:} In this part, we prove that item 4 of Theorem~\ref{theorem:NE_stage1_new_q<}, i.e. $p^{eq}_{NoN}=c$ and $p^{eq}_N=c-\kappa_u(2\tilde{q}_{p}-\tilde{q}_f)+t_{NoN}$ is not an NE. To do so, we prove that there exists a profitable unilateral deviation for ISP NoN. Note that, in this  case, $\Delta p^{eq}=\Delta p_t$. By item 2-a-i of Theorem~\ref{theorem:p_tilde_new}, when $\Delta p_t\leq \Delta p <t_N+\kappa_u(\tilde{q}_p-\tilde{q}_f)$, then   $(q^{eq}_N,q^{eq}_{NoN})=(\tilde{q}_f,\tilde{q}_p)\in F^I_1$. Thus, the expression of the payoff of ISP NoN is (by  $\tilde{p}_t=\tilde{p}_{t,3}$, Definition \ref{def:pt1,pt2}, Theorem~\ref{theorem:NE_stage2_new_suff}, and \eqref{equ:UNoN_new}):

\footnotesize
$$
\ba
\pi_{NoN,C}&(p_{NoN},\tilde{p}_{t,3})\\
&=(p_{NoN}-c+\kappa_{ad}(\tilde{q}_{p}-\tilde{q}_f))(\frac{t_N+\kappa_u (\tilde{q}_{p}-\tilde{q}_f)+p_N-p_{NoN}}{t_N+t_{NoN}})
\ea
$$
\normalsize

Note that:

\footnotesize
$$
\frac{d \pi_{NoN,C}}{d p_{NoN}}= \frac{t_N+(\tilde{q}_{p}-\tilde{q}_f)(\kappa_u-\kappa_{ad})+p_N-2 p_{NoN}+c}{t_N+t_{NoN}}
$$
\normalsize
Thus,
\footnotesize
$$
\ba
\frac{d \pi_{NoN,C}}{d p_{NoN}}&|_{p^{eq}_N,p^{eq}_{NoN}}\\
&=\frac{t_N+t_{NoN}+(\tilde{q}_{p}-\tilde{q}_f)(\kappa_u-\kappa_{ad})-\kappa_u(2\tilde{q}_p-\tilde{q}_f)}{t_N+t_{NoN}}
\ea
$$
\normalsize
Note that $\frac{d \pi_{NoN,C}}{d p_{NoN}}|_{p^{eq}_N,p^{eq}_{NoN}}>0$, when either $t_N$ or $t_{NoN}$ are large enough. Thus, in this case, the payoff is increasing with respect to $p_{NoN}$\footnote{Note that after this deviation, $\Delta p$ remains in the same region.}. Thus, $p'_{NoN}=p^{eq}_{NoN}+\epsilon$ for $\epsilon>0$ small, is a unilateral profitable deviation.

\textbf{Part 2:} We now prove that when one of $t_N$ or $t_{NoN}$ is large, then strategy 3) of Theorem~\ref{theorem:NE_stage1_new_q<} is an NE. To do so, we check conditions (i), (ii), and (iii) of strategy 3) of Theorem~\ref{theorem:NE_stage1_new_q<}, in Parts 2-i, 2-ii, and 2-iii, respectively. Later, in Part 2-iv, we prove that there is no unilateral profitable deviation for ISPs.  This completes the proof. \\
\textbf{Part 2-i:} In this part, we check the condition, i.e. $\kappa_u (2\tilde{q}_{p}-\tilde{q}_f)-t_{NoN}< \Delta p^{eq}<t_N+\kappa_u (\tilde{q}_{p}-\tilde{q}_f)$. Note that in this case:
\be
\Delta p^{eq}=\frac{1}{3}(t_N-t_{NoN}+(\tilde{q}_p-\tilde{q}_f)(2\kappa_u-\kappa_{ad}))
\ee
Comparing the lower boundary yields that:
\footnotesize
$$
\ba
\kappa_u &(2\tilde{q}_{p}-\tilde{q}_f)-t_{NoN}< \Delta p^{eq}\\
&\Rightarrow 2t_{NoN}+t_N+(\tilde{q}_p-\tilde{q}_f)(2\kappa_u-\kappa_{ad})-3\kappa_u (2\tilde{q}_{p}-\tilde{q}_f)>0
\ea
$$
\normalsize
which is true when one of $t_N$ or $t_{NoN}$ is large. Now, consider the upper boundary:
\footnotesize
$$
\ba
\Delta p^{eq}&<t_N+\kappa_u (\tilde{q}_{p}-\tilde{q}_f)\\
&\Rightarrow 2t_N+t_{NoN}+\kappa_u (\tilde{q}_{p}-\tilde{q}_f) -(\tilde{q}_p-\tilde{q}_f)(2\kappa_u-\kappa_{ad})>0
\ea
$$
\normalsize
which is true when one of $t_N$ or $t_{NoN}$ is large. Thus, condition (i)  of strategy 3) of Theorem~\ref{theorem:NE_stage1_new_q<} is true.\\
\textbf{Part 2-ii:} Condition (ii) of this strategy is $\tilde{q}_p-\tilde{q}_f\leq \frac{2t_{NoN}+t_N}{\kappa_u+\kappa_{ad}}$. This condition holds when one of $t_N$ or $t_{NoN}$ is large. \\
\textbf{Part 2-iii:} Now, we check the third condition, i.e. $\pi^{eq}_{NoN}=\pi_{NoN}(\tilde{p}^{eq}_{NoN},\tilde{p}_{t,3})>\pi_{NoN,z=0}(\tilde{p}^{eq}_{NoN},\tilde{p})$. We use \eqref{equ:payoffISPsGeneral_new} to find $\pi^{eq}_{NoN}=\pi_{NoN}(\tilde{p}^{eq}_{NoN},\tilde{p}_{t,3})$. Note that by using item 2-a-i of Theorem~\ref{theorem:p_tilde_new} (since $z^{eq}=1$), $(q^{eq}_N,q^{eq}_{NoN})=(\tilde{q}_f,\tilde{q}_p)$. Thus, by the definition of
$p^{eq}_{NoN}$, $\Delta p^{eq}$, $\tilde{p}_{t,3}$, and using Definition \ref{def:pt1,pt2}, Theorem~\ref{theorem:NE_stage2_new_suff}:

\footnotesize
\be \label{equ:equ:B>_eq_piNoN_newr}
\pi^{eq}_{NoN}=\frac{\big{(}t_{NoN}+2t_N+  (\tilde{q}_{p}-\tilde{q}_f)(\kappa_u +\kappa_{ad})\big{)}^2}{9(t_N+t_{NoN})}
\ee
\normalsize

	Now, we obtain $\pi_{NoN,z=0}(\tilde{p}^{eq}_{NoN},\tilde{p})$. Consider the case that $\tilde{p}$ is such that $z^{eq}=0$. Note that since  $\kappa_u (2\tilde{q}_{p}-\tilde{q}_f)-t_{NoN}< \Delta p^{eq}<t_N+\kappa_u (\tilde{q}_{p}-\tilde{q}_f)$, then  $-t_{NoN}<\Delta p^{eq}<t_N$ or $\Delta p^{eq}\geq t_N$. Using item 2 of Theorem~\ref{lemma:CP_z=0_new}, if $\Delta p^{eq}\geq t_N$, then $n_{NoN}=0$, and by \eqref{equ:payoffISPsGeneral_new}, $\pi_{NoN,z=0}(\tilde{p}^{eq}_{NoN},\tilde{p})=0$. Thus, $\pi^{eq}_{NoN}>\pi_{NoN,z=0}(\tilde{p}^{eq}_{NoN},\tilde{p})$, and this part follows. Now, consider the case that $-t_{NoN}<\Delta p^{eq}<t_N$. Using item 1 of Theorem~\ref{lemma:CP_z=0_new}, if $-t_{NoN}<\Delta p^{eq}<t_N$, then $(q^{eq}_N,q^{eq}_{NoN})=(\tilde{q}_f,\tilde{q}_f)\in F^I_0$. Since  $(q^{eq}_N,q^{eq}_{NoN})\in F^I_0$, we can use \eqref{equ:UNoN_new}. Thus,   by using  $p^{eq}_{NoN}$, $\Delta p^{eq}$,  and , $\pi_{NoN,z=0}(\tilde{p}^{eq}_{NoN},\tilde{p})$ is:
	
	\small
	\be
	\ba
	&\pi_{NoN,z=0}(\tilde{p}^{eq}_{NoN},\tilde{p})\\
	&=\frac{1}{9(t_N+t_{NoN})}\big{(}2t_N+t_{NoN}+(\tilde{q}_p-\tilde{q}_f)(\kappa_u-2\kappa_{ad})\big{)}\\
	&\qquad \qquad \qquad \times\big{(}2t_N+t_{NoN}-(\tilde{q}_p-\tilde{q}_f)(2\kappa_u -\kappa_{ad})\big{)}
	\ea
	\ee
	\normalsize
	
	 Next, we prove that $t_{NoN}+2t_N+  (\tilde{q}_{p}-\tilde{q}_f)(\kappa_u +\kappa_{ad})>2t_N+t_{NoN}+(\tilde{q}_p-\tilde{q}_f)(\kappa_u-2\kappa_{ad})$ and $t_{NoN}+2t_N+  (\tilde{q}_{p}-\tilde{q}_f)(\kappa_u +\kappa_{ad})>2t_N+t_{NoN}-(\tilde{q}_p-\tilde{q}_f)(2\kappa_u -\kappa_{ad})$. This yields  $\pi^{eq}_{NoN}>\pi_{NoN,z=0}(\tilde{p}^{eq}_{NoN},\tilde{p})$. To prove the inequalities, note that:
	 \footnotesize
	 $$
	 \ba
	 &t_{NoN}+2t_N+  (\tilde{q}_{p}-\tilde{q}_f)(\kappa_u +\kappa_{ad})\\
	 &>2t_N+t_{NoN}+(\tilde{q}_p-\tilde{q}_f)(\kappa_u-2\kappa_{ad})\iff 3\kappa_{ad}(\tilde{q}_p-\tilde{q}_f)>0\\
	 &t_{NoN}+2t_N+  (\tilde{q}_{p}-\tilde{q}_f)(\kappa_u +\kappa_{ad})\\
	 &>2t_N+t_{NoN}-(\tilde{q}_p-\tilde{q}_f)(2\kappa_u -\kappa_{ad})\iff 3\kappa_u(\tilde{q}_p-\tilde{q}_f)>0
	 \ea
	 $$
\normalsize

	 Since $\tilde{q}_p>\tilde{q}_f$, both inequalities hold. This completes the proof of this part. \\
	 \textbf{Part 2-iv:} In this part, we prove that there is no profitable unilateral deviation by ISPs when one of $t_N$ or $t_{NoN}$ is large. To do so, first, in Part 2-iv-NoN, we rule out the possibility of a profitable deviation by the non-neutral ISP. Then, in Part 2-iv-N, we rule out profitable deviations by the neutral ISP.
	
	 Note that,  by \eqref{equ:equ:B>_eq_piNoN_newr}, the equilibrium payoff of ISP NoN, $\pi^{eq}_{NoN}=\pi_{NoN}(\tilde{p}^{eq}_{NoN},\tilde{p}_{t,3})$ is:
	 $$
	  \pi^{eq}_{NoN}=\frac{\big{(}t_{NoN}+2t_N+  (\tilde{q}_{p}-\tilde{q}_f)(\kappa_u +\kappa_{ad})\big{)}^2}{9(t_N+t_{NoN})}
	 $$
	 In addition, using $(q^{eq}_N,q^{eq}_{NoN})=(\tilde{q}_f,\tilde{q}_p)$, $p^{eq}_{N}$, $\Delta p^{eq}$, and \eqref{equ:UN_new}, we can find $\pi^{eq}_{N}=\pi_{N}(\tilde{p}^{eq}_{N})$,:
\be \label{equ:equ:B>_eq_piN_newr}
\pi^{eq}_{N}=\frac{\big{(}2t_{NoN}+t_N-(\tilde{q}_{p}-\tilde{q}_f)(\kappa_u +\kappa_{ad})\big{)}^2}{9(t_N+t_{NoN})}
\ee
	Note that when $t_N$ and $t_{NoN}$ are large, $\pi^{eq}_N$ and $\pi^{eq}_{NoN}$ would be large.
	
	 Consider different regions in Theorem~\ref{theorem:p_tilde_new}.  We denote
	 $\Delta p\leq \kappa_u \tilde{q}_{p}-t_{NoN}$ by region A, $\kappa_u \tilde{q}_{p}-t_{NoN}<\Delta p<\Delta p_t=\kappa_u (2\tilde{q}_{p}-\tilde{q}_f)-t_{NoN}$ by region $B_1$, $\Delta p_t=\kappa_u (2\tilde{q}_{p}-\tilde{q}_f)-t_{NoN}\leq \Delta p<t_N+\kappa_u(\tilde{q}_{p}-\tilde{q}_f)$ by region C, $t_N+\kappa_u(\tilde{q}_{p}-\tilde{q}_f)\leq \Delta p<t_N+\kappa_u \tilde{q}_{p}$ by  $B_2$, and $\Delta p\geq t_N+\kappa_u \tilde{q}_{p}$ by D. Recall that $\Delta p^{eq}=p^{eq}_{NoN}-p^{eq}_N$ is in region C. Note that the payoffs are concave in C, and we found the strategies by solving the first order condition. Thus, there is no unilateral profitable deviation in C.\\
	 \textbf{Part 2-iv-NoN:}  Now, we consider unilateral deviations by ISP NoN. We prove that any deviation to regions A, $B_1$, $B_2$, and D is not profitable in Cases 2-iv-NoN-A, 2-iv-NoN-$B_1$, 2-iv-NoN-$B_2$, and 2-iv-NoN-D, respectively.  This yields that no deviation is profitable for ISP NoN.\\
	 \textbf{Case 2-iv-NoN-A:} First, we  prove  that in Region A, $z^{eq}=1$. Note that in this case, by Definition \ref{def:pt}, $\tilde{p}_t=\tilde{p}_{t,1}$. Thus,  $\pi_{NoN}(p_{NoN},\tilde{p}_{t})=p_{NoN}-c+\tilde{q}_{p}\tilde{p}_{t,1}=p_{NoN}-c+\kappa_{ad}(\tilde{q}_p-\tilde{q}_f)$ (by \eqref{equ:payoffISPsGeneral_new}, Definition~\ref{def:pt1,pt2}, and since $n_{NoN}=1$ by item 1 of Theorem~\ref{theorem:p_tilde_new}). On the other hand, $\pi_{NoN,z=0}(p_{NoN},\tilde{p})=(p_{NoN}-c)n_{NoN}$. Thus, $\pi_{NoN}(p_{NoN},\tilde{p}_t)>\pi_{NoN,z=0}(p_{NoN},\tilde{p})$ (since $\tilde{q}_p>\tilde{q}_f$ and $0\leq n_{NoN}\leq 1$). Thus, by Theorem~\ref{theorem:NE_stage2_new_suff}, in region A, $z^{eq}=1$.

	 Now, consider  $p^{eq}_N$ fixed and decreasing $p_{NoN}$ such that $\Delta p$  in region A, i.e. $\Delta p\leq \kappa_u \tilde{q}_{p}-t_{NoN}$.  Since in Region A, $z^{eq}=1$, and by Theorem~\ref{theorem:NE_stage2_new_suff}, the payoff after deviation is  $\pi_{NoN}'=p_{NoN}-c+\tilde{q}_p\tilde{p}_{t,1}$ (by \eqref{equ:payoffISPsGeneral_new}, Definition~\ref{def:pt1,pt2}, and since $n_{NoN}=1$ by item 1 of Theorem~\ref{theorem:p_tilde_new}).
	  Thus, the payoff of the ISP NoN is an increasing function of $p_{NoN}$. Therefore, all other prices are dominated by $p'_{NoN}={p}^{eq}_N+\kappa_u \tilde{q}_{p}-t_{NoN}$. The payoff in this case is $\pi'_{NoN}={p}^{eq}_N+\kappa_u \tilde{q}_{p}-t_{NoN}-c+\tilde{q}_{p}\tilde{p}_{t,1}$ (by \eqref{equ:payoffISPsGeneral_new}-I). Therefore:
	 \be \label{equ:deviation:NoN}
	 \pi'_{NoN}=\frac{1}{3}(t_N-t_{NoN})+\alpha
	 \ee
	where $\alpha$ is a constant independent of $t_N$ and $t_{NoN}$. Now, in Cases (i), (ii), and (iii),
	 we prove that $\pi^{eq}_{NoN}>\pi'_{NoN}$ when  (i) $t_N$
	  is sufficiently larger than other parameters, (ii) $t_{NoN}$ is sufficiently larger than other parameters, and (iii) $t_N$ and $t_{NoN}$ are of the same order of magnitude and
	  both are sufficiently larger than other parameters, respectively.\\
	  \textbf{Case (i): }  If $t_N$
	  is sufficiently larger than other parameters, then:
	$$
	\pi^{eq}_{NoN}\approx \frac{4 t_N}{9}>\pi'_{NoN}\approx\frac{1}{3}t_N
	$$
	Thus, this deviation is not profitable. \\
	\textbf{Case (ii): } If $t_{NoN}$ is sufficiently larger than other parameters, then:
		$$
		\pi^{eq}_{NoN}\approx \frac{ t_{NoN}}{9}>\pi'_{NoN}\approx-\frac{1}{3}t_{NoN}
		$$
 Thus, this deviation is also not profitable.\\
 \textbf{Case (iii): } If $t_N$ and $t_{NoN}$ are of the same order of magnitude ($t_N\approx t_{NoN}$) and both are sufficiently larger than other parameters, then:
 $$
 \pi^{eq}_{NoN}=\frac{t_N}{2}>\frac{t_N}{3}>\pi'_{NoN}
 $$
Thus, this deviation is not profitable.

Thus, any deviation to region A by ISP NoN is not profitable. This completes the proof of this case. \\
\textbf{Case 2-iv-NoN-$B_1$:} Now, consider a deviation by ISP NoN to region $B_1$, i.e. $\kappa_u \tilde{q}_{p}-t_{NoN}<\Delta p<\Delta p_t=\kappa_u (2\tilde{q}_{p}-\tilde{q}_f)-t_{NoN}$. Note that with this deviation, $p'_{NoN}=\frac{1}{3}(t_N-t_{NoN})+\alpha$, where $\alpha_l<\alpha<\alpha_u$, in which $\alpha_l$ and $\alpha_u$  are bounded. In addition, by \eqref{equ:EUs_linear}, after the deviation, $n'_{NoN}=\frac{t_N+t_{NoN}-\beta}{t_N+t_{NoN}}$, where $\beta>0$ is bounded ($\beta_l<\beta<\beta_u$, and $\beta_l$ and $\beta_u$  bounded ). Therefore, for large $t_N$ and $t_{NoN}$, $n'_{NoN}\rightarrow 1$. Thus, by \eqref{equ:UNoN_new}, the payoff of ISP NoN after deviation is:
$$
\pi'_{NoN}=\frac{1}{3}(t_N-t_{NoN})+\gamma
$$
where $\gamma$ is bounded (Note that $\tilde{p}$ is independent of $t_N$ and $t_{NoN}$). This expression is similar to  \eqref{equ:deviation:NoN}. Thus, we can exactly repeat the arguments in Cases i, ii, and iii to prove that  any deviation to region $B_1$ by ISP NoN is not profitable. This completes the proof of this case. \\
\textbf{Case 2-iv-NoN-$B_2$:} Now, consider a deviation by ISP NoN to region $B_2$, i.e. $t_N+\kappa_u(\tilde{q}_{p}-\tilde{q}_f)\leq \Delta p<t_N+\kappa_u \tilde{q}_{p}$. Note that with this deviation, $\Delta p'=t_N+\alpha$, and $p'_{NoN}=\frac{2t_{NoN}+4t_N}{3}+\gamma$
 where $\kappa_u(\tilde{q}_p-\tilde{q}_f)\leq \alpha\leq \kappa_u \tilde{q}_p$ and thus $\gamma$ is bounded. Thus,  by \eqref{equ:EUs_linear}, after this deviation, $n'_{NoN}=\frac{\beta}{t_N+t_{NoN}}$, where $\beta>0$ is a constant independent of $t_N$ and $t_{NoN}$, and the payoff of ISP NoN after deviation is $\pi'_{NoN}=\frac{2t_{NoN}+4t_N}{3(t_N+t_{NoN})}\beta +\eta$ (by \eqref{equ:payoffISPsGeneral_new} and considering that by Theorem~\ref{theorem:NE_stage2_new_suff}, if $z^{eq}=1$, then $\tilde{p}=\tilde{p}_{t,2},$, and independent of $t_N$ and $t_{NoN}$), where $\eta$ is a constant independent of $t_N$ and $t_{NoN}$. Thus, when one of $t_N$ and $t_{NoN}$ is large, $\pi'_{NoN}\rightarrow constant$. Therefore, $\pi^{eq}_{NoN}>\pi'_{NoN}$. Thus, any deviation to region $B_2$ by ISP NoN is not profitable.\\
\textbf{Case 2-iv-NoN-D:} By item 4 of Theorem~\ref{theorem:p_tilde_new}, in region D, $n_{NoN}=0$. Thus, a deviation to this region, yields a payoff of zero, by \eqref{equ:UNoN_new} and $z^{eq}=0$. Thus, a deviation of this kind is not profitable for ISP NoN. \\
\textbf{Part 2-iv-N:} Now, consider unilateral deviations by the  neutral ISP. Similar to Part 2-iv-N, we prove that any deviation to regions $A$, $B_1$, $B_2$, and $D$ is not profitable. We do so in Cases 2-iv-N-A, 2-iv-N-$B_1$, 2-iv-N-$B_2$, and 2-iv-N-D, respectively.  \\
\textbf{Case 2-iv-N-A:} Consider a deviation by ISP N to region A. In this case, by item 1 of Theorem~\ref{theorem:p_tilde_new}, $n_N=0$. Thus, the payoff of ISP N after deviation is zero (by~\eqref{equ:UN_new}), and this deviation is  not profitable.\\
\textbf{Case 2-iv-N-$B_1$:}  Now, consider a deviation by ISP N to region $B_1$, i.e. $\kappa_u \tilde{q}_{p}-t_{NoN}<\Delta p<\Delta p_t=\kappa_u (2\tilde{q}_{p}-\tilde{q}_f)-t_{NoN}$. Note that with this deviation, $\Delta p=-t_{NoN}+\alpha$, and $p'_N=\frac{4t_{NoN}+2t_N}{3}+\gamma$, where $\kappa_u\tilde{q}_p<\alpha<\kappa_u(2\tilde{q}_p-\tilde{q}_f)$ and thus $\gamma$ is bounded. Thus, by \eqref{equ:EUs_linear}, $n'_{N}=\frac{\beta}{t_N+t_{NoN}}$, where $\beta>0$ is bounded. By \eqref{equ:payoffISPsGeneral_new}. The payoff of ISP N after deviation is $\pi_N=\frac{4t_{NoN}+2t_N}{3(t_N+t_{NoN})}\beta$ (by \eqref{equ:payoffISPsGeneral_new}). Thus, when one of $t_N$ and $t_{NoN}$ is large, $\pi'_{N}\rightarrow constant$. Thus, $\pi^{eq}_N>\pi'_N$. Therefore, any deviation to region $B_1$ by ISP N is not profitable.

\textbf{case 2-iv-N-$B_2$:} Now, consider a deviation by ISP NoN to region $B_2$, i.e. $t_N+\kappa_u(\tilde{q}_{p}-\tilde{q}_f)\leq \Delta p<t_N+\kappa_u \tilde{q}_{p}$. Note that with this deviation, $\Delta p'=t_N+\alpha$, where $\kappa_u(\tilde{q}_p-\tilde{q}_f)\leq \alpha<\kappa_u\tilde{q}_p$. Thus, $p'_{N}=\frac{1}{3}(t_{NoN}-t_N)+\beta$, where $\beta$ is bounded.    In addition, by \eqref{equ:EUs_linear}, after the deviation, $n'_{N}=\frac{t_N+t_{NoN}-\gamma}{t_N+t_{NoN}}$, where $\gamma>0$ is bounded. Therefore, for large $t_N$ or $t_{NoN}$, $n'_{N}\rightarrow 1$. Thus, by \eqref{equ:UN_new}, the payoff of ISP N after deviation is:
\be \label{equ:deviation:N}
\pi'_{N}=\frac{1}{3}(t_{NoN}-t_{N})+\eta
\ee
where $\eta$ is bounded.  Now, in Cases i, ii, and iii,
we prove that $\pi^{eq}_{N}>\pi'_{N}$ when  (i) $t_N$
is sufficiently larger than other parameters, (ii) $t_{NoN}$ is sufficiently larger than other parameters, and (iii) $t_N$ and $t_{NoN}$ are of the same order of magnitude and
both are sufficiently larger than other parameters, respectively.\\
\textbf{Case i: }  If $t_N$
is sufficiently larger than other parameters, then:
$$
\pi^{eq}_{N}\approx \frac{ t_N}{9}>\pi'_{N}\approx-\frac{1}{3}t_N
$$
Thus, this deviation is not profitable. \\
\textbf{Case ii: } If $t_{NoN}$ is sufficiently larger than other parameters, then:
$$
\pi^{eq}_{N}\approx \frac{4 t_{NoN}}{9}>\pi'_{N}\approx\frac{1}{3}t_{NoN}
$$
Thus, this deviation is also not profitable.\\
\textbf{Case iii: } If $t_N$ and $t_{NoN}$ are of the same order of magnitude ($t_N\approx t_{NoN}$) and both are sufficiently larger than other parameters, then:
$$
\pi^{eq}_{N}=\frac{t_{NoN}}{2}>\frac{t_{NoN}}{3}>\pi'_{N}
$$
Thus, this deviation is not profitable.

Thus, any deviation to Region $B_2$ by ISP N is not profitable. This completes the proof of this case.\\
\textbf{Case 2-iv-N-D:}  Now, consider  decreasing $p_{NoN}$ such that $\Delta p$  in region D, i.e. $\Delta p\geq \kappa_u \tilde{q}_{p}+t_{N}$.  Note that by item 4 of Theorem~\ref{theorem:p_tilde_new}, $z^{eq}=0$, and $n_N=1$. Thus, the payoff of ISP N is equal to $p_{N}-c$ (by \eqref{equ:payoffISPsGeneral_new}). Thus, the payoff of the ISP N is an increasing function of $p_{N}$. Therefore, all other prices are dominated by $p'_{N}={p}^{eq}_{NoN}-(\kappa_u \tilde{q}_{p}+t_{N})$. Thus, the payoff in this case is $\pi'_{N}=\frac{1}{3}(t_{NoN}-t_N)+\alpha$, where $\alpha$ is a constant and is independent of $t_N$ and $t_{NoN}$. This expression is similar to \eqref{equ:deviation:N}. Thus, we can exactly repeat the arguments in Cases 2-iv-N-$B_2$-a, 2-iv-N-$B_2$-b, and 2-iv-N-$B_2$-c to prove that  any deviation to region $D$ by ISP NoN is not profitable. This completes the proof of this case. This completes the proof of this case, and the theorem.
\end{proof}

\section{Proof of Theorem~\ref{lemma:NEz=0}}\label{appendix:lemma:NEz=0}
\begin{proof}
 	
 In Section~\ref{appendix:theorem:neutralz=0_q<}, while proving 
  Theorem~\ref{theorem:neutralz=0_q<}-1, considering only $z^{eq}=0$, the only possible SPNE access fee by which  $(q^{eq}_N,q^{eq}_{NoN})\in F_0$, i.e. $z^{eq}=0$ is 	$p^{eq}_N=c+\frac{1}{3}(2t_{NoN}+t_N)$ and $p^{eq}_{NoN}=c+\frac{1}{3}(2t_N+t_{NoN})$. Following Remark~\ref{r0}, one may now obtain $q^{eq}_N, q^{eq}_{NoN}, n^{eq}_N, n^{eq}_{NoN}.$ From Theorem~\ref{lemma:CP_z=0_new}-1, $(q^{eq}_N,q^{eq}_{NoN})=(\tilde{q}_f,\tilde{q}_f)\in F^I_0$. From $\Delta p=p^{eq}_{NoN}-p^{eq}_N$, and \eqref{equ:EUs_linear}, $n^{eq}_N=\frac{2t_{NoN}+t_N}{3(t_{NoN}+t_N)}, n^{eq}_{NoN}=\frac{2t_N+t_{NoN}}{3(t_N+t_{NoN})}$. Finally,  since $z^{eq}=0$, $\tilde{p}^{eq}$ is of no importance. Thus, candidate (e) is the only possible SPNE when $z^{eq} = 0.$

 We next prove that candidate (e) is indeed a SPNE when the CP is restricted to choose $z=0.$ We prove that no unilateral deviation is profitable. First, in Case a, we rule out the possibility of a unilateral deviation when  $-t_{NoN}<\Delta p<t_N$ for both neutral and non-neutral ISPs. Then, we consider  $\Delta p\leq -t_{NoN}$ and $\Delta p\geq t_N$, and in Cases  NoN and N, we rule out the possibility of a unilateral deviation in these regions for ISP N and NoN, respectively.
 
 First, using \eqref{equ:UN_new} and \eqref{equ:UNoN_new}, and item 1 of Theorem~\ref{lemma:CP_z=0_new}, it follows that:
 	\be\label{equ:ISPcandidateMultihome}
 	\ba
 	p^{eq}_N&= c+\frac{1}{3}(2t_{NoN}+t_N)\\
 	p^{eq}_{NoN}&=c+\frac{1}{3}(2t_N+t_{NoN})
 	\ea
 	\ee
 	
 	\textbf{Case a:} First, note that  by concavity of the payoffs (using \eqref{equ:UN_new} and \eqref{equ:UNoN_new}) as long as $-t_{NoN}<\Delta p<t_N$, i.e. $0<x_N<1$, a unilateral deviation by one of the ISPs from $p^{eq}_N$ or $p^{eq}_{NoN}$  decreases this ISP's payoff. Thus, we should consider the deviation by ISPs by which $\Delta p\leq -t_{NoN}$ or $\Delta p\geq t_N$.
 	
 	\textbf{Case NoN:} Now, consider the deviations by the non-neutral ISP. Fix $p_N=p^{eq}_N$, and consider two cases. In Case 2-NoN-i (respectively, Case 2-NoN-ii), we consider deviation by ISP NoN such that $\Delta p\geq t_N$ (respectively, $\Delta p\leq -t_{NoN}$).
 	
 	\textbf{Case NoN-i:} Suppose the non-neutral ISP increases her price from ${p}^{eq}_{NoN}$ to make $\Delta p\geq t_N$. In this case, $n_{NoN}=0$, and the payoff of the ISP is zero (by \eqref{equ:payoffISPsGeneral_new}). Since in the candidate equilibrium strategy this payoff is non-negative\footnote{Note that $p^{eq}_{NoN}>c$ and $0\leq n_{NoN}\leq 1$.}, this deviation is not profitable.

 	\textbf{Case NoN-ii:}  Now, consider the case in which the non-neutral ISP decreases her price to make $\Delta p\leq -t_{NoN}$. In this case, $n_{NoN}=1$ and $\pi_{NoN}(p'_{NoN},z=0)=p'_{NoN}-c$ (by \eqref{equ:payoffISPsGeneral_new}). Thus, the payoff is a strictly  increasing function of $p'_{NoN}$, and is  maximized at $p'_{NoN}=p^{eq}_{N}-t_{NoN}$. We show that $\pi_{NoN}(p'_{NoN},z=0)<\pi_{NoN}(p^{eq}_{NoN},z=0)$ \footnote{note that $p_N=p^{eq}_N$ is fixed.}. Note that $\pi_{NoN}(p'_{NoN},z=0)=\frac{1}{3}(t_N-t_{NoN})$. In addition, using \eqref{equ:ISPcandidateMultihome}, \eqref{equ:payoffISPsGeneral_new}, $0\leq x_{N}\leq 1$, \eqref{equ:EUs_linear}, and the fact that with $p^{eq}_N$ and $p^{eq}_{NoN}$, $q^{eq}_{NoN}-q^{eq}_{N}=0$:
 	$$
 	\pi_{NoN}(p^{eq}_{NoN},z=0)=\frac{1}{9}\frac{(2t_N+t_{NoN})^2}{t_{NoN}+t_N}
 	$$
 	Thus:
 	 $$
 	\ba
 	&\pi_{NoN}(p'_{NoN},z=0)< \pi_{NoN}(p^{eq}_{NoN},z=0)\\ &\qquad \iff 3(t^2_N-t^2_{NoN})<4t^2_N+t^2_{NoN}+4 t_N t_{NoN}\\
 	& \qquad \iff t^2_N+4 t^2_{NoN}+4t_Nt_{NoN}>0
 	\ea
 	$$
 	where the last inequality is always true. Thus, this deviation is not profitable for ISP NoN.

 	These cases prove that no deviation form \eqref{equ:ISPcandidateMultihome} is profitable for ISP NoN.
 	
 	\textbf{Case N:}  Now, consider a deviation by the neutral ISP from \eqref{equ:ISPcandidateMultihome}.
 	Similar argument can be done for the neutral ISP. Fix, $p_{NoN}=p^{eq}_{NoN}$, and consider two cases.  In Case 2-N-i (respectively, Case 2-N-ii), we consider deviation by ISP N such that $\Delta p\leq -t_{NoN}$ (respectively, $\Delta p\geq t_{N}$).
 	
 	\textbf{Case N-i:} Suppose the neutral ISP increases her price from ${p}^{eq}_{N}$ to get $\Delta p\leq -t_{NoN}$. In this case, $n_{N}=0$, and the payoff of this ISP is zero. Since in the candidate equilibrium strategy the payoff is non-negative\footnote{Note that $p^{eq}_{N}>c$ and $0\leq n_{N}\leq 1$.}, this deviation is not profitable.
 	
 	\textbf{Case N-ii:} Now, consider the case in which  the non-neutral ISP decreases her price such that $\Delta p\geq t_{N}$. In this case, $n_{N}=1$ and $\pi_{N}(p'_{N})=p'_{N}-c$. Thus, the payoff is a strictly  increasing function of $p'_{N}$, and is  maximized at $p'_{N}=p^{eq}_{NoN}-t_{N}$. We show that $\pi_{N}(p'_{N})<\pi_{N}(p^{eq}_{N})$. Note that $\pi_{N}(p'_{N})=\frac{1}{3}(t_{NoN}-t_N)$ (by \eqref{equ:payoffISPsGeneral_new}). In addition, using \eqref{equ:ISPcandidateMultihome}, \eqref{equ:payoffISPsGeneral_new}, $0\leq x_{N}\leq 1$, \eqref{equ:EUs_linear}, and the fact that with $p^{eq}_N$ and $p^{eq}_{NoN}$, $q^{eq}_{NoN}-q^{eq}_{N}=0$:
 	$$
 	\pi_{N}(p^{eq}_{N})=\frac{1}{9}\frac{(2t_{NoN}+t_{N})^2}{t_{NoN}+t_N}
 	$$
 	Thus:
 	$$
 	\ba
 	\pi_{N}(p'_{N})&< \pi_{N}(p^{eq}_{N}) \\
 	&\iff 3(t^2_{NoN}-t^2_{N})<4t^2_{NoN}+t^2_{N}+4 t_N t_{NoN}\\
 	& \iff t^2_{NoN}+4 t^2_{N}+4t_Nt_{NoN}>0
 	\ea
 	$$
 	where the last inequality is always true. Thus, this deviation is not profitable for ISP N.
 	
 	Thus, there is no profitable deviation for ISP N. This completes the proof, and the Theorem follows.
 \end{proof}

\section{Continuous Strategy Set for the CP}\label{section:general}

In this section, we consider $q_N\in [0,\tilde{q}_f]$ and $q_{NoN}\in[0,\tilde{q}_p]$. In this case, the CP pays a side payment of $\tilde{p}q_{NoN}$ if she chooses $q_{NoN}\in (\tilde{q}_f,\tilde{q}_p]$. The rest of the model is the same as before. Note that in this case, the optimum strategies in Stage 4 of the game, in which end-users decide on the ISP, is the same as before. \emph{We prove that the optimum decisions made by the CP is similar to the decisions of the CP when she has a discrete set of strategies (explained in the paper)}. This yields that the results of the model would the same as before when the CP chooses her strategy from a continuous set.

Therefore, we focus on characterizing the optimum strategies of the CP when she chooses her strategy from  continuous sets, i.e.  $q_N\in [0,\tilde{q}_f]$ and $q_{NoN}\in[0,\tilde{q}_p]$. The following lemma is useful in defining the maximization and to characterize the optimum answers.

\begin{lemma}\label{lemma:qfree}
	$\pi_{CP}(q_N,\tilde{q}_{f,NoN},z=0)\geq \pi_{CP}(q_N,\tilde{q}_{f,NoN},z=1)$.
\end{lemma}

\begin{remark}
	Note that although we considered $z$ to be a dummy variable, in this lemma and for the purpose of analysis, we treat it as an independent variable.
\end{remark}

\begin{proof}
	The lemma follows by \eqref{equ:payoffCP_new}, and comparing the expressions in these two cases:
	
	\footnotesize
	$$
	\ba
	\pi_{CP}(q_N,\tilde{q}_{f,NoN},z=0)-\pi_{CP}(q_N,\tilde{q}_{f,NoN},z=1)&= \tilde{q}_{f,NoN} \tilde{p} \geq 0
	\ea
	$$
	
	\normalsize
	Note that we used the fact that  from \eqref{equ:EUs_linear}, since the qualities are the same in  both cases, $n_N$ and $n_{NoN}$ are equal for both cases.
\end{proof}

Lemma~\ref{lemma:qfree} provides the ground to formally define the maximization for the CP as:
\begin{equation}\label{equ:CPopt_initial2}
\begin{aligned}
\max_{z,q_N,q_{NoN}}&\pi_{CP}(q_N,q_{NoN},z)=\\
&\max_{z,q_N,q_{NoN}} \l n_{N}\kappa_{ad}q_N+n_{NoN}\kappa_{ad}q_{NoN}-z\tilde{p}q_{NoN}\r\\
&\text{s.t:}\\
&\qquad q_N\leq \tilde{q}_f\\
&\qquad \mbox{if } z=1 \quad   \tilde{q}_{f} < q_{NoN} \leq \tilde{q}_{p}\\
&\qquad \mbox{if } z=0 \quad   q_{NoN} \leq \tilde{q}_{f}
\end{aligned}
\end{equation}

\textit{Existence of the maximum:} Note that the mixed integer programming \eqref{equ:CPopt_initial2} can be written as two convex maximizations, one for $z=0$ and one for $z=1$. In addition, note that for the case $z=1$, the feasible set is not closed (since  $ \tilde{q}_{f} < q_{NoN} \leq \tilde{q}_{p}$). Thus, in this case, we should use the ``supremum" instead of ``maximum". However, using Lemma~\ref{lemma:qfree}, we prove that the maximum of  \eqref{equ:CPopt_initial2} exists, and therefore the term maximum can be used safely. To prove this, consider the closure of the feasible set when $z=1$ formed by adding $\tilde{q}_{f}$ to the set, i.e. $\tilde{F}_1$. Since the feasible set associated with $z=0$ ($F_0$) and the closure of the feasible set associated to $z=1$ ($\tilde{F}_1$) are closed and bounded (compact) and the objective function is continuous for each $z\in \{0,1\}$, using Weierstrass Extreme Value Theorem, we can say that a maximum exists in each of these two  sets and for the overall optimization~\eqref{equ:CPopt_initial2}. If the maxima in $\tilde{F}_1$ is not $\tilde{q}_{f}$, then the maxima is in the original feasible set ($F_1$). Therefore the maximum of \eqref{equ:CPopt_initial2} exists. If not and $\tilde{q}_{f}$ is the maxima in the set $\tilde{F}_1$, then by Lemma~\ref{lemma:qfree}, the maximum in the set $F_0$ dominates the maximum of the set $\tilde{F}_1$. Thus, the  maxima of \eqref{equ:CPopt_initial2} is in $F_0$. Therefore, the maximum of \eqref{equ:CPopt_initial2} exists, and we can use the term maximum safely.

Henceforth, the solution $(\tilde{q}^*_N,\tilde{q}^*_{NoN},z^*)$ of the maximization~\eqref{equ:CPopt_initial2} would be called the optimum strategies of the CP. This solution yields $x^*_N$ and subsequently $n^*_N$ and $n^*_{NoN}$  by \eqref{equ:EUs_linear}.  In addition, we denote the feasible set of \eqref{equ:CPopt_initial2} by $\mathcal{F}$.

\textit{Finding the optimum strategies of the CP:}  To characterize the optimum strategies, we use the partition the feasible set in Table~\ref{table:subsets}, and characterize the candidate optimum strategies, i.e. the strategies that yield a higher payoff than the rest of the feasible solutions, in each subset. The overall optimum, which is chosen by the CP, is the one that yields the highest payoff among candidate strategies.

Note that although the maximum of the overall optimization exist, a maximum may not necessarily exist in each of the subsets. We will show in the next set of lemmas that the optimization in each subset of the feasible set can be reduced to a convex maximization over linear constraints. Thus, only the extreme points of the feasible sets may constitute the optimum solution. This means that the CP chooses her strategy among the discrete strategies, $q_N\in \{0,\tilde{q}_f\}$ and $q_{NoN}\in \{0,\tilde{q}_f,\tilde{q}_p\}$.

We now characterize optimum strategies of the CP, by considering each of the sub-feasible sets and characterizing the optimum solutions in each of them. In Lemma~\ref{lemma:optimumCPqs}, we prove that if $(q^*_N,q^*_{NoN})\in F^I$, then $q^*_N\in \{0,\tilde{q}_f\}$, $q^*_{NoN}\in\{0,\tilde{q}_{f},\tilde{q}_{p}\}$, and $(q^*_N,q^*_{NoN})\neq (0,0)$. In Lemma~\ref{lemma:xn<0}, we prove that if $(q^*_N,q^*_{NoN})\in F^L$, then ${q}^*_{NoN}=\tilde{q}_{f}$, if $q^*_{NoN}\in F^L_0$, and  ${q}^*_{NoN}=\tilde{q}_{p}$, if $q^*_{NoN}\in F^L_1$.  Moreover, $0\leq q^*_N\leq \frac{1}{\kappa_u}(\kappa_u {q}^*_{NoN}-t_{NoN}-\Delta p)$, and $\Delta p \leq \kappa_u {q}^*_{NoN}-t_{NoN}$. In Lemma~\ref{lemma:xn>1}, we prove that if $(q^*_N,q^*_{NoN})\in F^U$, then ${q}^*_{N}=\tilde{q}_f$ and $0\leq {q}^*_{NoN}\leq \frac{1}{\kappa_u}(\kappa_u \tilde{q}_f-t_{N}+\Delta p)$ and $\Delta p \geq t_N -\kappa_u \tilde{q}_f$. In addition, Lemmas \ref{lemma:tildep_condition} and  \ref{lemma:pos} provide some results that are useful in proving Lemmas~\ref{lemma:optimumCPqs}-\ref{lemma:xn>1}.

\begin{lemma}\label{lemma:tildep_condition}
	In an optimum solution of \eqref{equ:CPopt_initial2}, $n_{NoN}\kappa_{ad}-z\tilde{p}\geq 0$.
\end{lemma}
\begin{proof}
	Suppose there exists an optimum answer such that $n_{NoN}\kappa_{ad}-z\tilde{p}< 0$. Note that $0\leq n_N,n_{NoN}\leq 1$ and qualities are non-negative. Thus, in this case, $\pi_{CP} <\kappa_{ad} q_N$. However,  choosing $z=0$ and $q_{NoN}=q_N$, yields a profit equal to $\kappa_{ad} q_N$. This contradicts the solution with $n_{NoN}\kappa_{ad}-z\tilde{p}< 0$ to be optimum. Thus,  the Lemma follows.
\end{proof}


\begin{lemma}\label{lemma:pos}
	In an optimum solution, the CP offers the content quality equal to one of the threshold at least on one ISP, i.e. $q^*_{N}=\tilde{q}_f$ OR $(q^*_{NoN}=\tilde{q}_{p} \text{ XOR } q^*_{NoN}=\tilde{q}_{f} )$, where XOR means only one the qualities is chosen.
\end{lemma}

\begin{proof}
	Suppose not. Let the optimum qualities to be $\hat{q}_{NoN}<\tilde{q}_{f}$ if $z=0$, or $\tilde{q}_{f} <\hat{q}_{NoN}<\tilde{q}_{p}$ if $z=1$, and  $\hat{q}_N<\tilde{q}_f$. The difference between the qualities offered in two platforms is $\Delta q=\hat{q}_{NoN}-\hat{q}_N$. Consider $q'_{NoN}=\hat{q}_{NoN}+\epsilon$ and $q'_{N}=\hat{q}_{N}+\epsilon$ in which $\epsilon>0$ and is such that  ${q}'_{NoN}\leq \tilde{q}_{f}$ if $z=0$, or $\tilde{q}_{f} \leq {q}'_{NoN}\leq \tilde{q}_{p}$ if $z=1$, and  ${q}'_N\leq \tilde{q}_f$. Note that $z$ remains fixed and ${q}'_{NoN}-{q}'_N=\hat{q}_{NoN}-\hat{q}_N=\Delta q$. Since $\Delta q$ is the same for two cases, the number of subscriber to each ISP is the same for both cases by \eqref{equ:EUs_linear}. Lemma~\ref{lemma:tildep_condition}, \eqref{equ:payoffCP_new}, and the fact that $n_N,n_{NoN}\geq 0$ yield that $\pi'_{CP}\geq \hat{\pi}_{CP}$, where $\hat{\pi}_{CP}$ (, respectively  $\pi'_{CP}$) is the payoff of the CP when the vector of qualities is $(\hat{q}_N,\hat{q}_{NoN})$ (, respectively, $({q}'_N,{q}'_{NoN})$).
	
	We now prove if $(\hat{q}_N,\hat{q}_{NoN})$ is the optimum solution, then the inequality is strict, i.e. $\pi'_{CP}> \hat{\pi}_{CP}$. Suppose not, and  $\pi'_{CP}= \hat{\pi}_{CP}$. This only happens if $n_{NoN}\kappa_{ad}-z\tilde{p}=0$ and $n_N=0$.  Note that in this case, $\pi'_{CP}= \hat{\pi}_{CP}=0$. However, in the previous paragraph, we argued that with $q_N=\tilde{q}_f$ and $q_{NoN}=\tilde{q}_f$, the CP can get a payoff of $\kappa_{ad}\tilde{q}_f>0$. This contradicts the assumption that $(\hat{q}_N,\hat{q}_{NoN})$ is the optimum solution. Thus, $\pi'_{CP}> \hat{\pi}_{CP}$.
	
	This inequality contradicts the assumption that $(\hat{q}_N,\hat{q}_{NoN})$ is the optimum solution. Thus, the result follows.
\end{proof}

Clearly, the decision of the CP about the vector of qualities depends on the parameter $x_N$ \eqref{equ:xn}, and subsequently on $n_N$. First, we characterize the candidate strategies of the CP when $0\leq x_N\leq 1$, i.e. $(q^*_N,q^*_{NoN})\in F^I$ and therefore $n_N=x_N$. Then, we consider the case of $x_N< 0$ ($n_N=0$ and $(q^*_N,q^*_{NoN})\in F^L$) and $x_N> 1$ ($n_N=1$ and $(q^*_N,q^*_{NoN})\in F^U$). Finally, we combine both cases to determine the optimum strategies of the CP.  In the following lemma, we characterize the candidate optimum qualities in $F^I$, i.e. the strategies by which $0\leq x_N\leq 1$.
\begin{lemma}\label{lemma:optimumCPqs}
	If $(q^*_{N},q^*_{NoN})\in F^I$, i.e. optimum strategies are such that $0< x^*_N< 1$, then $q^*_N\in\{0,\tilde{q}_f\}$, $q^*_{NoN}\in\{0,\tilde{q}_{f},\tilde{q}_{p}\}$, $(q^*_N,q^*_{NoN})\neq (0,0)$.
\end{lemma}
\begin{remark}
	Note that to be in $F^I$ and from \eqref{equ:EUs_linear}, $(q^*_N,q^*_{NoN})$ should be such that $\frac{\Delta p-t_N}{\kappa_u}< \Delta q^*=q^*_{NoN}-q^*_N< \frac{\Delta p + t_{NoN}}{\kappa_u}$. In Lemma~\ref{lemma:pos}, we have proved that the quality on at least one of the ISPs is equal to a threshold. In this lemma, we prove that the qualities offered on both ISPs are equal to thresholds or one of them is zero.
\end{remark}
\begin{proof}
	We would like to characterize the optimum qualities in $F^I=F^I_0\bigcup F^I_1$, i.e. optimum strategies for which $0< x_N< 1$.
	First note that by Lemma~\ref{lemma:pos}, either (a) $q^*_{N}=c$ and $q^*_{NoN}=c+\Delta q$ where $c=\tilde{q}_f$, or (b) $q^*_{NoN}=c$ and $q^*_{N}=c-\Delta q$ where $c\in\{\tilde{q}_{f},\tilde{q}_{p}\}$. Note that the feasible sets for each case can be rewritten as a function of $\Delta q$. We characterize the candidate solutions for each case:
	
	\begin{itemize}
		\item Case (a):  The feasible set for the case (a) is $\Delta q\in G_0=[-c, \tilde{q}_{f}-c]$ (for $z=0$) and $\Delta q\in G_1=(\tilde{q}_{f}-c, \tilde{q}_{p}-c]$  (for $z=1$), where $c=\tilde{q}_f$. Let $G=G_0\cup G_1$. Note that if $0\leq x_N\leq 1$, then $n_N=x_N$ and $n_{NoN}=1-x_N$. Thus, \eqref{equ:CPopt_initial2} can be written as,
		
		\be \label{equ:obj1}
		\ba
	&	\max_{z,\Delta q \in G=G_0\cup G_1}  \pi_{CP}(c,c+\Delta q,z)=\\
	&\max_{z,\Delta q \in G}  \big{(} t_{NoN}-\kappa_{u}\Delta q+p_{NoN}-p_{N}\big{)} \kappa_{ad}c+\\
		& +\big{(} t_{N}+\kappa_{u}\Delta q+p_{N}-p_{NoN}\big{)} \kappa_{ad}(c+\Delta q)-z\tilde{p}(c+\Delta q)
		\ea
		\ee
\normalsize		
		
		Note that although the feasible set  $G_1$ is not closed, we used maximum instead of supremum. We will show that the maximum of \eqref{equ:obj1} exists. Thus, the term maximum can be used safely. Note that the objective functions of \eqref{equ:obj1} is a strictly convex functions of $\Delta q$. Note that henceforth wherever we refer to maximum without further clarification, we refer to the solution of \eqref{equ:CPopt_initial2}.
		
		Let $\tilde{G}_1$  be the closure of $G_1$, then $\tilde{G}_1 \backslash G_1=\{\tilde{q}_{f}-c\}$. First, we prove that the maximum of \eqref{equ:obj1} exists.  Note that $G_0$ and $\tilde{G}_1$ are closed and bounded (compact) and the objective function of \eqref{equ:obj1} is continuous with respect to $\Delta q$ for each $z\in \{0,1\}$. Using Weierstrass Extreme Value Theorem, we can say that a maxima for $\pi_{CP}(c,\Delta q+c,z=0)$ and $\pi_{CP}(c,\Delta q+c,z=1)$ exists in each of two sets  $G_0$ and $\tilde{G}_1$, respectively. Thus, the overall maximum for the objective function of~\eqref{equ:obj1} over $G_0$ and $\tilde{G}_1$ exists. Now, consider two cases:
		
		\begin{enumerate}
			\item If the maxima of  $\pi_{CP}(c,\Delta q+c,z=1)$ in $\tilde{G}_1$ is not $\Delta q=\tilde{q}_{f}-c$, then the maxima is in the original feasible set ($G_1$). Therefore the maximum of \eqref{equ:obj1} exists (since $G_0$ is closed).
			\item  If $\Delta q=\tilde{q}_{f}-c$ is the maxima  of  $\pi_{CP}(c,\Delta q+c,z=1)$ in the set $\tilde{G}_1$, then by Lemma~\ref{lemma:qfree}, the maximum of  $\pi_{CP}(c,\Delta q+c,z=0)$ in the set $G_0$ greater than or equal to the maximum of $\pi_{CP}(c,\Delta q+c,z=1)$  in $\tilde{G}_1$. Thus, the maxima of  \eqref{equ:obj1} over $G_0$ and $G_1$ exists and is in $G_0$.
		\end{enumerate}
		
		Now, that we have proved the existence of the maximum for \eqref{equ:obj1}, we aim to find all the candidate optimum solutions. Note that  the set $G_0$ is closed. Thus, by the strict convexity of the objective function of \eqref{equ:obj1}, the candidate optimums in $G_0$ are the extreme points of $G_0$. Using the definition of this feasible set, the candidate answers are (i) $q^*_N=\tilde{q}_f$ and $q^*_{NoN}\in\{0,\tilde{q}_{f}\}$.
		
		Now, consider the feasible set $\tilde{G}_1$, and consider two cases:
		\begin{enumerate}
			\item If $\Delta q=\tilde{q}_{f}-c$ is not the unique maxima of \eqref{equ:obj1} in $\tilde{G}_1$, then the maxima is in $G_1$ or $G_0$. The candidate answers in the set $G_0$ are already characterized. In addition, by strict convexity of the objective function, the maxima can only be an extreme point of $\tilde{G}_1$. Since $\tilde{q}_{f}-c$ is not the unique maxima of \eqref{equ:CPopt_initial2} in $\tilde{G}_1$, $\tilde{q}_{p}$ is a maxima of \eqref{equ:CPopt_initial2} in $G_1$. Thus, by strong convexity, for all $\Delta q\in G_1$ $\pi_{CP}(c,\tilde{q}_{p},z=1)>\pi_{CP} (c,\Delta q+c,z=1)$, and the only candidate optimum solution over $G_1$ is at $\Delta q=\tilde{q}_{p}-c\in G_1$ which yields   (ii) $q^*_N=\tilde{q}_f$ and $q^*_{NoN}=\tilde{q}_{p}$.
			
			\item If $\tilde{q}_{f}-c$ is the unique maxima in  $\tilde{G}_1$, then  $\pi_{CP}(c,\tilde{q}_{f},z=1)>\pi_{CP} (c,\Delta q+c,z=1)$ for $\Delta q\in G_1$.  By Lemma~\ref{lemma:qfree}, $\pi_{CP}(c,\tilde{q}_{f},z=0)\geq \pi_{CP}(c,\tilde{q}_{f},z=1)$. Therefore, the overall maximum of \eqref{equ:obj1} is in the set $G_0$, and is as characterized previously.
		\end{enumerate}
		
		\item Case (b): The feasible set for the case (b) is $\Delta q\in \hat{G}_0= [c-\tilde{q}_f, c]$ where $c=\tilde{q}_{f}$ (for $z=0)$, and  $\Delta q\in \hat{G}_1=[c-\tilde{q}_f, c]$ where $c=\tilde{q}_{p}$ (for $z=1$). For this case,  \eqref{equ:CPopt_initial2} can be written as:
		\be \label{equ:obj2}
		\ba
		&\max_{z,\Delta q \in \hat{G}=\hat{G}_0\cup \hat{G}_1} \pi_{CP}(c-\Delta q,c,z)=\\
		&\max_{z,\Delta q \in \hat{G}} \kappa_{ad} \big{(} t_{NoN}-\kappa_{u}\Delta q+p_{NoN}-p_{N}\big{)} (c-\Delta q)+\\
		& \qquad \qquad \quad  +\kappa_{ad}c \big{(} t_{N}+\kappa_{u}\Delta q+p_{N}-p_{NoN}\big{)}-z\tilde{p}c
		\ea
		\ee
		Note that the feasible set is closed. Thus the term maximum is fine. In addition, the objective functions of \eqref{equ:obj2} are strictly convex functions of $\Delta q$. Thus, using the strict convexity and the definition of the feasible set, i.e. $c-\tilde{q}_f\leq \Delta q^*\leq c$ where $c$ is $\tilde{q}_{f}$ and $\tilde{q}_{p}$, respectively, we can get the other set of candidate answers, (iii) $q^*_{NoN}=\tilde{q}_{f}$ and $\tilde{q}^*_N\in \{0,\tilde{q}_f\}$, and (iv) $\tilde{q}^*_{NoN}=\tilde{q}_{p}$ and $\tilde{q}^*_N\in\{0,\tilde{q}_f\}$.
	\end{itemize}

	From, (i), (ii), (iii), and (iv), the result follows.
	

	

\end{proof}

The following corollary follows immediately from Lemma  \ref{lemma:optimumCPqs}:

\begin{corollary}\label{lemma:candidatesCP}
	The possible candidate optimum strategies by which $0< x^*_N< 1$, i.e.  $(q^*_N,q^*_{NoN})\in F^I$, are $(1)$  $(0,\tilde{q}_{f})$, $(2)$ $(\tilde{q}_f,0)$, and $(3)$ $(\tilde{q}_f,\tilde{q}_{f})$ when $z=0$, i.e. $(q^*_N,q^*_{NoN})\in F^I_0$, and $(1)$ $(0,\tilde{q}_{p})$ and $(2)$ $(\tilde{q}_f,\tilde{q}_{p})$ when $z=1$, i.e. $(q^*_N,q^*_{NoN})\in F^I_1$. Note that the necessary and sufficient condition for each of these candidate outcomes to be in $F^I$ is $\frac{\Delta p-t_N}{\kappa_u}< \Delta q^*< \frac{\Delta p + t_{NoN}}{\kappa_u}$.
\end{corollary}

Note that Corollary~\ref{lemma:candidatesCP} lists all the candidate answers by which $0< x_N< 1$. In the next three lemmas,  we focus on the candidate answers when $x_N\geq1$ or $x_N\leq 0$.

\begin{lemma}\label{lemma:necessary}
	If $\Delta p> \kappa_u \tilde{q}_{f}-t_{NoN}$ then $x_N> 0$ for all choices of $q_{NoN}$ and $q_N$ in the feasible set $F_0$ \footnote{That is $F^L_0$ is an empty set.}. Similarly,  If $\Delta p> \kappa_u \tilde{q}_{p}-t_{NoN}$ then $x_N> 0$ for all choices of $q_{NoN}$ and $q_N$ in the feasible set $F_1$ \footnote{That is $F^L_1$ is an empty set.}.  In addition,  if $\Delta p<  t_N -\kappa_u \tilde{q}_f$ then $x_N< 1$ for all choices of $q_{NoN}$ and $q_N$ in the overall feasible set $\mathcal{F}$\footnote{That is $F^U$ is empty.}.
\end{lemma}

\begin{proof}
	First note that  from \eqref{equ:EUs_linear}, $x_N> 0$ is equivalent to:
	\begin{equation}\label{equ:equivalent}
	\Delta p > \kappa_u (q_{NoN}-q_N)-t_{NoN}
	\end{equation}
	
	Consider $\Delta p> \kappa_u \tilde{q}_{f}-t_{NoN}$ (respectively, $\Delta p>\kappa_u \tilde{q}_{p}-t_{NoN}$), if $(q_N,q_{NoN})\in F_0$ (respectively, $(q_N,q_{NoN})\in F_1$) then  $\Delta p > \kappa_u \tilde{q}_{f}-t_{NoN}\geq  \kappa_u(q_{NoN}-q_N)-t_{NoN}$ (respectively, $\Delta p > \kappa_u \tilde{q}_{p}-t_{NoN}\geq \kappa_u(q_{NoN}-q_N)-t_{NoN}$) for every choice of $(q_N,q_{NoN})\in F_0$ (respectively, $(q_N,q_{NoN})\in F_1$). The inequality $\Delta p>\kappa_u(q_{NoN}-q_N)-t_{NoN}$ yields $x_N> 0$. The first result of the lemma follows.
	
	Now, we prove the second statement. From \eqref{equ:EUs_linear}, $x_N< 1$ is equivalent to:
	\be\label{equ:equivalent2}
	\Delta p < t_N+\kappa_u (q_{NoN}-q_N)
	\ee
	
	Consider $\Delta p<  t_N-\kappa_u \tilde{q}_f$. Note that:
	
	$$
	\Delta p< t_N-\kappa_u \tilde{q}_f\leq t_N+ \kappa_u(q_{NoN}-q_N)
	$$
	for every choice of $0 \leq q_N\leq \tilde{q}_f$ and $0 \leq q_{NoN}\leq \tilde{q}_{p}$ which are all the possible choices in $\mathcal{F}$. The inequality $\Delta p< t_N+ \kappa_u(q_{NoN}-q_N)$ yields that $x_N< 1$. The result follows.
\end{proof}

The following lemma characterizes all the candidate answers when $x^*_N\leq 0$, and characterize the necessary condition on parameters for this solutions to be feasible.

\begin{lemma}\label{lemma:xn<0}
	Let $(q^*_N,q^*_{NoN})\in F^L$. If $(q^*_N,q^*_{NoN})\in F^L_0$ (respectively, if $(q^*_N,q^*_{NoN})\in F^L_1$), then ${q}^*_{NoN}=\tilde{q}_{f}$ (respectively, ${q}^*_{NoN}=\tilde{q}_{p}$). Moreover, for every $x\in [0, \frac{1}{\kappa_u}(\kappa_u \tilde{q}_{f}-t_{NoN}-\Delta p)]$ (respectively, $x\in [0, \frac{1}{\kappa_u}(\kappa_u \tilde{q}_{p}-t_{NoN}-\Delta p)]$) and $\Delta p \leq \kappa_u \tilde{q}_{f}-t_{NoN}$ (respectively, $\Delta p \leq \kappa_u \tilde{q}_{p}-t_{NoN}$),
	$(x,\tilde{q}_{f})$ (respectively, $(x,\tilde{q}_{p})$) constitutes an optimum solution in $F^L_0$ (respectively, in $F^L_1$).
\end{lemma}

\begin{proof}
	From \eqref{equ:EUs_linear}, $x_N\leq  0$ is equivalent to:
	\begin{equation}\label{equ:equivalentr}
	\Delta p \leq \kappa_u (q_{NoN}-q_N)-t_{NoN}
	\end{equation}
	
	Note that from \eqref{equ:EUs_linear}, if $x_N\leq 0$ then $n_N=0$ and $n_{NoN}=1$. In this case, the payoff of the CP is,
	\be \label{equ:increasingxn<0}
	\pi_G=\kappa_{ad} q_{NoN}-z\tilde{p}q_{NoN}
	\ee
	Note that the value of the payoff is independent of $q_N$ as long as $n_N=0$, and from \eqref{equ:EUs_linear} $n_N$ is a function of $q_N$ and $q_{NoN}$. In addition, note that if there exist a $q_{NoN}$ that satisfies the constraint  $\Delta p \leq \kappa_u (q_{NoN}-q_N)-t_{NoN}$ (and therefore $n_N=0$) then $q'_{NoN}\geq q_{NoN}$ also satisfies this constraint. Therefore for $q'_{NoN}\geq q_{NoN}$, $n_{N}=0$ and \eqref{equ:increasingxn<0} is true. Note  that from Lemma~\ref{lemma:tildep_condition}, \eqref{equ:increasingxn<0} is an increasing function of $q_{NoN}$. Thus,  if $x_N\leq 0$, then $q^*_{NoN}=\tilde{q}_{f}$ if  $(q^*_N,q^*_{NoN})\in F^L_0$ or $q^*_{NoN}=\tilde{q}_{p}$ if $(q^*_N,q^*_{NoN})\in F^L_1$ (using the feasible sets in Table~\ref{table:subsets} and their definitions).

	Using \eqref{equ:equivalentr}, $(q^*_N,q^*_{NoN})\in F^L_0$ (respectively, $(q^*_N,q^*_{NoN})\in F^L_1$) if and only if,
	
	\footnotesize
	\be \label{equ:local1}
	{q}^*_N\leq \frac{1}{\kappa_u}(\kappa_u \tilde{q}_{f}-\Delta p-t_{NoN}) \quad \bigg{(}\text{respectively, } {q}^*_N\leq \frac{1}{\kappa_u}(\kappa_u \tilde{q}_{p}-\Delta p-t_{NoN})  \bigg{)}
	\ee
	\normalsize
	
	Note that every $q^*_N$ that satisfies \eqref{equ:local1} is an optimum answer since when $(q^*_N,q^*_{NoN})\in F^L$, $n_N=0$ and $q^*_N$ is of no importance. Also, note that $q_N\geq 0$. Thus,  \eqref{equ:local1} is true for at least one $q^*_N$ if $\Delta p\leq \kappa_u \tilde{q}_{f}-t_{NoN}$ (respectively, $\Delta p\leq \kappa_u \tilde{q}_{p}-t_{NoN}$). The result follows.
\end{proof}

The following lemma characterizes all the candidate answers when $x_N\geq 1$, and characterize the necessary condition on parameters for this solutions to be feasible.

\begin{lemma}\label{lemma:xn>1}
	If $(q^*_N,q^*_{NoN})\in F^U$, i.e. optimum strategies  are such that $x^*_N\geq 1$. Then ${q}^*_{N}=\tilde{q}_f$. Moreover, for all $x\in [0, \frac{1}{\kappa_u}(\kappa_u \tilde{q}_f-t_{N}+\Delta p)]$ and  $\Delta p \geq t_N -\kappa_u \tilde{q}_f$, $(q^*_N,x)$ constitutes an optimum solution in $F^U$.
\end{lemma}

\begin{proof}
	From \eqref{equ:EUs_linear}, $x_N\geq 1$ is equivalent to:
	\be\label{equ:equivalent2r}
	\Delta p\geq t_N+\kappa_u (q_{NoN}-q_N)
	\ee

	Now, we prove the first result of the lemma. Note that from \eqref{equ:EUs_linear}, if $x_N\geq 1$ then $n_N=1$ and $n_{NoN}=0$. In this case, the payoff of the CP is,
	\be \label{equ:increasingxn>1}
	\pi_G=\kappa_{ad} q_{N}
	\ee
	Note that the value of the payoff is independent of $q_{NoN}$ as long as $n_N=1$, and from \eqref{equ:EUs_linear}, $n_N$ is a function of $q_N$ and $q_{NoN}$. In addition, note that if there exist a $q_{N}$ that satisfies   $\Delta p \geq t_N+\kappa_u (q_{NoN}-q_N)$, then $q'_N\geq q_{N}$ also satisfies this constraint. Therefore, for $q'_N\geq q_N$, $n_N=1$ and \eqref{equ:increasingxn>1} is true. Note that \eqref{equ:increasingxn>1} is an increasing function of $q_{N}$. Thus, ${q}^*_N=\tilde{q}_f$ (using the feasible sets in Table~\ref{table:subsets} and their definitions).
	
	Using \eqref{equ:equivalent2r}, $(q^*_N,q^*_{NoN})\in F^U$ if and only if:
	\be\label{equ:local2}
	{q}^*_{NoN}\leq \frac{1}{\kappa_u}(\kappa_u \tilde{q}_{f}-t_{N}+\Delta p)
	\ee
	Note that every $q^*_{NoN}$ that satisfies \eqref{equ:local2} is an optimum answer since when $(q^*_N,q^*_{NoN})\in F^U$, $n^*_{NoN}=0$ and $q^*_{NoN}$ is of no importance. Also, note that $q_{NoN}\geq 0$. Thus, the condition \eqref{equ:local2} is true for at least one $q^*_{NoN}$ if $\kappa_u \tilde{q}_{f}-t_{N}+\Delta p \geq 0$. The result follows.
\end{proof}

\begin{corollary}\label{corollary:equi_coun}
	 If  $(q^{eq}_N,q^{eq}_{NoN})\in F_0^L$, then $(q^{eq}_N,q^{eq}_{NoN})=(0,\tilde{q}_f)$.   If  $(q^{eq}_N,q^{eq}_{NoN})\in F_1^L$, then $(q^{eq}_N,q^{eq}_{NoN})=(0,\tilde{q}_p)$.	If $(q^{eq}_N,q^{eq}_{NoN})\in F^U$, then $(q^{eq}_N,q^{eq}_{NoN})=(\tilde{q}_f,0)$.
\end{corollary}

\begin{proof}
	Note that when $(q^*_N,q^*_{NoN})\in F^L$ (, respectively $(q^*_N,q^*_{NoN})\in F^U$), then the payoff of the CP is independent of $q^*_N$ and $q^*_{NoN}$. Thus, result of the corollary follows from Tie-Breaking Assumption \ref{assumption:tie_n=0}.
\end{proof}

\begin{theorem} All possible equilibrium strategies are:
\be \label{equ:candidates_countinuos}
\ba
&(0,\tilde{q}_f)\in F^I_0\cup F^L_0\ ,\ (\tilde{q}_f,0) \in F^I_0\cup F^U_0\ ,\ (\tilde{q}_f,\tilde{q}_f) \in F^I_0\ ,\\
& (0,\tilde{q}_{p}) \in F^I_1\cup F^L_1\ ,\  (\tilde{q}_f,\tilde{q}_{p}) \in F^I_1
\ea
\ee
\end{theorem}

\begin{proof}
	Results follow directly from Corollaries	\ref{lemma:candidatesCP} and \ref{corollary:equi_coun}.
\end{proof}

Note that  \eqref{equ:candidates_countinuos} and \eqref{equ:summarize_CP_candidate_new} are exactly similar. This implies that the strategies chosen by the CP when she chooses from continuous sets is exactly similar to the strategies when she chooses from the discrete set characterized in our model. This completes our proof.

\section{Proof when ISP NoN incurs an additional cost for delivering CP's content to EUs at premium quality  }\label{appendix:fixedcost}
 We now consider that ISP  NoN incurs an additional marginal cost $\kappa$   each time she offers CP's content with premium quality, since NoN could presumably have utilized the associated bandwidth elsewhere.    
 Then   $\pi_{NoN}(p_{NoN}, \tilde{p})$ in \eqref{equ:payoffISPsGeneral_new} must be decremented by $z\kappa\tilde{q}_p$. 
 We outline the proof that (e) is the unique SPNE when $\kappa$ is large.

  
  From the outline in Section~\ref{future}, it now follows that candidates (a)-(e) of Section~\ref{section:summaryof resutls}  constitute the only possible SPNEs. We now argue that for large $\kappa$, (e) is the unique SPNE. From Theorem~\ref{lemma:NEz=0}, 
 	 (e)  is the unique SPNE when the CP is constrained to choose $z^{eq}=0$. Thus, we only need to show that (1) there is no profitable deviation from (e) of either SP,  into the region $z=1$ (2) there is an unilateral  profitable deviation of at least one of the SPs from (a)-(d).
 
 First, notice that for the CP to choose $z = 1$, the marginal side-payment, $\tilde{p}$,  offered by NoN, is at most $\kappa_{ad}.$  This is because when $z=1$, she earns at most $\kappa_{ad}\tilde{q}_p$, from the advertising revenue,   while when $z=0$ she gets a payoff of at least $0.$ Thus, unless her side-payment when $z=1$,  $\tilde{p}\tilde{q}_p$, is at most $\kappa_{ad}\tilde{q}_p$, CP will choose $z=0.$ 
 
 We start with (1). First, note that in (e), N's price choice maximizes the only component present in her payoff, given NoN's choices, her subscription revenue. So unilateral deviation from her choice in (e) is not profitable for N. We now  consider that NoN  chooses $p_{NoN}, \tilde{p}$ such that CP opts for $z=1.$ Thus,  $\tilde{p} \leq \kappa_{ad}.$ NoN's payoff has two components, (1) subscription revenue (2) side-payment from CP minus cost due to delivery of premium-quality , $z(\tilde{p}-\kappa)\tilde{q}_p$.
Suppose NoN selects her price only to maximize (1) given N's price choice. Clearly, the maximum value  gives an upper-bound of (1) in the expression for NoN's payoff.  That maximum will also be finite and depend only on $t_N, t_{NoN}, \tilde{q}_f, \tilde{q}_p.$  Let $\Gamma$ be the difference between this maximum and NoN's payoff under (e). Note that $\Gamma$  depends on system parameters, other than $\kappa$. NoN's payoff will decrease as compared to (e) if $(\kappa-\tilde{p})\tilde{q}_p > \Gamma$, which holds, if  $\kappa >  \Gamma/\tilde{q}_p + \kappa_{ad}$, since $\tilde{p} \leq \kappa_{ad}.$ Thus, the deviation is not profitable for NoN. Thus, there is no profitable deviation from (e) of either SP,  into the region $z=1,$
Thus, (e) is a SPNE.

We now consider (2). We will show that for each of (a)-(d), NoN has an unilateral profitable deviation. When $\kappa$ is large enough, NoN's payoff is negative for each of (a)-(d), because of the additional cost of $\kappa \tilde{q}_p$. Now, as part of the unilateral deviation,  by choosing $\tilde{p} > \kappa_{ad}$, NoN can ensure that CP opts for $z=0.$ Thus, the CP offers lower than the premium quality on NoN. As a result, in the worst case, all EUs may choose N, and NoN may get $0$ subscription revenue. But then her payoff becomes $0$, which is higher than the negative value she was getting earlier. Thus, the deviation is profitable for NoN. Thus,  (a)-(d) are not SPNEs. 

\section{Proofs when there is an upper bound on the marginal side-payment}
\label{appendix:upperbound}
We assume that there exists an upper-bound $\alpha > 0$ on the marginal side-payment that NoN can not exceed. We define  a modified threshold $\tilde{p}_t^{mod} = \min(\tilde{p}_t, \alpha)$.  

We start with by proving a modified version of Theorem~\ref{theorem:NE_stage2_new_suff}. The modified version is:

\begin{theorem}\label{theorem:NE_stage2_new_suff_ub}
	\begin{enumerate}
\item If $z^{eq}=1$, then $\tilde{p}^{eq}=\tilde{p}_{t}^{mod}$.
\item If $z^{eq}=1$, then 
	1) if $\tilde{p}_t < \alpha$, $\pi_{NoN}(p_{NoN},\tilde{p}_{t})>\pi_{NoN,z=0}(p_{NoN},\tilde{p})$ and 2) $\Delta p<t_N+\kappa_u \tilde{q}_{p}$. \\
If $\pi_{NoN}(p_{NoN},\tilde{p}_{t}^{mod}) > \pi_{NoN,z=0}(p_{NoN},\tilde{p})$ and $\Delta p<t_N+\kappa_u \tilde{q}_{p}$, $z^{eq}=1.$
\end{enumerate}
\end{theorem}

\begin{proof}
%

Consider a SPNE in which $z^{eq} = 1$. If  $\tilde{p} >  \tilde{p}_t^{mod}$, $\tilde{p}_t^{mod} = \tilde{p}_t$, as $\tilde{p} \leq \alpha.$ Then, $\tilde{p} > \tilde{p}_t$, and $z^{eq}=0$, by Theorem~\ref{theorem:p_tilde_new}. So $\tilde{p} \leq \tilde{p}_t^{mod}$. The payoff of ISP NoN is  equal to $(p_{NoN}-c)n_{NoN}+\tilde{p}\tilde{q}_f$, by \eqref{equ:payoffISPsGeneral_new},  and is a strictly increasing function of $\tilde{p}$ (note that $p_{NoN}$ is fixed, and by \eqref{equ:EUs_linear}, $n_{NoN}$ is independent of $\tilde{p}$). Thus, for NoN to maximize her payoff, $\tilde{p}$ must assume the maximum value in the range that  $\tilde{p} \leq \tilde{p}_t^{mod}.$ Thus, $\tilde{p} = \tilde{p}_t^{mod}.$ Thus Theorem~\ref{theorem:NE_stage2_new_suff_ub}-1 follows.


Now, we prove Theorem~\ref{theorem:NE_stage2_new_suff_ub}-2.  Again, consider a SPNE in which $z^{eq} = 1$. By Theorem~\ref{theorem:p_tilde_new}, when $\Delta p\geq  t_N+\kappa_u \tilde{q}_{p}$, $z^{eq}=0$. Thus, 	$\Delta p<t_N+\kappa_u \tilde{q}_{p}$. By Theorem~\ref{theorem:NE_stage2_new_suff_ub}-1, $\tilde{p} = \tilde{p}_t^{mod}.$ If $\tilde{p}_t < \alpha$, $\tilde{p}_t^{mod} = \tilde{p}_t.$ Thus, $\pi_{NoN}(p_{NoN},\tilde{p}_{t}^{mod}) = \pi_{NoN}(p_{NoN},\tilde{p}_{t})=\pi_{NoN}(p_{NoN},\tilde{p})$. Thus, if $\pi_{NoN}(p_{NoN},\tilde{p}_{t}) \leq \pi_{NoN,z=0}(p_{NoN},\tilde{p})$, then 	$\pi_{NoN}(p_{NoN},\tilde{p}) \leq \pi_{NoN,z=0}(p_{NoN},\tilde{p}).$ Note that by Theorem~\ref{theorem:p_tilde_new}, NoN can ensure $z^{eq}=0$, by choosing $\tilde{p} > \tilde{p}_{t}$, since $\tilde{p}_{t} < \alpha.$ Thus, by tie-breaking assumption~\ref{assumption:ISP_p_tilde}, $z^{eq}=0$, which is a contradiction. Thus, $\pi_{NoN}(p_{NoN},\tilde{p}_{t}) > \pi_{NoN,z=0}(p_{NoN},\tilde{p}).$

Next, consider that $\pi_{NoN}(p_{NoN},\tilde{p}_{t}^{mod}) > \pi_{NoN,z=0}(p_{NoN},\tilde{p})$ and $\Delta p<t_N+\kappa_u \tilde{q}_{p}.$ By Theorem~\ref{theorem:NE_stage2_new_suff_ub}-1, if $z^{eq}=1$, $\pi_{NoN}(p_{NoN},\tilde{p}_{t}^{mod}) = \pi_{NoN}(p_{NoN},\tilde{p})$. By Theorem~\ref{theorem:p_tilde_new}, NoN can ensure $z^{eq}=1$, by choosing $\tilde{p} = \tilde{p}_{t}^{mod}.$ This fetches NoN a payoff, $\pi_{NoN}(p_{NoN},\tilde{p})$, which exceeds the payoff at $z=0$,  $\pi_{NoN,z=0}(p_{NoN},\tilde{p}).$
Thus, $z^{eq}=1.$ Theorem~\ref{theorem:NE_stage2_new_suff_ub}-2 follows.
\end{proof}

We now outline the derivation of the possible SPNEs. Theorem~\ref{theorem:neutralnotexists_q>}-1 holds, and therefore provides candidate (e) of Section~\ref{section:summaryof resutls} as the only possible SPNE when $z^{eq}=0.$ We now consider that  $z^{eq}=1.$ Thus,  $\Delta p <  t_N+\kappa_u \tilde{q}_{p}$, by Theorem~\ref{theorem:p_tilde_new}. We consider the following cases separately:  $\Delta p\leq \kappa_u \tilde{q}_{p}-t_{NoN}$, and $\kappa_u \tilde{q}_{p}-t_{NoN} <  \Delta p < t_N+\kappa_u \tilde{q}_{p}.$
When  $\Delta p\leq \kappa_u \tilde{q}_{p}-t_{NoN}$, following the arguments in Section~\ref{appendix:firstproof}, Case A-1, and subsequently Remark~\ref{r2}, SPNE candidate (a) of Section~\ref{section:summaryof resutls} becomes a possible SPNE. For $\kappa_u \tilde{q}_{p}-t_{NoN} <  \Delta p < t_N+\kappa_u \tilde{q}_{p}$, following Theorem~\ref{theorem:p_tilde_new}, we need to separately consider the cases of $t_N+t_{NoN} < \kappa_u \tilde{q}_f$ and  $t_N+t_{NoN} \geq \kappa_u \tilde{q}_f.$ In these cases, $\tilde{p} \in \{ \tilde{p}_{t, 2}, \tilde{p}_{t, 3}\}.$  Now, recall that  $\tilde{p}_{t,2}=\kappa_{ad} (n_{NoN}-\frac{\tilde{q}_f}{\tilde{q}_p})$, where $n_{NoN}=\frac{t_N+\kappa_u \tilde{q}_p-\Delta p}{t_N+t_{NoN}}$, and $ \tilde{p}_{t,3}= \kappa_{ad}n_{NoN}(1-\frac{\tilde{q}_f}{\tilde{q}_{p}})$, where $n_{NoN}=\frac{t_N+\kappa_u (\tilde{q}_{p}-\tilde{q}_f)-\Delta p}{t_N+t_{NoN}}$. If $\tilde{p_t} > \alpha$, the side-payment equals a constant $\alpha \tilde{q}_p$.  Thus, in each of these regions, we need to consider the cases that $\tilde{p}_t \leq \alpha, \tilde{p}_t > \alpha$, and $\tilde{p}_t = \tilde{p}_{t, 2}$, $\tilde{p}_t = \tilde{p}_{t, 3}$. Now, $\tilde{p}_{t, 2} \leq \alpha$,  $\tilde{p}_{t, 3} \leq \alpha$,  $\tilde{p}_{t, 2} > \alpha$,  $\tilde{p}_{t, 3} > \alpha$, lead to specific ranges of $\Delta p.$ The intersection of these ranges with those considered in Theorem~\ref{theorem:p_tilde_new}-2, \ref{theorem:p_tilde_new}-3 must be considered. The possible SPNEs in the interior of these regions can be determined by application of the first order conditions as in the proofs of Theorems~\ref{theorem:NE_stage1_new_q>}-2 (Section~\ref{appendix:firstproof}), \ref{theorem:NE_stage1_new_q<}-2, \ref{theorem:NE_stage1_new_q<}-3, \ref{theorem:NE_stage1_new_q<}-4 (Section~\ref{appendix:theorem:NE_stage1_new_q<}, Cases $B_1, B_2, C$ therein). The SPNEs on the boundaries of the regions considered in these proofs may be obtained or ruled out as in these proofs. SPNEs on other boundaries, i.e., $\tilde{p}_{t, 2} = \alpha$, $\tilde{p}_{t, 3} = \alpha$, may be obtained by considering the first-order condition on $p_N$,
\small
\be
\ba
 \frac{d \pi_N}{d p_{N}}=0& \Rightarrow t_{NoN}-\kappa_u \tilde{q}_{p} + p_{NoN}-2 p_N+c=0,
\ea
\ee
\normalsize
and $\Delta p = p_{NoN} - p_N$ given by either $\tilde{p}_{t, 2} = \alpha$ or $\tilde{p}_{t, 3} = \alpha$, as the case may be. 

Thus, overall, possible SPNEs include (a)-(e) and some other candidates identified as above. The number of such candidates will be $11$ overall.  Through detailed analysis, some of these candidates may be ruled out for all parameter values, and some others may be ruled out  for some parameter values. We defer this to future research.

\end{document}